\documentclass[11pt]{article}
\usepackage[margin=1in]{geometry}
\usepackage{constants}
\newconstantfamily{c}{symbol=c}
\usepackage{preamble}
\usepackage{braket}
\usepackage{wrapfig}
\title{Toward Instance-Optimal State Certification With Incoherent Measurements}

\DeclareMathAlphabet\mathbfcal{OMS}{cmsy}{b}{n}

\newcommand{\U}{\vec{U}}
\newcommand{\V}{\vec{V}}

\DeclareMathOperator{\Wg}{Wg}
\newcommand{\hatM}{\widehat{M}}
\newcommand{\diag}{\text{diag}}

\renewcommand{\d}{d}
\newcommand{\A}{\vec{A}}
\newcommand{\B}{\vec{B}}

\newcommand{\W}{\vec{W}}
\newcommand{\Sig}{\vec{\Sigma}}

\newcommand{\rhomm}{\rho_{\mathsf{mm}}}
\newcommand{\Eps}{\mathbfcal{E}}
\newcommand{\HS}{\text{HS}}
\newcommand{\op}{\text{op}}
\renewcommand{\epsilon}{\varepsilon}
\newcommand{\eps}{\epsilon}
\DeclareMathOperator{\polylog}{polylog}
\let\Re\relax
\DeclareMathOperator{\Re}{Re}
\DeclareMathOperator{\TV}{TV}

\usepackage{multirow}

\author{
Sitan Chen\thanks{This work was supported in part by NSF Award 2103300, NSF CAREER Award CCF-1453261, NSF Large CCF-1565235, and Ankur Moitra's ONR Young Investigator Award.}
\\
\texttt{sitanc@berkeley.edu}\\
UC Berkeley 
\and Jerry Li
\\
\texttt{jerrl@microsoft.com} \\
Microsoft Research
\and Ryan O'Donnell\thanks{Some of this work was done while the author was working at Microsoft Quantum.
    Supported by NSF grant FET-1909310 and ARO grant W911NF2110001.
    This material is based upon work
    supported by the National Science Foundation under grant numbers
    listed above. Any opinions, findings and conclusions or
    recommendations expressed in this material are those of the author
    and do not necessarily reflect the views of the National Science
    Foundation (NSF).}
\\
\texttt{odonnell@cs.cmu.edu}\\
Carnegie Mellon University
}

\begin{document}

\maketitle
\thispagestyle{empty}

\begin{abstract}

We revisit the basic problem of quantum state certification: given copies of unknown mixed state $\rho\in\co^{d\times d}$ and the description of a mixed state $\sigma$, decide whether $\sigma = \rho$ or $\norm{\sigma - \rho}_{\mathsf{tr}} \ge \epsilon$. When $\sigma$ is maximally mixed, this is \emph{mixedness testing}, and it is known that $\Omega(d^{\Theta(1)}/\epsilon^2)$ copies are necessary, where the exact exponent depends on the type of measurements the learner can make \cite{o2015quantum,bubeck2020entanglement}, and in many of these settings there is a matching upper bound \cite{o2015quantum,buadescu2019quantum,bubeck2020entanglement}.

Can one avoid this $d^{\Theta(1)}$ dependence for certain kinds of mixed states $\sigma$, e.g. ones which are approximately low rank? More ambitiously, does there exist a simple functional $f:\co^{d\times d}\to\R_{\ge 0}$ for which one can show that $\Theta(f(\sigma)/\epsilon^2)$ copies are necessary and sufficient for state certification with respect to \emph{any} $\sigma$? Such \emph{instance-optimal} bounds are known in the context of classical distribution testing, e.g. \cite{valiant2017automatic}.

Here we give the first bounds of this nature for the quantum setting, showing (up to log factors) that the copy complexity for state certification using nonadaptive incoherent measurements is essentially given by the copy complexity for mixedness testing times the fidelity between $\sigma$ and the maximally mixed state.
Surprisingly, our bound differs substantially from instance optimal bounds for the classical problem, demonstrating a qualitative difference between the two settings.
\end{abstract}

\clearpage

\setcounter{tocdepth}{2}
\tableofcontents
\addtocontents{toc}{\protect\thispagestyle{empty}}
\pagenumbering{gobble}

\clearpage
\pagenumbering{arabic} 


\section{Introduction}

We consider the problem of \emph{quantum state certification}.
We are given a description of a mixed state $\sigma \in \co^{d \times d}$ as well as $N$ copies of a state $\rho \in \co^{d \times d}$.
We are promised that either $\rho = \sigma$, or $\rho$ is $\eps$-far from $\sigma$ in trace norm, and our goal is to distinguish between these two cases with high probability.
From a practical perspective, the development of better methods for state certification is motivated by the need to efficiently verify the output of quantum devices. From a theoretical perspective, state certification is the natural quantum analogue of the well-studied classical problem of \emph{identity testing}: given a description of a probability distribution $p$ and samples from another distribution $q$, determine with high probability whether $q = p$ or $\norm{q - p}_1 \ge \epsilon$.

It is known that for general $\sigma$, $O(d / \eps^2)$ copies of $\rho$ suffice \cite{buadescu2019quantum}. Notably, this is smaller than the $\Theta(d^2 / \eps^2)$ copies needed to learn the state to $\eps$-accuracy in trace norm. Prior work of \cite{o2015quantum} also demonstrated that when $\sigma$ is the maximally mixed state, $\Omega(d/\epsilon^2)$ copies are necessary \cite{o2015quantum}. 
While these results settle the copy complexity of this problem for worst-case choices of $\sigma$, they leave a number of interesting questions unanswered:

\vspace{0.5em}
\noindent\textbf{Using Incoherent Measurements.} An important practical drawback of \cite{buadescu2019quantum} is that it makes a \emph{coherent measurement} across the product state $\rho^{\otimes N}$. While such measurements are very powerful, they require the learner to keep all $N$ copies of $\rho$ in quantum memory without any of them decohering. In practice, creating such a large amount of quantum memory, even for medium sized $d$, has proven to be a difficult task, limiting the near-term viability of coherent measurements. In contrast, algorithms that make \emph{incoherent measurements} only need to maintain one copy of $\rho$ at a time. 
Additionally, whereas the measurement in \cite{buadescu2019quantum} takes $\poly(d,N)$ time to prepare,
the protocol we present later in this paper can be implemented in $N \cdot \poly \log d$ time (see Remark~\ref{remark:efficient}).
Understanding whether one achieve statistical guarantees similar to that of \cite{buadescu2019quantum} using only incoherent measurements is thus a crucial step towards reliable near-term quantum computation.
    
Recent work of~\cite{bubeck2020entanglement} studied this question in the special case where $\sigma$ is the maximally mixed state\--- this special case of state certification is sometimes called \emph{mixedness testing}. They showed that the practical viability of incoherent measurements unfortunately comes at a statistical cost: in this setting $\Omega (d^{4/3} / \eps^2)$ copies are necessary, even if the incoherent measurements are chosen \emph{adaptively} as a function of the previous measurement outcomes. When they are chosen \emph{non-adaptively}, \cite{bubeck2020entanglement} further showed that $\Theta (d^{3/2} / \eps^2)$ copies are necessary \emph{and sufficient}.

It is not too hard to modify their upper bound to show that for \emph{general $\sigma$}, $O(d^{3/2}/\epsilon^2)$ copies still suffice for state certification. This settles the copy complexity of state certification with non-adaptive, incoherent measurements for \emph{worst-case} choices of $\sigma$.
    
\vspace{0.5em}
\noindent\textbf{Beyond Worst-Case $\sigma$.} This raises another important question: for which $\sigma$ can this $O(d^{3/2}/\epsilon^2)$ upper bound be improved? This bound is certainly not tight for all $\sigma$: for instance, if $\sigma$ is maximally mixed over a known subspace of dimension $r$, a simple argument demonstrates that $O(r^{3/2} / \eps^2)$ copies suffice. A natural hypothesis might be that some relaxed notion of rank of $\sigma$ dictates the true copy complexity of state certification with respect to $\sigma$.
    
This is inspired by a line of work in classical distribution testing on so-called \emph{instance-optimal} bounds for identity testing \cite{acharya2011competitive,acharya2012competitive,valiant2017automatic,diakonikolas2016new,blais2019distribution,jiao2018minimax}. The flagship result in this literature, due to \cite{valiant2017automatic}, states that for any distribution $p$ over $d$ elements, the optimal sample complexity $N$ of identity testing with respect to $p$ is essentially characterized by the \emph{$\ell_{2/3}$-quasinorm of $p$}. More formally, $N$ satisfies:
\begin{equation}
    \Omega(\Max{\epsilon^{-1}}{\epsilon^{-2}\|p^{-\max}_{-\eps / 16}\|_{2/3}}) \le N \le O( \Max{\epsilon^{-1}}{\eps^{-2}\|p^{-\max}_{-\eps}\|_{2/3}}) \label{eq:vv}
\end{equation}
for absolute constants $C_1,C_2 > 0$. Here $\| \cdot \|_{2/3}$ is the $\ell_{2/3}$-quasinorm, and $p^{-\max}_{-\delta}$ is the vector given by zeroing out the largest entry as well as the bottom $\delta$ mass from the probability vector for $p$. Note that when $p$ is uniform over $d$ elements, this recovers the well-known sample complexity bound of $\Theta (\sqrt{d} / \eps^2)$ for uniformity testing \cite{paninski2008coincidence}. 

Together, these two bounds give a striking and more or less tight characterization of the sample complexity landscape for identity testing: for \emph{any instance} of the problem, we know the optimal sample complexity up to constant factors! This begs the natural question:
\begin{center}
    \emph{Can we get a similarly tight characterization for the copy complexity of state certification with incoherent measurements?}
\end{center}

\subsection{Our Results}
In this work, we answer this in the affirmative by presenting an instance-optimal characterization of the copy complexity of state certification with non-adaptive incoherent measurements.
Surprisingly, our results demonstrate that the behavior of quantum state certification is qualitatively quite different from that of classical identity testing.
More formally, our main result is the following:
\begin{theorem}[Informal, see Theorems~\ref{thm:cert_lower_main} and \ref{thm:cert_upper_main}]
\label{thm:instance_informal}
	Given any mixed state $\sigma\in\co^{d\times d}$, there are mixed states $\overline{\sigma}$ and $\underline{\sigma}$ respectively given by projecting away some eigenvectors with eigenvalues summing to at most $\Theta(\epsilon^2)$ and $\Theta(\epsilon)$ and normalizing, such that the following holds.

	Let $\overline{d}_{\mathsf{eff}}$ (resp. $\underline{d}_{\mathsf{eff}}$) be the rank of $\overline{\sigma}$ (resp. $\underline{\sigma}$). The optimal copy complexity $N$ of state cert-ification with respect to $\sigma$ to trace distance $\epsilon$ using non-adaptive, incoherent measurements satisfies\footnote{Throughout, we use $\wt{\Omega}(\cdot)$ and $\wt{O}(\cdot)$ solely to suppress factors of $\log(d/\epsilon)$.} \begin{equation}
		\wt{\Omega}\left(\frac{d\cdot \underline{d}^{1/2}_{\mathsf{eff}}}{\epsilon^2}\cdot F(\underline{\sigma},\rhomm)\right) \le N \le \wt{O}\left(\frac{d\cdot \overline{d}^{1/2}_{\mathsf{eff}}}{\epsilon^2}\cdot F(\overline{\sigma},\rhomm)\right),
	\end{equation}
	where $\rhomm$ is the maximally mixed state $\frac{1}{d}\Id$ and $F$ denotes the fidelity between two quantum states.
\end{theorem}

\noindent Note that when $\sigma$ is maximally mixed and $0 < \epsilon < 1$ is bounded away from 1, then $\overline{\sigma}$ and $\underline{\sigma}$ are projectors to subspaces of dimension $\Omega(d)$, so $\underline{d}_{\mathsf{eff}}, \overline{d}_{\mathsf{eff}} = \Theta(d)$ and $F(\underline{\sigma},\rhomm) = \Theta(1)$, recovering\footnote{As our techniques are a strict generalization of those of \cite{bubeck2020entanglement}, in this special case where $\sigma$ is the maximally mixed state, our analysis does not actually lose log factors.} the $\Theta(d^{3/2}/\epsilon^2)$ bound of \cite{bubeck2020entanglement} for mixedness testing with non-adaptive, incoherent measurements.

Qualitatively, our result says that unless $\sigma$ puts $1 - \poly (\eps)$ mass on $o(d)$ dimensions, the copy complexity of state certification is equal to the worst-case copy complexity of state certification, times the fidelity between $\sigma$ and the maximally mixed state.
Surprisingly, unlike in the classical case, our bound demonstrates that there is no clean dimension-independent functional which controls the complexity of quantum state certification.
Rather, there is some inherent ``curse of dimensionality'' for this problem.
Also note that in the quantum case, unlike in the classical case, we do not remove the largest element from the spectrum of $\sigma$.

\begin{example}
\label{ex:quantum-diff}
To elaborate on this curse of dimensionality, consider the following example.
Let $\sigma \in \mathbb{C}^{(d + 1) \times (d + 1)}$ be the mixed state given by $\sigma = \diag (1 - 1/d^2, 1/d^3, \ldots, 1/d^3)$.
The classical analogue of certifying this state is identity testing to the distribution $p$ over $d + 1$ elements which has one element with probability $1 - 1/d^2$, and $d$ elements with probability $1/d^3$.

For the classical case, the bound from~\cite{valiant2017automatic} demonstrates that the sample complexity of identity testing to $p$ is $\Theta \left( \frac{1}{d^{3/2} \eps^2} \right)$ for sufficiently small $\epsilon$.
In particular, in this regime the sample complexity actually is \emph{decreasing} in $d$.
This phenomena is not too surprising---this distribution is  very close to being a point distribution, and the only ``interesting'' part of it, namely, the tail, only has total mass $1/d$, which vanishes as we increase $d$.

In contrast, Theorem~\ref{thm:instance_informal} shows that the copy complexity of the quantum version of this problem using incoherent measurements is $\widetilde{\Theta} (d^{1/2} / \eps^2)$.
Notably, this is \emph{increasing} in $d$!
At a high level (see Section~\ref{sec:overview} for further discussion), it is because the unknown state $\rho$ may share the same diagonal entries with $\sigma$ but may not commute with it, so the ``interesting'' behavior need not be constrained to the subspace given by the small eigenvalues of $\sigma$.
In particular, $\rho$ might be far from $\sigma$ only because $\rho$ contains nontrivial mass in its off-diagonal entries.
This allows us many more degrees of freedom in constructing the lower bound instance, resulting in a much stronger bound.
\end{example}

\noindent It turns out this curse of dimensionality persists even for \emph{adaptive}, incoherent measurements. Formally, we show the following lower bound which is qualitatively similar to that of Theorem~\ref{thm:instance_informal}:

\begin{theorem}[Informal, see Theorem~\ref{thm:cert_adaptive_main}]\label{thm:instance_adaptive_informal}
	In the notation of Theorem~\ref{thm:instance_informal}, \begin{equation}
		N \ge \wt{\Omega}\left(\frac{d\cdot \underline{d}^{1/3}_{\mathsf{eff}}}{\epsilon^2}\cdot F(\underline{\sigma},\rhomm)\right) \label{eq:adaptive_lb}
	\end{equation}
	copies are needed for state certification w.r.t. $\sigma$ to error $\epsilon$ using adaptive, incoherent measurements.
\end{theorem}
\noindent
When $\sigma = \rhomm$, we recover the best known adaptive lower bound for mixedness testing \cite{bubeck2020entanglement}. 
Furthermore, since non-adaptive measurements are a subset of adaptive ones, the upper bound in Theorem~\ref{thm:instance_informal} also provides a per-instance upper bound for this problem which matches \eqref{eq:adaptive_lb} up to the factor of $d^{1/2}$ versus $d^{1/3}$.
Obtaining tight bounds in this setting is an interesting open question; however, we note that this is not known even for mixedness testing.


\subsection{Related Work}

A full survey of the literature on quantum (and classical) testing is beyond the scope of this paper; we only discuss the most relevant works below. We also note there is a vast literature on related quantum learning problems such as state tomography, see e.g.~\cite{kueng2017low,gross2010quantum,flammia2012quantum,voroninski2013quantum,haah2017sample,o2016efficient,o2017efficient} and references therein.

\vspace{0.5em}
\noindent\textbf{Quantum Property Testing.} The problem of quantum state certification lies within the broader field of quantum state property testing.
See~\cite{montanaro2016survey} for a more complete survey.
Within this field, there are two regimes studied.
In the \emph{asymptotic} regime, the goal is to precisely characterize the rate at which the error converges as $n \to \infty$, and $d$ and $\eps$ are fixed.
Here quantum state certification is more commonly known as \emph{quantum state discrimination}~\cite{chefles2000quantum,barnett2009quantum,audenaert2008asymptotic}.
For a more complete survey of work on this problem, see~\cite{bae2015quantum}.
However, this line of work does not attempt to characterize the statistical dependence on the dimension.

In contrast, we consider the \emph{non-asymptotic} regime, where the goal is to characterize the rate of convergence for quantum state certification as a function of $d$ and $\eps$.
As discussed above, recent work of~\cite{o2015quantum,buadescu2019quantum} has demonstrated that $\Theta (d / \eps^2)$ copies are necessary and sufficient for quantum state certification over the worst choice of $\sigma$, when the measurements are allowed to be arbitrary.
However, the representation theoretic tools used within seem to be quite brittle and do not easily extend to give instance-optimal rates.
Understanding the instance-optimal rate for quantum state certification using arbitrary measurements is a very interesting open question.

\vspace{0.5em}
\noindent\textbf{Incoherent Measurements.} A number of recent papers on quantum learning and testing have also considered the power of incoherent measurements, and more generally, other types of restricted measurements for quantum property testing tasks apart from state certification.
Following the aforementioned~\cite{bubeck2020entanglement}, subsequent work of~\cite{aharonov2021quantum} defined a more general notion of quantum algorithmic measurement which includes incoherent measurements and proved some incomparable lower bounds for other problems such as purity testing and channel discrimination under this model. The recent work of \cite{huang2021information} also showed a separation for shadow tomography with Pauli observables using incoherent versus 2-entangled measurements. Lastly, another very recent work \cite{chen2021exponential} refined these two works by showing nearly optimal separations for shadow tomography, purity testing, and channel discrimination using incoherent versus coherent measurements.

We also note that a number of papers in quantum tomography have considered the power of incoherent measurements, see e.g.~\cite{kueng2017low,gross2010quantum,flammia2011direct,voroninski2013quantum,haah2017sample}.
Another line of work considers the complexity of testing using only Pauli measurements~\cite{flammia2011direct,flammia2012quantum,da2011practical,aolita2015reliable}.
However, because of the restrictive setting, these latter bounds are typically weaker, and these papers also do not obtain instance-optimal bounds for this setting. 

\vspace{0.5em}
\noindent\textbf{Classical Distribution Testing.} State certification is the quantum version of the well-studied classical problem of distribution identity testing.
A complete survey of this field is also beyond the scope of this paper.
See~\cite{canonne2020survey,goldreich2017introduction} and references within for a more detailed discussion. 
Of particular interest to us is the line of work on instance-optimal testing, the direct classical analog of the problem we consider in this paper.
The works of~\cite{acharya2011competitive,acharya2012competitive} consider sample complexity bounds which improve upon the worst case sample complexity for different choices of probability distributions.
The setting that we consider is most directly inspired by the aforementioned work of~\cite{valiant2017automatic}.
Subsequent work has re-proven and/or derived new instance-optimal bounds for identity testing and other problems as well, see e.g.~\cite{diakonikolas2016new,blais2019distribution,jiao2018minimax}.

\section{Overview of Techniques}
\label{sec:overview}

As with many other property testing lower bounds, ours is based on showing hardness for distinguishing between a simple ``null hypothesis'' and a ``mixture of alternatives,'' i.e. whether the unknown state $\rho$ that we get copies of is equal to $\sigma$ or was randomly sampled at the outset from some distribution over states $\epsilon$-far from $\sigma$. Throughout, we will assume that $\sigma$ is a diagonal matrix.
This is without loss of generality since we are given a description of $\sigma$ and can change basis.

When $\sigma = \frac{1}{d}\Id$, the standard choice for the mixture (and the one that leads to optimal lower bounds in this case) is the distribution over mixed states of the form $\frac{1}{d}\left(\Id + \U^{\dagger}\diag(\epsilon,\ldots,-\epsilon,\ldots)\U\right)$ where $\U$ is sampled from the Haar measure over $d\times d$ unitary matrices, and previous works have shown lower bounds for mixedness testing with entangled measurements \cite{o2015quantum} and incoherent measurements \cite{bubeck2020entanglement} by analyzing this particular distinguishing task. Indeed, our proof builds upon the general framework introduced in the latter work (see Section~\ref{sec:framework} for an exposition of the main ingredients from \cite{bubeck2020entanglement}) but differs in crucial ways.

To get a sense for what the right distinguishing task(s) to consider for general $\sigma$ are, it is instructive to see first how to prove instance-optimal bounds for classical distribution testing.

\subsection{Instance-Optimal Lower Bounds for Identity Testing}

Here we sketch how to prove the lower bound of \cite{valiant2017automatic} for identity testing (up to log factors). Recall this is the setting where one gets access to independent samples from an unknown distribution $p$ over $d$ elements and would like to test whether $p = q$ or $\norm{p - q}_1 > \epsilon$ for a known distribution $q$.

When $q$ is the uniform distribution over $d$ elements, a classical result of~\cite{paninski2008coincidence} demonstrates that the fundamental bottleneck is distinguishing whether the samples come from $p$, or if the samples come from a version of $q$ where each entry from its vector of probabilities has been perturbed by $\pm \eps / d$. In this setting, the mixture of alternatives consists of all distributions $q^{\zeta}$ that could have been obtained in this fashion, where the index $\zeta$ indicates the sign pattern of the perturbation chosen.

The main conceptual challenge to extending this lower bound strategy to more general $q$ is that the entries of the probability vector for $q$ may take values across many different scales, and whatever lower bound instance one designs must be sensitive to these scales.


One approach to account for these different scales is to ``bucket'' the probability vector for $q$, where each given bucket contains all entries within a fixed multiplicative factor of one another.
It turns out that Paninski's analysis works even if $q$ is not exactly uniform as long as its probabilities are within a multiplicative factor of each other. For this reason, within each bucket we could simply apply Paninski's construction and randomly perturb the probabilities by a carefully chosen multiple of $\pm \epsilon/d$. Combining these constructions across buckets after appropriately scaling them thus gives a natural mixture of alternatives $\brc{q^{\zeta}}$ to distinguish from the true distribution $q$, where again, $\zeta$ denotes the sign pattern of the perturbations chosen.

The main technical challenge then is to upper bound $\tvd(q^{\otimes N}, \E[\zeta]{(q^{\zeta})^{\otimes N}})$, that is, the total variation distance between the distribution over $N$ i.i.d. draws from $q$ and the distribution over $N$ i.i.d. draws from $q^{\zeta}$ where $\zeta$ was sampled uniformly at random from the set of all possible sign patterns corresponding to perturbations of $q$.

A common analytical trick for carrying out this bound--- and the approach that \cite{valiant2017automatic} take--- is to first Poissonize, that is, take $N$ to be a Poisson random variable. Unfortunately, Poissonization does not seem to have any straightforward analogue in the quantum setting, where the choice of measurement can vary across copies, so we eschew this technique in favor of an alternative approach that we sketch next.

\paragraph{Ingster-Suslina Method and Moment Bounds.}
Apart from Poissonization, another way to bound $\tvd(q^{\otimes N}, \E[\zeta]{(q^{\zeta})^{\otimes N}})$ is to pass to chi-squared divergence and invoke the Ingster-Suslina method (see e.g. Section 3.3 of \cite{ingster2012nonparametric}, or Lemma 22.1 and its application in Section 24.3 in \cite{wu2017lecture}). At a high level, this approach amounts to bounding higher-order moments of the \emph{pairwise correlation} \begin{equation}
    \phi^{\zeta,\zeta'} \triangleq \E[i]*{(\Delta_{\zeta}(i) - 1)(\Delta_{\zeta'}(i)-1)}
\end{equation} as a random variable in $\zeta,\zeta'$. Here, the expectation is over sample $i\in[d]$ drawn from $q$, and \begin{equation}
    \Delta_{\zeta}(i) = q^{\zeta}_i / q_i
\end{equation} is the \emph{likelihood ratio} between the probability of drawing $i$ when $p = q^{\zeta}$ versus the probability of drawing $i$ when $p = q$. Concretely, if one can show that \begin{equation}
    \E[\zeta,\zeta']*{\left(1 + \phi^{\zeta,\zeta'}\right)^t} = 1 + o(1) \label{eq:ingster_overview}
\end{equation} for some $t$, this would imply a sample complexity lower bound of $t$ for testing identity to $q$.



It turns out to be possible to give sufficiently good upper bounds on the moments of $\phi^{\zeta,\zeta'}$ (after some appropriate preprocessing on $q$ as done in \cite{valiant2017automatic}) that one can recover the same bound as~\cite{valiant2017automatic} up to poly-logarithmic factors in $d / \epsilon$.
It is this approach that we will generalize to the quantum setting.

\subsection{Passing to the Quantum Setting}
\label{subsec:pass2quantum}

We now describe how to extend some of these ideas to quantum state certification.

\paragraph{Scale-Sensitive Rotations.}
Recall from the discussion at the beginning of this section that in the case where $\sigma = \rhomm$, the right ``mixture of alternatives'' to consider is to perturb every eigenvalue of $\rhomm$ and then randomly rotate by a Haar-random unitary over $\co^d$; this is sometimes called the \emph{quantum Paninski} instance \cite{o2015quantum} for its resemblance to Paninski's construction in the classical setting.

For general $\sigma$, we could try the same thing, but motivated by the classical setting, we would tune how much we perturb each eigenvalue based on its magnitude. Unfortunately, if we then simply rotate the resulting perturbed state by a Haar-random unitary over $\co^d$, it turns out that we can't hope to prove a sufficiently strong lower bound.

To see this, let's consider the following extreme example. Imagine that $\sigma$ is nearly a pure state. A random global rotation of a perturbation of $\sigma$, no matter how cleverly we picked the perturbation, is close to a Haar-random pure state. So its trace inner product with $\sigma$ will be on the order of $\Theta(1/d)$ with high probability, whereas the trace inner product of $\sigma$ with itself is on the order of $\Theta(1)$. In particular, just by measuring the observable given by $\sigma$, we can easily distinguish whether $\rho = \sigma$ or $\rho$ comes from this particular mixture of alternatives using $O(1)$ measurements.

The point is that in the quantum setting, we need to be sensitive to the different scales of $\sigma$'s eigenvalues not only in picking the perturbations to the eigenvalues of $\sigma$, but \emph{also in picking the ensemble of rotations}!

\paragraph{An Attempt: Generalized Quantum Paninski.}
We now outline an attempt at generalizing the quantum Paninski construction in a way that is sufficiently sensitive to the different scales for the eigenvalues of $\sigma$. Motivated by the classical construction described above, we can group the eigenvalues of $\sigma$ into \emph{buckets}, where a given bucket contains all eigenvalues within a fixed multiplicative factor of each other, and consider a mixture of alternatives defined as follows. First, given any $m\in \mathbb{N}$, define the matrix:
\begin{equation}
    \vec{Z}_m \triangleq \begin{cases}
        \diag(1,\ldots,-1,\ldots) & $m$ \ \text{even} \\
        \diag(0,1,\ldots,-1,\ldots) & $m$ \ \text{odd},
    \end{cases}
\end{equation}
where $\vec{Z}_m$ consists of $\floor{m/2}$ $1$'s and $\floor{m/2}$ $-1$'s. The mixture of alternatives is given by the distribution over mixed states of the form $\sigma + \U^{\dagger}\vec{\calE}\U$, where now $\U$ is a \emph{block-diagonal unitary matrix whose blocks are Haar-random} and whose block structure corresponds to the buckets, and $\vec{\calE}$ is a direct sum of scalings of $\vec{Z}_m$, where the different $m$'s and scalings correspond to the sizes and relative magnitudes of the buckets. 

For instance, if $\sigma = \left(\frac{1}{2\sqrt{d}} \Id_{\sqrt{d}}\right)\oplus\left(\frac{1}{2(d - \sqrt{d})}\Id_{d - \sqrt{d}}\right)$, we can take $\U$ to be distributed as $\U_1\oplus\U_2$, where $\U_1\in U(\sqrt{d})$ and $\U_2\in U(d - \sqrt{d})$ are Haar-random, and $\vec{\calE} = \left(\frac{\epsilon_1}{2\sqrt{d}} \vec{Z}_{\sqrt{d}}\right)\oplus\left(\frac{\epsilon_2}{2(d - \sqrt{d})}\vec{Z}_{d - \sqrt{d}}\right)$ for appropriately chosen $\epsilon_1,\epsilon_2$ summing to 2.


Our analysis for this instance follows the Ingster-Suslina method in the nonadaptive case and the general framework of \cite{bubeck2020entanglement} in the adaptive case (see Section~\ref{sec:framework} for an exposition of these two frameworks), and the central object for both proofs is the pairwise correlation \begin{equation}
    \phi^{\U,\V} \triangleq \E[z]*{(\Delta_{\U}(z) - 1)(\Delta_{\V}(z) - 1)}.
\end{equation} Analogously to the classical setup described above, here the expectation is over outcomes $z$ if one makes some quantum measurement on a single copy of the state $\rho = \sigma$, and $\Delta_{\U}(z)$ is the likelihood ratio between the probability of observing outcome $z$ when $\rho = \sigma + \U^{\dagger}\vec{\calE}\U$ versus the probability of observing the same outcome when $\rho = \sigma$ under a particular POVM (see Section~\ref{sec:framework} for formal definitions). And as in the classical setup, it turns out that we need to show that \begin{equation}
    \E[\U,\V]*{\left(1 + \phi^{\U,\V}\right)^t} = 1 + o(1)
\end{equation} for sufficiently large $t$, so the primary challenge is to control the moments of $\phi^{\U,\V}$ (regarded as a random variable in $\U,\V$), or equivalently to show that it concentrates sufficiently around its mean.

If $\U,\V$ were Haar-random unitary matrices, one could do this by invoking standard concentration of measure for Haar-random unitary matrices \cite{anderson2010introduction,meckes2013spectral}. Indeed, this is the approach of \cite{bubeck2020entanglement}, but for general $\sigma$ we need to control the tails of $\phi^{\U,\V}$ when $\U,\V$ have the abovementioned block structure, for which off-the-shelf tail bounds will not suffice. Instead, we argue that because we can assume without loss of generality that the optimal measurements to use to distinguish $\rho = \sigma$ from $\rho = \sigma + \U^{\dagger}\vec{\calE}\U$ must respect the block structure, $\phi^{\U,\V}$ is a weighted sum of pairwise correlations $\phi^{\U,\V}_j$ for many independent sub-problems, one for each ``bucket'' $j$ (see \eqref{eq:phij_def}). These are independent random variables, each parametrized by an independent Haar-random unitary matrix in a lower-dimensional space, so we can show a tail bound for $\phi^{\U,\V}$ by combining the tail bounds for $\brc{\phi^{\U,\V}_j}$ (see Section~\ref{subsubsec:general_nonadaptive}).

In Section~\ref{subsubsec:tune_nonadaptive}, we show how to optimally tune the entries of $\vec{\calE}$. Here however, we finally arrive at the surprising juncture where instance-optimal state certification deviates significantly from its classical analogue:
\begin{center}
    \emph{This generalized quantum Paninski construction does not always yield the right lower bound!}
\end{center}
It turns out that even with the optimal tuning of $\vec{\calE}$, the approach outlined thus far only achieves a copy complexity lower bound of roughly $\wt{\Omega}(\norm{\sigma'}_{2/5}/\epsilon^2)$ (see Lemma~\ref{lem:nonadaptive_1}), where $\sigma'$ is obtained from $\sigma$ by projecting out its largest eigenvalue and some small eigenvalues.

While this recovers the lower bound of \cite{bubeck2020entanglement} when $\sigma = \rhomm$, in other situations one can readily see that $\wt{\Omega}(\norm{\sigma'}_{2/5}/\epsilon^2)$ can be much worse than the lower bound in Theorem~\ref{thm:instance_informal}. Consider $\sigma$ given by Example~\ref{ex:quantum-diff}. For that choice of $\sigma = \diag(1 - 1/d^2,1/d^3,\ldots,1/d^3)$, $\norm{\sigma'}_{2/5} = 1/\sqrt{d}$, so we only get a lower bound of $\Omega\left(\frac{1}{d^{1/2}\epsilon^2}\right)$. In contrast, as discussed in Example~\ref{ex:quantum-diff}, the right copy complexity for this problem turns out to be $\wt{\Theta}(\sqrt{d}/\epsilon^2)$. We now describe a second lower bound instance that, combined with the generalized quantum Paninski construction, yields an instance near-optimal lower bound.

\paragraph{Missing Ingredient: Perturbing the Off-Diagonals.}
For simplicity, consider a mixed state $\sigma$ with exactly two buckets, e.g. $\sigma = (\lambda_1 \Id_{d_1})\oplus(\lambda_2\Id_{d_2})$ where $d_1 \ge d_2$. In this case, one can regard the generalized Paninski instance as a family of perturbations of the two principal submatrices indexed by the coordinates $\brc{1,\ldots d_1}$ in bucket 1 and the coordinates $\brc{d_1+1,\ldots,d}$ in bucket 2 respectively. But one could also perturb $\sigma$ along the \emph{off-diagonal blocks}, rather than on the principal blocks, by considering matrices of the form \begin{equation}
    \sigma + \begin{pmatrix}
        \vec{0}_{d_1} & (\epsilon/2d_2)\cdot \W \\
        (\epsilon/2d_2)\cdot \W^{\dagger} & \vec{0}_{d_2} \label{eq:offdiag}
    \end{pmatrix}
\end{equation} parametrized by Haar-random $\W\in\co^{d_1\times d_2}$ consisting of orthonormal columns. One can show that as long as $\epsilon \le d_{j_1} \cdot \sqrt{\lambda_1\lambda_2}$, then \eqref{eq:offdiag} is a valid density matrix (Lemma~\ref{lem:rhoW_facts}) and is $\epsilon$-far in trace distance from $\sigma$. In this regime, we show a lower bound of $\Omega(d_1\sqrt{d_2}/\epsilon^2)$ for distinguishing whether $\rho = \sigma$ or whether $\rho$ is given by a matrix \eqref{eq:offdiag} where $\W$ is sampled Haar-randomly at the outset.

For general $\sigma$, by carefully choosing which pair of buckets to apply this construction to, we obtain the lower bound of Theorem~\ref{thm:instance_informal} for very small $\epsilon$. For larger $\epsilon$ we show that if the lower bound from the generalized Paninski instance were inferior to that of Theorem~\ref{thm:instance_informal}, then this would contradict the assumption that $\epsilon$ is large (see Section~\ref{subsec:conclude_nonadaptive}). Altogether, this completes the proof of the claimed lower bound in Theorem~\ref{thm:instance_informal}, modulo one last corner case that we now discuss.

\paragraph{Handling the Largest Eigenvalue.}
Indeed, there is one more feature of Theorems~\ref{thm:instance_informal} and \ref{thm:instance_adaptive_informal} which is unique to the quantum setting. In the classical setting, the instance-optimal sample complexity of testing identity to a given distribution $p$ is essentially given by $\Max{\frac{1}{\epsilon}}{\frac{\norm{p'}_{2/3}}{\epsilon^2}}$, where $p'$ is derived from $p$ by zeroing out not just the bottom $O(\epsilon)$ mass from $p$ but also the \emph{largest} entry of $p$. 
To see why the latter, as well as the additional $\frac{1}{\epsilon}$ term, is necessary, consider a discrete distribution $p$ which places $1 - \epsilon/100$ mass on some distinguished element of the domain, call it $x^*$. The $\Max{\frac{1}{\epsilon}}{\frac{\norm{p'}_{2/3}}{\epsilon^2}}$ lower bound would yield $\Omega(1/\epsilon)$ sample complexity, and an algorithm matching this bound would simply be to estimate the mass the unknown distribution places on $x^*$. The reason is that because $p$ places total mass $\epsilon/100$ on elements distinct from $x^*$, any distribution $\epsilon$-far from $p$ in $\ell_1$-distance must place at most $1 - \epsilon$ mass on $x^*$, which can be detected in $O(1/\epsilon)$ samples.

In stark contrast, in the quantum setting if $\sigma$ had an eigenvalue of $1 - \epsilon/100$, then the copy complexity of state certification with respect to $\sigma$ scales with $1/\epsilon^2$. The reason is that there is ``room in the off-diagonal entries'' for a state $\rho$ to be $\epsilon$-far from $\sigma$. Indeed, we can formalize this by considering a lower bound instance similar to \eqref{eq:offdiag}. In fact it is even simpler, because for mixed states whose largest eigenvalue is particularly large, it suffices to randomly perturb a single pair of off-diagonal entries! To analyze the resulting distinguishing task, we eschew the framework of \cite{bubeck2020entanglement} and directly bound the likelihood ratio between observing any given sequence of measurement outcomes under the alternative hypothesis versus under the null hypothesis (see Section~\ref{subsec:case3} and Lemma~\ref{lem:corner} in particular).

\paragraph{Adaptive Lower Bounds.} As we discussed following Theorem~\ref{thm:instance_adaptive_informal}, the ideas above can also be implemented in the setting where one can choose incoherent measurements \emph{adaptively} (see Theorem~\ref{thm:instance_adaptive_informal}). The reason the lower bound we obtain is not instance-optimal is the same technical reason that \cite{bubeck2020entanglement} was not able to obtain an optimal lower bound in the special case of mixedness testing, namely that there is some lossy balancing step to handle a certain low-probability event (see the proof of Theorem~\ref{thm:bcl} in Appendix~\ref{app:bcl_defer}).

\subsection{Upper Bound}

As in our lower bound proof, we will partition the spectrum of $\sigma$ into buckets. We will also place all especially small eigenvalues of $\sigma$ in a single bucket of their own\--- this latter bucket will contain the smallest eigenvalues of $\sigma$ that together sum to $O(\epsilon^2)$. For the purposes of discussion in this section, we will call this the ``negligible bucket'' and we will call all others ``non-negligible buckets.''

For starters, in Section~\ref{subsec:generic_cert} we give a simple algorithm ({\sc BasicCertify}, see Algorithm~\ref{alg:basic}) for state certification which is already optimal up to constant factors when the eigenvalues of $\sigma$ all fall within the same bucket. Similar to the mixedness tester in \cite{bubeck2020entanglement}, this algorithm is based on measuring our copies of unknown state $\rho$ in a Haar-random basis and running a classical identity tester \cite{diakonikolas2016new}. As the analysis is very similar to that of \cite{bubeck2020entanglement}, we defer the details to Section~\ref{subsec:generic_cert}.

Now consider a general mixed state $\sigma$ given by an arbitrary diagonal density matrix. Suppose its diagonal entries fall into $m$ buckets in total; by virtue of the bucketing scheme, $m$ is guaranteed to be at most logarithmic in $d/\epsilon$ (see Fact~\ref{fact:fewbuckets}). At a high level, if the state $\rho$ that we get copies of is $\epsilon$-far in trace distance from $\sigma$, then by an averaging argument, there should be some pair of buckets such that the corresponding block submatrix of $\sigma$ is somewhat far from $\rho$ in trace distance. Indeed, one of four things could happen (see Figure~\ref{fig:blocks}):
\begin{enumerate}[leftmargin=*]
    \item[(A)] There may be a non-negligible bucket for which the corresponding principal submatrix of $\sigma$ is $\Omega(\epsilon/m^2)$-far from that of $\rho$, in which case we can detect that $\rho$ is far from $\sigma$ simply by running {\sc BasicCertify} restricted to that bucket (see Lemma~\ref{lem:scenario3}).
    \item[(B)] There may be two non-negligible buckets for which the corresponding pair of off-diagonal blocks in $\sigma$ are $\Omega(\epsilon/m^2)$-far from the corresponding submatrix in $\rho$, in which case we can detect that $\rho$ is far from $\sigma$ by running {\sc BasicCertify} restricted to these two buckets (see Lemma~\ref{lem:scenario4}).
    \item[(C)] For the negligible bucket, the corresponding principal submatrix of $\sigma$ is $\Omega(\epsilon^2)$-far from that of $\rho$, in which case we can measure the observable given by the projector to that submatrix. In this case, $O(1/\epsilon^2)$ copies suffice (see Lemma~\ref{lem:scenario1}).
    \item[(D)] None of the above three cases hold, and $\rho$ and $\sigma$ differ primarily in the off-diagonal block with rows indexed by the negligible bucket and columns indexed by all non-negligible buckets. But by basic linear algebra (Lemma~\ref{lem:tracepsd}) and the fact that the eigenvalues in the negligible bucket sum to $\epsilon^2$, this would contradict the fact that we are not in case (C) (see Lemma~\ref{lem:scenario2})!
\end{enumerate}

\begin{wrapfigure}{R}{0cm}\\
\centering
\includegraphics[width=0.4\textwidth]{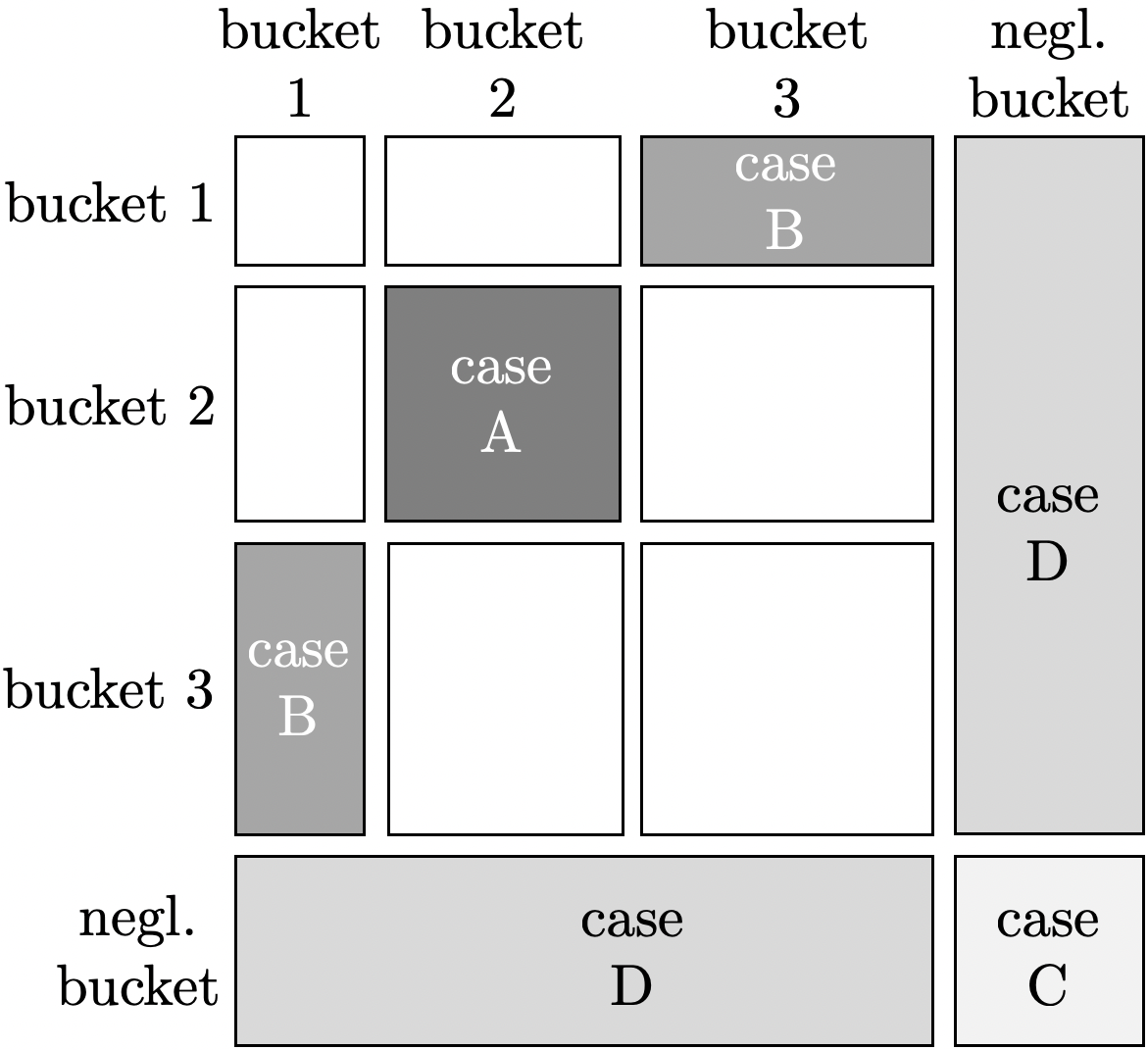}
    \caption{Partition of $\sigma$ into blocks corresponding to buckets, relevant submatrix for each case highlighted in gray.}
    \label{fig:blocks}
\end{wrapfigure}

We remark that the idea of reducing from state certification to mixedness testing by performing a case analysis on buckets of the spectrum is reminiscent of the instance near-optimal algorithm of \cite{diakonikolas2016new} for classical identity testing. That said, as is clear in the above proof sketch, the off-diagonal entries of $\rho$ pose a number of technical hurdles not present in the classical setting, just as they did in the proof of the lower bound.

\paragraph{Why Do the Upper and Lower Bounds ``Line Up''?} The casework above gives a good sense for why our upper and lower bounds happen to ``line up'' up to log factors. Ignoring the negligible bucket for the time being, recall that the averaging argument in our upper bound proof essentially implies that any $\rho$ which is far from $\sigma$ must be relatively far either 1) within a principle submatrix corresponding to a single bucket, or 2) within an off-diagonal submatrix corresponding to a pair of buckets.

This upper bound strategy complements our lower bound constructions nicely. Indeed, if we ignore the contribution from all other entries apart from the submatrix in question, then we can ask: what mixture of alternatives is hardest to tell apart from $\sigma$ if the alternatives all differ from $\sigma$ only in that submatrix? Depending on whether that submatrix is principle or off-diagonal, our generalized quantum Paninski and off-diagonal lower bound constructions provide essentially the optimal answer to this question.


\paragraph{Why Truncation?} The reader might be wondering why we need to truncate some of the eigenvalues of $\sigma$ in our bounds in Theorems~\ref{thm:instance_informal} and \ref{thm:instance_adaptive_informal}. For instance, how are we able to prove an upper bound which only depends on $\sigma$ after we have thrown out $O(\epsilon^2)$ of its eigenmass, rather than on $\sigma$ itself? As the above description of our algorithm makes clear, the reason is that the copy complexity of state certification with respect to $\sigma$ is really dominated by cases (A) and (B), and in these cases the complexity of running {\sc BasicCertify} only depends on the non-negligible buckets, i.e. the buckets containing the eigenvalues corresponding to the truncation $\overline{\sigma}$.

That said, there is a gap between the amount of mass we need to truncate in the definition of $\overline{\sigma}$ in the upper bound versus $\underline{\sigma}$ in the lower bound ($\Theta(\epsilon^2)$ versus $\Theta(\epsilon)$) in our theorems. The latter level of truncation appears to be an artifact of our techniques, and we conjecture that the lower bound can be upgraded to hold even if $\underline{\sigma}$ is defined by removing only $\Theta(\epsilon^2)$ mass from $\sigma$.

\paragraph{Why Fidelity?} Finally, we give some intuition for why fidelity with respect to the maximally mixed state arises in our copy complexity bounds. To do so, we will go into slightly more detail about the analysis of the algorithm we sketched above, focusing on cases (A) and (B).

First consider case (A). Suppose for simplicity that $\rho$ was identical to $\sigma$ except in the principal submatrix corresponding to the diagonal entries of $\sigma$ in the interval $[2^{-j-1},2^{-j}]$. Denote the number of rows/columns of this submatrix by $d_j$. As we alluded to above, it turns out that mixedness testing to error $\epsilon'$ for $d_j$-dimensional mixed states whose eigenvalues are all in the same bucket has copy complexity $\Theta(d_j^{3/2}/\epsilon'^2)$. On the other hand, because the trace of this submatrix is $\Theta(d_j 2^{-j})$, we would need to make $\Theta(2^j/d_j)$ measurements of $\rho$ in expectation to simulate one measurement of the conditional state given by $\rho$ restricted to this submatrix. But for the same reason, the trace distance $\epsilon'$ between the \emph{normalized} states given by this principal submatrix of $\rho$ and $\sigma$ is also $2^j/d^j$ times bigger than $\epsilon/m^2$. As $m$ is logarithmic in $d/\epsilon$, this means that $\wt{O}(d_j^{5/2} 2^{-j} /\epsilon^2)$ copies suffice to detect that $\rho$ differs noticeably from $\sigma$ in this submatrix (see Lemma~\ref{lem:scenario3}).

Now consider case (B). Suppose for simplicity that $\rho$ was identical to $\sigma$ except in the $d_j\times d_{j'}$ and $d_{j'}\times d_j$ off-diagonal blocks corresponding to two buckets of eigenvalues, namely those in $[2^{-j-1},2^{-j}]$ and those in $[2^{-j'-1},2^{-j'}]$ (here $d_j, d_{j'}$ denote the sizes of these buckets). Also suppose without loss of generality that $d_j \ge d_{j'}$. It turns out that if we ran {\sc BasicCertify} restricted to the $(d_j + d_{j'})\times (d_j + d_{j'})$ principal submatrix of $\rho$ containing these off-diagonal blocks, then by a more involved version of the reasoning in the previous paragraph (see Lemma~\ref{lem:scenario4}), we can show that $\wt{O}(\sqrt{d_j} d^2_{j'} 2^{-j'}/\epsilon^2)$ measurements of $\rho$ suffice to detect that $\rho$ differs noticeably from $\sigma$.

Putting everything together, we conclude that our algorithm needs to make, up to log factors, \begin{equation}
    \max_{j,j': d_j \ge d_{j'}} \sqrt{d_j} d^2_{j'} 2^{-j'}/\epsilon^2 = \left(\max_j\sqrt{d_j}\right)\cdot \left(\max_{j'} d^2_{j'} 2^{-j'}\right)/\epsilon^2\label{eq:maximization}
\end{equation} measurements, where $j,j'$ range over non-negligible buckets, $j = j'$ corresponds to case (A), and $j\neq j'$ corresponds to case (B). As there are logarithmically many nonempty non-negligible buckets of eigenvalues of $\sigma$, it is elementary to check that \eqref{eq:maximization} is, up to log factors, equal to $\overline{d}_{\mathsf{eff}}^{1/2}\cdot\norm{\overline{\sigma}}_{1/2}$ (see Fact~\ref{fact:optimize}), where $\norm{\cdot}_{1/2}$ denotes the Schatten $1/2$-quasinorm. Finally we can see where the fidelity term comes from: for any density matrix $\overline{\sigma}\in\co^{d\times d}$,
\begin{equation}
    d\cdot F(\overline{\sigma},\Id/d) = d\cdot \frac{1}{d}\Tr(\overline{\sigma}^{1/2})^2 = \norm{\overline{\sigma}}_{1/2}, \label{eq:fidelity_rewrite}
\end{equation}
Thus far we have only provided justification for why fidelity emerges in the upper bound. But as we mentioned in our discussion for why the upper and lower bounds happen to ``line up,'' our generalized quantum Paninski and off-diagonal lower bound constructions closely parallel case (A) and case (B) in the upper bound analysis. Naturally, we end up seeing the same kinds of terms, e.g. $\norm{\cdot}_{2/5}$ and $\norm{\cdot}_{1/2}$, emerge in the proof of the lower bound for essentially the same reasons.

\paragraph{Roadmap} In Section~\ref{sec:prelims}, we review basic notions in quantum property testing and present various technical tools we will use in our proofs. In Section~\ref{sec:framework} we describe the general framework introduced in \cite{bubeck2020entanglement} for proving lower bounds with incoherent measurements. In Section~\ref{sec:instance}, we prove the lower bound in Theorem~\ref{thm:instance_informal}, and in Section~\ref{sec:alg} we prove the upper bound. In Appendix~\ref{subsec:adaptive_instance} we prove Theorem~\ref{thm:instance_adaptive_informal}. In Appendix~\ref{sec:defer} we collect some deferred proofs from the main body.


\section{Technical Preliminaries}
\label{sec:prelims}

\paragraph{Notation}


Let $S_{\ell}$ denote the symmetric group on $\ell$ elements. Given $\pi\in S_{\ell}$, let $\kappa(\pi)$ denote the number of cycles in $\pi$. Recall from the introduction that we let $\rhomm\triangleq \frac{1}{d}\Id$ denote the maximally mixed state. Given a matrix $M$ and $p>0$, let $\norm{M}_p$ denote the Schatten-$p$ (quasi)norm. Let $\hatM\triangleq M/\Tr(M)$. Let $U(d)$ denote the unitary group of $d\times d$ matrices.

\subsection{Quantum Property Testing}

We will work with the following standard notions, using notation and terminology borrowed from \cite{bubeck2020entanglement}.

\begin{definition}
	A positive operator-valued measurement (POVM) $\calM$ consists of a collection of psd matrices $M_1,...,M_m$ for which $\sum M_i = \Id$. We will refer to the set of measurement outcomes $[m]$ as $\Omega(\calM)$. Given mixed state $\rho$, the \emph{distribution over outcomes from measuring $\rho$ with $\calM$} is the distribution over $\Omega(\calM)$ which places mass $\iprod{M_i,\rho}$ on outcome $i$.
\end{definition}

As demonstrated in \cite{bubeck2020entanglement}, the techniques in that work and in the present paper generalize easily to POVMs for which $\Omega(\calM)$ is infinite, so for simplicity we will simply consider the finite case in this work.

\begin{definition}
\label{def:povm-schedule}
	Let $N\in\N$. A \emph{POVM schedule} $\calS$ is a collection of POVMs $\brc*{\calM^{x_{<t}}}_{t\in[N],x_{<t}\in\calT_t}$, where each $\calM^{x_{<t}}$ is over $\co^d$, $\calT_1\triangleq \brc{\emptyset}$, and for every $t>1$, $\calT_t$ denotes the set of all possible transcripts of measurement outcomes $x_{<t}$ for which $x_i\in\Omega(\calM^{x_{<i}})$ for all $1\le i\le t-1$ (recall that $x_{<i}\triangleq (x_1,...,x_{i-1})$).
	The schedule works in the natural manner: at time $t$ for $t = 1, \ldots, N$, given a transcript $x_{< t}\in\calT_t$, it measures the $t$-th copy of $\rho$ using the POVM $\calM^{x_{<t}}$.

	If in addition every $\calM^{x_{<t}}$ only depends on $t$ and not on the specific transcript $x_{<t}$, we say it is a \emph{nonadaptive POVM schedule} and denote it simply by $\brc*{\calM^t}_{t\in[N]}$.
\end{definition}

\begin{definition}[Quantum property testing task]\label{def:quantum_testing}
	A \emph{quantum property testing task} $\calT$ is specified by two disjoint sets $S_0$ and $S_1$ of mixed states. For any $N\in\N$, we say that task $T$ \emph{has copy complexity $N$} if there exists a POVM schedule $\calS$ and a (potentially randomized) post-processing algorithm $A$ so that for any $\alpha\in\brc{0,1}$ and any $\rho\in S_{\alpha}$, if $z_{\le N}$ is the transcript obtained from measuring $N$ copies of $\rho$ according to $\calS$, then $A(z_{\le N}) = \alpha$ with probability at least $2/3$ over the randomness of $\calS$ and $A$.
\end{definition}

For a mixed state $\sigma$, if we specialize Definition~\ref{def:quantum_testing} to $S_0 = \sigma$ and $S_1$ to all mixed states $\epsilon$-far in trace distance from $\sigma$, we obtain the following standard task:

\begin{definition}[State certification]
	Fix $\epsilon > 0$. Given an explicit description of a mixed state $\sigma$ along with copies of an unknown mixed state $\rho$, the task of \emph{state certification to error $\epsilon$ with respect to $\sigma$} is to determine with high probability whether $\rho = \sigma$ or $\norm{\rho - \sigma}_1 > \epsilon$ by making measurements on the copies of $\rho$. When $\sigma = \rhomm$, this is the task of \emph{mixedness testing}.
\end{definition}

We will employ the following standard framework for proving testing lower bounds:

\begin{definition}[Lower Bound Setup: Point vs. Mixture]
	In the setting of Definition~\ref{def:quantum_testing}, a \emph{point vs. mixture task} is specified by a null hypothesis $\rho\in S_0$, a set of alternatives $\rho_{\theta}\subseteq S_1$ parametrized by $\theta$, and a distribution $\calD$ over $\theta$.

	For any POVM schedule $\calS$, let $p^{\le N}_0(\calS)$ be the induced distribution over transcripts from measuring $N$ copies of $\rho$ according to $\calS$, and let $p^{\le N}_1(\calS)$ be the induced distribution over transcripts from first sampling $\theta\sim\calD$ and then measuring $N$ copies of $\rho_{\theta}$ according to $\calS$. For instance, if $\calS$ is nonadaptive, then $p^{\le N}_1$ is simply a mixture of product distributions.
\end{definition}

The following is a standard fact that lets us relate this back to property testing:

\begin{fact}\label{fact:basic_lowerbound}
	Given quantum property testing task $\calT$ specified by sets $S_0,S_1$, let $N\in\N$, and let $\calF$ be a family of measurement schedules using $N$ measurements. Suppose there exists a point vs. mixture task so that for every $\calS\in\calF$, we have that $\tvd(p^{\le N}_0(\calS), p^{\le N}_1(\calS)) \le 1/3$. Then $\calT$ has copy complexity at least $N$.
\end{fact}

For the remainder of the paper, we will fix a measurement schedule $\calS$ and just write $p^{\le N}_0$ and $p^{\le N}_1$. The possible families $\calF$ we work with in Fact~\ref{fact:basic_lowerbound} are the family of nonadaptive POVM schedules, and the family of adaptive POVM schedules.


\subsection{Tail Bounds}

We first collect some elementary facts about sub-exponential random variables.

\begin{definition}
    We say that a random variable $Z$ is \emph{$(\sigma^2,b)$-sub-exponential} if it has mean zero and satisfies
    \begin{equation}
        \Pr{|Z| > s} \le \exp\left(\frac{1}{2}\brc*{\Min{\frac{s^2}{\sigma^2}}{\frac{s}{b}}}\right)
    \end{equation}
    for all $s > 0$.
\end{definition}

It is a standard fact that sub-exponential random variables satisfy the following moment bounds:

\begin{lemma}\label{lem:subexp_moment}
    If $Z$ is $(\sigma^2,b)$-sub-exponential, then for any $t \ge 1$, $\E{|Z|^t} \le (t/2)!\cdot (2\sigma^2)^{t/2} + t!\cdot (2b)^t$.
\end{lemma}

\begin{proof}
	We have \begin{align}
	\E*{\abs*{Z}^t} &= \int^{\infty}_0 \Pr*{\abs*{Z}> s^{1/t}} \, \d s \\
	&\le \int^{\infty}_0 \exp\left(-\frac{s^{2/t}}{2\sigma^2}\right) \, \d s + \int^{\infty}_0 \exp\left(-\frac{s^{1/t}}{2b}\right) \, \d s \\
	&= \Gamma(1 + t/2) \cdot (2\sigma^2)^{t/2} + \Gamma(1+t)\cdot (2b)^t
    \end{align} as desired.
\end{proof}

It is also standard that sub-exponential random variables have mgf bounded as follows:

\begin{lemma}\label{lem:mgf}
    If $Z$ is $(\sigma^2,b)$-sub-exponential, then for any $\lambda \le \min(1/4b,1/\sigma)$, \begin{equation}
        \E{e^{\lambda Z}} \le \exp\left(O(\lambda^2(\sigma^2 + b^2))\right).
    \end{equation}
\end{lemma}

\begin{proof}
    As $\E{Z} = 0$ by definition, we can expand \begin{equation}
        \E{e^{\lambda Z}} = 1 + \sum^{\infty}_{t = 2} \frac{\lambda^t}{t!}\E{Z^t}.
    \end{equation} By Lemma~\ref{lem:subexp_moment},
    \begin{equation}
        \sum^{\infty}_{t = 2} \frac{\lambda^t}{t!}\E{Z^t} \le \sum^{\infty}_{t = 2} \left(\frac{(t/2)!}{t!}(2\lambda^2\sigma^2)^{t/2} + (2\lambda b)^t\right) = 8\lambda^2 b^2 + \sum^{\infty}_{t=2} (\lambda^2\sigma^2/2)^{t/2} \le 8\lambda^2 b^2 + \lambda^2\sigma^2.
    \end{equation}
    The lemma follows from the inequality $1 + x \le e^x$.
\end{proof}

We will need the following basic fact about sums of random variables satisfying sub-exponential moment bounds.

\begin{lemma}\label{lem:sum_subexp}
	Fix any $t\in\N$. Given a collection of independent mean-zero random variables $Z_1,\ldots,Z_m$ whose odd moments vanish and such that for every $i \in[m]$ and even $1\le \ell \le t$, $\E{\abs{Z_i}^\ell}^{1/\ell} \le \ell\cdot \sigma_i$, we have that for every even $1\le \ell \le t$ \begin{equation}
		\E{(Z_1+\cdots+Z_m)^{\ell}}^{1/\ell} \le \ell(\sigma_1^2+\cdots+\sigma_m^2)^{1/2}
	\end{equation}
\end{lemma}

\begin{proof}
	Using the sub-exponential moment bound, we can expand $\E{(Z_1 + \cdots + Z_m)^{\ell}}$ and use the fact that the $Z_i$'s are independent to get \begin{equation}
		\E{(Z_1 + \cdots + Z_m)^{\ell}} = \sum_{\alpha} \prod_i \E{Z_i^{\alpha_i}} \le \sum_{\alpha}\prod_i \alpha_i^{\alpha_i} \sigma_i^{\alpha_i} \le \ell^{\ell}\sum_{\alpha}\prod_i(\sigma_i^2)^{\alpha_i/2} = \ell^{\ell}(\sigma^2_1 + \cdots + \sigma^2_m)^{\ell/2}
	\end{equation}
	where $\alpha$ ranges over even monomials of total degree $\ell$.
\end{proof}

Concentration of measure for Haar-random unitary matrices will also be crucial to our analysis:

\begin{theorem}[\cite{meckes2013spectral}, Corollary 17, see also \cite{anderson2010introduction}, Corollary 4.4.28]
\label{thm:conc}
	Equip $M \triangleq U(d)^k$ with the $L_2$-sum of Hilbert-Schmidt metrics. If $F: M\to \R$ is $L$-Lipschitz, then for any $t > 0$: \begin{equation}
		\Pr[(\U_1,...,\U_k)\in M]{|F(\U_1,...,\U_k) - \E{F(\U_1,...,\U_k)}|\ge t} \le e^{-d t^2/12L^2},
	\end{equation} where $\U_1,...,\U_k$ are independent unitary matrices drawn from the Haar measure.
\end{theorem}

\subsection{Weingarten Calculus}

In this section we recall some standard facts about integrals over the Haar measure on the unitary group. Given a permutation $\pi\in S_{\ell}$, let $\Wg(\pi,d)$ denote the \emph{Weingarten function} (see e.g. \cite{collins2006integration}). Given a matrix $M\in\co^{d\times d}$ and permutation $\pi\in S_{\ell}$, let $\iprod{M}_{\pi} \triangleq \prod_{C\in \pi}\Tr(M^{|C|})$, where $C$ ranges over the cycles of $\pi$ and $|C|$ denotes the length of $C$. Equivalently, if $P_{\pi}$ is the permutation operator associated to $\pi$, then \begin{equation}
    \iprod{M}_{\pi} = \Tr(P_{\pi}M^{\otimes\ell}). \label{eq:permop}
\end{equation}

We will use the following consequence of the Weingarten calculus and Schur-Weyl duality:
\begin{lemma}[See e.g. Eq 7.32 from \cite{brandao2019models}]\label{lem:twirl}
    For any matrix $\vec{M}\in(\co^{d\times d})^{\otimes \ell}$,
    \begin{equation}
        \E[\U]*{{\U^{\dagger}}^{\otimes \ell} \vec{M} \U^{\otimes \ell}} = \sum_{\sigma,\tau\in S_{\ell}} \Wg(\sigma^{-1}\tau,d) \Tr(P_{\tau}\vec{M}) P_{\sigma},
    \end{equation}
    where the expectation is with respect to the Haar measure on $U(d)$.
\end{lemma}

Lemma~\ref{lem:twirl} yields the following useful integral:

\begin{lemma}\label{lem:collins}
	For $d\ge 2$, $\ell\in\N$, and any $\A,\B\in\co^{d\times d}$, we have that
	\begin{equation}
		\E[\U]{\Tr(\A\U^{\dagger}\B\U)^{\ell}} = \sum_{\pi,\tau\in S_{\ell}}\Wg(\pi^{-1}\tau,d) \iprod{\A}_{\pi}\iprod{\B}_{\tau}.
	\end{equation}
    In particular, when $\ell = 1$, $\E[\U]{\Tr(\A\U^{\dagger}\B\U)} = \frac{1}{d}\Tr(\A)\Tr(\B)$.
\end{lemma}

\begin{proof}
    We can write $\Tr(\A\U^{\dagger}\B\U)^{\ell}$ as $\Tr\left(\A^{\otimes \ell} {\U^{\dagger}}^{\otimes \ell} \B^{\otimes \ell} \U^{\otimes \ell}\right)$, so by Lemma~\ref{lem:twirl}, the expectation of this over $\U$ is $\sum_{\sigma,\tau\in S_{\ell}}\Wg(\sigma^{-1}\tau,d) \Tr(P_{\tau} \B^{\otimes \ell}) \Tr(P_{\sigma} \A^{\otimes\ell})$, and the first part of the lemma then follows by \eqref{eq:permop}. The second part of the lemma then follows by the fact that for the identity permutation $e$ on one element, $\Wg(e,d) = \frac{1}{d}$.
\end{proof}

\subsection{Block Matrices}

Here we record two basic results about block matrices, beginning with the following standard fact about Schur complements (see e.g. Theorem 1.12 from \cite{zhang2006schur}):



\begin{lemma}[Schur complements]\label{lem:schur}
	For a block matrix $\rho = \begin{pmatrix}
		\vec{A} & \vec{B} \\
		\vec{B}^{\dagger} & \vec{C},
	\end{pmatrix}$ for which $\vec{A}$ and $\vec{C}$ are positive definite, $\rho$ is positive definite if and only if Schur complement $\vec{C} - \vec{B}^{\dagger}\vec{A}^{-1}\vec{B}$ is positive definite.
\end{lemma}

The second result of this subsection upper bounds the trace norm of the off-diagonal blocks of a psd block matrix in terms of the traces of the diagonal blocks:

\begin{lemma}\label{lem:tracepsd}
	For psd block matrix $\rho = \begin{pmatrix}
		\vec{A} & \vec{B} \\
		\vec{B}^{\dagger} & \vec{C},
	\end{pmatrix}$, where $\vec{A}$ and $\vec{C}$ are square, we have that $\Tr(\vec{A})\Tr(\vec{C}) \ge \norm{\vec{B}}_1^2$. In particular, $\norm{\vec{B}}_1 \le \Tr(\rho)/2$.
\end{lemma}

\begin{proof}
	Without loss of generality suppose that $\vec{A}$ has at least as many rows/columns as $\vec{C}$. First note that we may assume $\vec{B}$ is actually square. Indeed, consider the matrix $\rho'$ given by padding $\rho$ with zeros, \begin{equation}
		\rho' = \begin{pmatrix}
			\vec{A} & \vec{B} & \vec{0} \\
			\vec{B}^{\dagger} & \vec{C} & \vec{0} \\
			\vec{0} & \vec{0} & \vec{0},
		\end{pmatrix}
	\end{equation}
	so that $\vec{A}$ and $\vec{C}'\triangleq \begin{pmatrix}
		\vec{C} & \vec{0} \\
		\vec{0} & \vec{0}
	\end{pmatrix}$ have the same dimensions. Clearly, $\norm{\begin{pmatrix} \vec{B} & \vec{0}\end{pmatrix}}_1 = \norm{\vec{B}}_1$, and $\norm{\vec{C}'}_1 = \norm{\vec{C}}_1$, so to show Lemma~\ref{lem:tracepsd} for $\rho$ it suffices to prove it for $\rho'$. So henceforth, assume $\vec{B}$ is square.

	We will further assume that $\vec{B}$ is diagonal. To see why this is without loss of generality, write the singular value decomposition $\vec{B} = \U^{\dagger}\Sig\V$ and note that \begin{equation}
		\begin{pmatrix}
			\U & \vec{0} \\
			\vec{0} & \V
		\end{pmatrix} \rho \begin{pmatrix}
			\U^{\dagger} & \vec{0} \\
			\vec{0} & \V^{\dagger}
		\end{pmatrix} = \begin{pmatrix}
			\U^{\dagger}\vec{A}\U & \Sig \\
			\Sig & \V^{\dagger}\vec{C}\V.
		\end{pmatrix}
	\end{equation}
	If $\vec{B}$ is diagonal, then for every diagonal entry $\vec{B}_{i,i}$, we have that $\vec{B}_{i,i}^2 \le \vec{A}_{i,i} \vec{C}_{i,i}$, so \begin{equation}
		\norm{\vec{B}}^2_1 = \left(\sum_i \vec{B}_{i,i}\right)^2 \le \left(\sum_i \vec{A}_{i,i}^{1/2}\vec{B}_{i,i}^{1/2}\right)^2 \le \Tr(\vec{A})\Tr(\vec{B}),
	\end{equation} where the last step is by Cauchy-Schwarz.

	The second part of the claim follows by AM-GM.
\end{proof}

\subsection{Instance-Optimal Distribution Testing}

Here we record the precise statement of the instance-optimal lower bound from \cite{valiant2017automatic}.

\begin{theorem}[\cite{valiant2017automatic}, Theorem 1]\label{thm:VV}
	Given a known distribution $p$ and samples from an unknown distribution $q$, any tester that can distinguish between $q = p$ and $\norm{p -q}_1 \ge \epsilon$ with probability 2/3 must draw at least $\Omega(\Max{1/\epsilon}{\norm{p^{-\max}_{-\epsilon}}_{2/3}/\epsilon^2})$ samples.
\end{theorem}

Note that this immediately implies a lower bound for state certification:

\begin{corollary}\label{cor:apply_VV}
	Given a known mixed state $\rho$ and copies of an unknown mixed state $\sigma$, any tester that can distinguish between $\sigma = \rho$ and $\norm{\rho - \sigma}_1 \ge \epsilon$ with probability 2/3 using measurements on the copies of $\rho$ must use at least $\Omega(\norm{\rho^{-\max}_{-\epsilon}}_{2/3}/\epsilon^2)$ samples.
\end{corollary}

We will use this corollary in our proof to handle mixed states whose eigenvalues are all pairwise separated by at least a constant factor. Intuitively, these mixed states are close to being low-rank, and one would expect that the copy complexity for testing identity to such a state is $\wt{\Theta}(1/\epsilon^2)$. We show that this is indeed the case (see Lemma~\ref{lem:spiky_nonadaptive}).

\subsection{Miscellaneous Facts}

The following elementary facts will be useful:

\begin{fact}\label{fact:optimize}
	Let $S$ be any set of distinct positive integers. Given a collection of numbers $\brc{d_j}_{j\in S}$ satisfying $\sum_j d_j 2^{-j} \le 2$, let $p$ be the vector with $d_j$ entries equal to $2^{-j}$ for every $j\in S$. Then $\max_j d_j^b 2^{-a j} \ge |S|^{-b}\norm{p}^{-a}_{a/b}$ for any $a,b > 0$.
\end{fact}

\begin{proof}
	Let $j^*$ be the index attaining the maximum. By maximality we know $d_{j^*} 2^{-aj/b} \ge \frac{1}{|S|}\sum_j d_j\cdot 2^{-aj/b}$. Raising both sides to the $b$-th power and taking reciprocals, we conclude that $2^{aj}/d_j^b \le |S|^{b}\norm{p}^a_{a/b}$.
\end{proof}

\begin{fact}\label{fact:geoseries}
	Let $c > 1$ and $p,q > 0$. Given a vector $v$ with entries $v_1 > \cdots > v_m > 0$ for which $v_i \ge c\cdot v_{i+1}$ for every $i$, we have that $\norm{v}_p \ge (1 - c^{-q})^{1/q}\cdot \norm{v}_q$. 
\end{fact}

\begin{proof}
	We have that $\norm{v}^{q}_{q} \le \sum^{\infty}_{i=1}(c^{-i} v_1)^q = \frac{v_1^{q}}{1 - c^{-q}}$, so $\norm{v}_{p} \ge v_1 \ge \norm{v}_{q}\cdot(1 - c^{-q})^{1/q}$.
\end{proof}

We will also need the following when describing the framework of \cite{bubeck2020entanglement} in Section~\ref{sec:framework}:

\begin{fact}[Integration by parts, see e.g. Fact C.2 in \cite{bubeck2020entanglement}]\label{fact:stieltjes}
	Let $a,b\in\R$. Let $Z$ be a nonnegative random variable satisfying $Z\le b$ and such that for all $x\ge a$, $\Pr{Z > x} \le \tau(x)$. Let $f: [0,b]\to\R_{\ge 0}$ be nondecreasing and differentiable. Then \begin{equation}
		\E{f(Z)} \le f(a)(1 + \tau(a)) + \int^b_a \tau(x) f'(x) \ \d\, x.
	\end{equation}
\end{fact}


\section{General Lower Bound Framework}
\label{sec:framework}

All of our lower bounds are based on analyzing a suitable point vs. mixture distinguishing problem. In this section we outline a general framework, implicit in \cite{bubeck2020entanglement}, for showing copy complexity lower bounds for such problems. After outlining some basic objects, in Section~\ref{subsec:assume} we describe a set of conditions (see Assumption~\ref{assume:main}) that, if true for a particular distinguishing problem, imply by the machinery of \cite{bubeck2020entanglement} a strong copy complexity lower bound for that problem. We formally state these implications in Sections~\ref{subsec:nonadaptive_generic} and \ref{subsec:adaptive_generic} and, for the sake of completeness, provide their proofs in Appendix~\ref{app:bcl_defer}.\footnote{That said, as our techniques are a generalization of the approach of \cite{bubeck2020entanglement}, readers unfamiliar with that work may find it more convenient to consult it first before proceeding. Either way, here we will try to distill the main ingredients from \cite{bubeck2020entanglement} in as modular a fashion as possible.}

Concretely, we will lower bound the smallest $N$ for which it is possible to distinguish, using an unentangled POVM schedule $\calS$, between $\sigma^{\otimes N}$ and $\E[\U\sim\calD]{\rho^{\otimes N}_{\U}}$ for some prior distribution $\calD$. Given schedule $\calS$, let $p^{\le N}_0$ (resp. $p^{\le N}_1$) denote the distribution over transcripts given by measuring $\sigma^{\otimes N}$ (resp. $\E[\U]{\rho^{\otimes N}_{\U}}$) with $\calS$. A key component of our analysis is to bound how well a single step of $\calS$ can distinguish between a single copy of $\sigma$ and a single copy of $\sigma_{\U}$ for $\U\sim\calD$:

\begin{definition}
	A \emph{single-copy sub-problem} $\calP = (\calM,\sigma,\brc{\sigma_{\U}}_{\U\sim\calD})$ consists of the following data: a POVM $\calM$ over $\co^d$, a mixed state $\sigma\in\co^{d\times d}$, and a distribution over mixed states $\sigma_{\U}\in\co^{d\times d}$ where $\U$ is drawn from some distribution $\calD$.
\end{definition}

To quantify how much information a single step of $\calS$ can reveal about the unknown state, we introduce the following quantities:

\begin{definition}
	Given a single-copy sub-problem $\calP = (\calM,\sigma,\brc{\sigma_{\U}}_{\U\sim\calD})$, let $p_0(\calM)$ denote the distribution over outcomes upon measuring $\sigma$ using $\calM = \brc{M_z}$. Given POVM outcome $z$, and $\U,\V\in\supp(\calD)$, define the quantities
	\begin{equation}
		g^{\U}_{\calP}(z) \triangleq \frac{\iprod{M_z,\sigma_{\U}}}{\iprod{M_z,\sigma}} - 1 \qquad \phi^{\U,\V}_{\calP} \triangleq \E[z\sim p_0(\calM)]*{g^{\U}_{\calP}(z) \cdot g^{\V}_{\calP}(z)}.
	\end{equation} We will omit the subscript $\calP$ when the context is clear.
\end{definition}

We can interpret $1 + g^{\U}_{\calP}$ as the likelihood ratio between the distribution under measuring a single copy of $\sigma_{\U}$ and the distribution under measuring a single copy of $\sigma$.

\subsection{Sufficient Conditions on \texorpdfstring{$g^{\U}_{\calP}(z)$}{gz}}
\label{subsec:assume}


We will design $\brc{\sigma_{\U}}_{\U\sim\calD}$ in such a way that the following three conditions hold.

\begin{assumption}\label{assume:main}
    Suppose that $g^{\U}_{\calP}$ satisfies the following three properties for parameters $\varsigma, L > 0$:
    \begin{enumerate}
        \item \underline{First moment bound}: For any $z\in\Omega(\calM)$, $\E[\U]{g^{\U}_{\calP}(z)}=0$.\label{cond:gexp}
        \item \underline{Second moment bound}: $\E[\U\sim\calD]{g^{\U}_{\calP}(z)^2} \le \varsigma^2$ for all measurement outcomes $z$.\label{cond:secondmoment}
        \item \underline{Lipschitzness}: $\E[z\sim p_0(\calM)]{(g^{\U}_{\calP}(z) - g^{\V}_{\calP}(z))^2}^{1/2} \le L\cdot\norm{\U - \V}_{\HS}$ for any $\U,\V\in\supp(\calD)$.\label{cond:normequiv}
    \end{enumerate}
\end{assumption}






\begin{example}\label{ex:maxmixed}
	It was shown in \cite{bubeck2020entanglement} that if $\sigma = \rhomm$, $\sigma_{\U} = \rhomm + \U^{\dagger}\diag(\frac{\epsilon}{d},\ldots,-\frac{\epsilon}{d},\ldots)\U$, and $\calD$ is given by the Haar measure over $U(d)$, then Assumption~\ref{assume:main} holds for $\varsigma, L = O(\epsilon/\sqrt{d})$ for any sub-problem $\calP$ of the form $(\calM,\sigma,\brc{\sigma_{\U}}_{\U\sim\calD})$.
\end{example}

Here we prove some intuition for these conditions. As we mentioned above, $1 + g^{\U}_{\calP}$ is simply the likelihood ratio between the distributions over outcomes under measuring a single copy of $\sigma_{\U}$ versus a single copy of $\sigma$. Condition~\ref{cond:gexp} thus ensures that for any POVM element $z$, the probability of observing outcome $z$ under $\sigma_{\U}$ is in expectation over $\U$ equal to the probability of observing $z$ under $\sigma$. By Chebyshev's, Condition~\ref{cond:secondmoment} then ensures that the former has some mild concentration around the latter. 

In other words, because of Conditions~\ref{cond:gexp} and \ref{cond:secondmoment}, there is no single observable that we can repeatedly measure $O(1/\varsigma^2)$ times to solve the point vs. mixture distinguishing problem. It turns out that if $g^{\U}_{\calP}$ additionally satisfies the Lipschitzness constraint of Condition~\ref{cond:normequiv}, then we can invoke concentration of Lipschitz functions of Haar-random unitary matrices (recall Theorem~\ref{thm:conc} from the preliminaries) to get a strong lower bound for the distinguishing problem. 

This last point requires some unpacking. For starters, let us spell out what kinds of tail bounds we leverage. Specifically, using Assumption~\ref{assume:main} and concentration of measure, one can show the following tail bound which is an important starting point for our lower bounds.

\begin{lemma}\label{lem:subexp}
    Suppose $\calP$ satisfies Assumption~\ref{assume:main} for parameters $\varsigma, L > 0$. Then for $\U,\V$ sampled independently from the Haar measure over $U(d)$, $\phi^{\U,\V}_{\calP}$ is a $\left(\Theta(\varsigma^2 L^2/d), \Theta(L^2/d)\right)$-sub-exponential random variable in the randomness of $\U,\V$.
    In particular, by Lemma~\ref{lem:subexp_moment}, 
    \begin{equation}
		\E[\U,\V]*{\abs*{\phi^{\U,\V}_{\calP}}^t}^{1/t} \le O\left(\Max{\varsigma L\sqrt{t/d}}{L^2t/d}\right) \le O(t\cdot L\cdot\brc{\Max{\varsigma}{L}}/\sqrt{d}) \label{eq:moment_tri}
	\end{equation}
\end{lemma}

In the next two sections, we show how to use Lemma~\ref{lem:subexp} to derive lower bounds for the distinguishing problem.

\subsection{Non-adaptive Lower Bounds}
\label{subsec:nonadaptive_generic}

As discussed in Section~\ref{sec:overview}, our non-adaptive lower bounds are based on the Ingster-Suslina method~\cite{ingster2012nonparametric}. In \cite{bubeck2020entanglement}, the main ingredients of this method are stated in the preceding notation as follows:

\begin{lemma}[\cite{bubeck2020entanglement}, Lemma 2.8]\label{lem:ingster}
	If the unentangled POVM schedule $\calS$ is non-adaptive and consists of POVMs $\calM_1,...,\calM_N$, then if $\calP_t = (\calM_t,\sigma,\brc{\sigma_{\U}}_{\U\sim\calD})$ denotes the $t$-th single-copy sub-problem for an arbitrary $\calD$, then \begin{equation}
		\chisq{p^{\le N}_1}{p^{\le N}_0} \le \max_{t\in[N]} \E[\U,\V\sim\calD]*{\left(1 + \phi^{\U,\V}_{\calP_t}\right)^N} - 1\label{eq:ingster}
	\end{equation}
\end{lemma}

Lemma~\ref{lem:ingster} is one reason why we care about tail bounds for $\phi^{\U,\V}_{\calP}$: with sufficiently good moment bounds on $\phi$, we can upper bound the right-hand side of \eqref{eq:ingster} and conclude that for $N$ small, the chi-squared divergence between $p^{\le N}_1$ and $p^{\le N}_0$ is small. By Pinsker's, this implies that the total variation distance between $p^{\le N}_1$ and $p^{\le N}_0$ is small, so by Fact~\ref{fact:basic_lowerbound} we get a lower bound on the copy complexity $N$ of distinguishing $\sigma^{\otimes N}$ and $\E{\sigma_{\U}^{\otimes N}}$. We spell this out explicitly in the next lemma.

\begin{lemma}\label{lem:ingster_with_conds}
	Let $\calD$ be the Haar measure over $U(d)$, and fix $\sigma$ and $\brc{\sigma_{\U}}_{\U\sim\calD}$. Suppose that for any POVM $\calM$, the single-copy sub-problem $\calP = (\calM, \sigma, \brc{\sigma_{\U}}_{\U\sim\calD})$ satisfies Assumption~\ref{assume:main}.
    Then distinguishing $\sigma^{\otimes N}$ from $\E[\U]{\rho^{\otimes N}_{\U}}$ with probability at least 2/3 using an unentangled, non-adaptive POVM schedule $\calS$ requires $N = \Omega\left(\Min{\sqrt{d}/(L\varsigma)}{d/L^2}\right)$.
\end{lemma}

\begin{proof}
	Fix any $t\in[N]$ and note that $(1 + \phi^{\U,\V}_{\calP_t})^N\le \exp\left(N\phi^{\U,\V}_{\calP_t}\right)$. As $\phi^{\U,\V}_{\calP_t}$ is $\left(\Theta(\varsigma^2 L^2/d), \Theta(L^2/d)\right)$-sub-exponential, its moment generating function is bounded by Lemma~\ref{lem:mgf}. In particular, for any $N \le O(d/L^2)$,
	\begin{equation}
	    \E[\U,\V]*{\exp\left(N\phi^{\U,\V}_{\calP_t}\right)} \le \exp\left(O(N^2(\varsigma^2 L^2/d + L^4/d^2)\right),
	\end{equation}
	so for $N = o\left(\Min{\sqrt{d}/(L\varsigma)}{d/L^2}\right)$, the above quantity is $1 +o(1)$.  The lemma then follows from relating KL to total variation using Pinsker's and then invoking Fact~\ref{fact:basic_lowerbound}.
\end{proof}

\begin{example}
    If $\sigma = \rhomm$, $\sigma_{\U} = \rhomm + \U^{\dagger}\diag(\frac{\epsilon}{d},\ldots,-\frac{\epsilon}{d},\ldots)\U$, and $\calD$ is the Haar measure on $U(d)$, recall from Example~\ref{ex:maxmixed} that we can take $\varsigma,L = O(\epsilon/\sqrt{d})$. So by Lemma~\ref{lem:ingster_with_conds} we get a lower bound of $N = \Omega(d^{3/2}/\epsilon^2)$. This recovers the non-adaptive lower bound for mixedness testing from \cite{bubeck2020entanglement}.
\end{example}









\subsection{Adaptive Lower Bounds}
\label{subsec:adaptive_generic}

For our adaptive lower bounds, we follow the chain rule-based framework introduced in \cite{bubeck2020entanglement}, the main result of which can be abstracted as follows:

\begin{theorem}[Implicit in \cite{bubeck2020entanglement}]\label{thm:bcl}
	Let $\calD$ be the Haar measure over $U(d)$, and fix $\sigma$ and $\brc{\sigma_{\U}}_{\U\sim\calD}$. Suppose that for any POVM $\calM$, the single-copy sub-problem $\calP = (\calM, \sigma, \brc{\sigma_{\U}}_{\U\sim\calD})$ satisfies Assumption~\ref{assume:main} and additionally, for all $z\in\Omega(\calM)$, $\abs{g^{\U}_{\calP}(z)} \le 0.99$ almost surely.
	Then for any $\tau > 0$ and $N = o(d/L^2)$, \begin{equation}
		\KL{p^{\le N}_1}{p^{\le N}_0} \le N\tau + O(N)\cdot \exp\left(-\Omega\left(\brc*{\Min{\frac{d\tau^2}{L^2\varsigma^2}}{\frac{d\tau}{L^2}}} - N\cdot \varsigma^2\right)\right).\label{eq:adaptive_bound}
	\end{equation}
\end{theorem}





Like the proof of Theorem~\ref{lem:ingster_with_conds}, the proof of Theorem~\ref{thm:bcl} also makes crucial use of the fact that $\phi^{\U,\V}_{\calP}$ is a sub-exponential random variable. As it is somewhat more involved, we defer the proof to Appendix~\ref{app:bcl_defer}.

\begin{example}
	Take any $\epsilon \le 0.99$. If $\sigma = \rhomm$ and $\sigma_{\U} = \rhomm + \U^{\dagger}\diag(\frac{\epsilon}{d},\ldots,-\frac{\epsilon}{d},\ldots)\U$ as in Example~\ref{ex:maxmixed}, where recall that $\U\sim\calD$ for $\calD$ given by the Haar measure over $U(d)$, then note that \begin{equation}
	    |g^{\U}_{\calP}(z)| \le \norm{\U\diag(\epsilon,\ldots,-\epsilon,\ldots)\U^{\dagger}} = \epsilon \le 0.99
	\end{equation} for any sub-problem $\calP$ of the form $(\calM,\sigma,\brc{\sigma_{\U}}_{\U\sim\calD})$. So by taking $\tau = \epsilon^2/d^{4/3}$ in Theorem~\ref{thm:bcl}, one gets that for $N = o(d^{4/3}/\epsilon^2)$, the KL divergence in \eqref{eq:adaptive_bound} is $o(1)$. This recovers the $\Omega(d^{4/3}/\epsilon^2)$ adaptive lower bound for mixedness testing from \cite{bubeck2020entanglement}.
\end{example}





\newcommand{\Sjunk}{S_{\mathsf{tail}}}
\newcommand{\Slight}{S_{\mathsf{light}}}
\newcommand{\Ssing}{S_{\mathsf{sing}}}
\newcommand{\Smany}{S_{\mathsf{many}}}
\newcommand{\rhojunk}{\rho_{\mathsf{junk}}}
\newcommand{\sigmajunk}{\sigma_{\mathsf{junk}}}
\newcommand{\Moff}{\vec{M}_{\mathsf{off}}}

\section{Nonadaptive Lower Bound for State Certification}
\label{sec:instance}

In this section we will show our instance-near-optimal lower bounds for state certification with nonadaptive, unentangled measurements.

\begin{theorem}\label{thm:cert_lower_main}
	There is an absolute constant $c > 0$ for which the following holds for any $0 < \epsilon < c$.\footnote{As presented, our analysis yields $c$ within the vicinity of $1/3$, but we made no attempt to optimize for this constant.} Let $\sigma\in\co^{d\times d}$ be a diagonal density matrix. There is a matrix $\sigma^{**}$ given by zeroing out at most $O(\epsilon)$ mass from $\sigma$ (see Definition~\ref{defn:remove_nonadaptive} and Fact~\ref{fact:fewbuckets} below), such that the following holds:

	Let $\wh{\sigma}^{**} \triangleq \sigma^{**}/\Tr(\sigma^{**})$, and let $d_{\mathsf{eff}}$ denote the number of nonzero entries of $\sigma^{**}$. 
	Then any algorithm for state certification to error $\epsilon$ with respect to $\sigma$ using nonadaptive, unentangled measurements has copy complexity
	 at least \begin{equation}
		\Omega\left(d\sqrt{d_{\mathsf{eff}}}\cdot F(\wh{\sigma}^{**},\rhomm)/(\epsilon^2\polylog(d/\epsilon))\right).
	\end{equation}
\end{theorem}

In Section~\ref{sec:bucketing_nonadaptive}, we describe a bucketing scheme that will be essential to our analysis. In Section~\ref{subsec:instance_instance} we describe and analyze the first of our two lower bound instances, a distinguishing problem based on a generalization of the standard quantum Paninski construction. Specifically, in Section~\ref{subsubsec:general_nonadaptive}, we give a generic copy complexity lower bound for this problem, and in Section~\ref{subsubsec:tune_nonadaptive} we show how to tune the relevant parameters to obtain a copy complexity lower bound based on the Schatten 2/5-quasinorm of $\sigma$. In Section~\ref{subsec:instance_instance_2}, we describe and analyze the second of our two lower bound instances, a distinguishing problem based on perturbing the off-diagonal entries of an appropriately chosen principal submatrix of $\sigma$, obtaining for restricted choices of $\epsilon$ a copy complexity lower bound based on the effective dimension and Schatten 1/2-quasinorm of $\sigma$. In Section~\ref{subsec:conclude_nonadaptive}, we put together the analyses of our two lower bound instances to conclude the proof of Theorem~\ref{thm:cert_lower_main}.

\subsection{Bucketing and Mass Removal}
\label{sec:bucketing_nonadaptive}

We may without loss of generality assume that $\sigma$ is some diagonal matrix $\diag(\lambda_1,\ldots,\lambda_d)$.

For $j\in\Z_{\ge 0}$, let $S_j$ denote the set of indices $i\in[d]$ for which $\lambda_i\in\brk{2^{-j-1},2^{-j}}$; denote $|S_j|$ by $d_j$. Let $\calJ$ denote the set of $j$ for which $S_j\neq \emptyset$. We will refer to $j\in\calJ$ as \emph{buckets}. It will be convenient to refer to the index of the bucket containing a particular index $i\in[d]$ as $j(i)$. Also let $\Ssing$ denote the set of $i\in[d]$ belonging to a size-1 bucket $S_j$ for some $j\in\calJ$, and let $\Smany$ denote the set of $i\in[d]$ which lie in a bucket $S_j$ of size greater than 1 for some $j\in\calJ$.

Our bounds are based on the following modification of $\sigma$ obtained by zeroing out a small fraction of its entries:

\begin{definition}[Removing low-probability elements- nonadaptive lower bound]\label{defn:remove_nonadaptive}
	Without loss of generality, suppose that $\lambda_1,\ldots,\lambda_d$ are sorted in ascending order according to $\lambda_i/d^2_{j(i)}$.\footnote{The only place where we need this particular choice of sorting is in the proof of Corollary~\ref{cor:min} below.} Let $d'\le d$ denote the largest index for which $\sum^{d'}_{i=1}\lambda_i \le 3\epsilon$. Let $\Sjunk\triangleq [d']$, and let $\Slight$ be the set of $i\in\brc{d'+1,\ldots,d}$ for which $\sum_{i'\in S_{j(i)}\backslash\Sjunk} \lambda_{i'} \le 2\epsilon/\log(d/\epsilon)$. 

	Let $i_{\max}$ denote the index of the largest entry of $\sigma$. Let $\sigma'$ denote the matrix given by zeroing out the largest entry of $\sigma$ and the entries indexed by $\Sjunk$, and let $\sigma^*$ denote the matrix given by zeroing out the entries indexed by $\Sjunk\cup \Slight$. Finally, let $\sigma^{**}$ denote the matrix given by further zeroing out from $\sigma^*$ as many of the smallest entries as possible without removing more than $2\epsilon$ mass.

	Lastly, it will be convenient to define $\calJ'$ (resp. $\calJ^*$) to be the set of $j\in\calJ$ for which $S_j$ has nonempty intersection with $(([d]\backslash\brc{i_{\max}})\cap \Smany)\backslash\Sjunk$ (resp. $[d]\backslash(\Sjunk\cup \Slight)$). Note that by design, $\calJ'$ and $\calJ^*$ denote the indices of the nonzero diagonal entries of $\sigma'$ and $\sigma^*$ respectively.
\end{definition}

We will use the following basic consequence of bucketing:

\begin{fact}\label{fact:fewbuckets}
	There are at most $O(\log(d/\epsilon))$ indices $j\in\calJ$ for which $S_j$ and $\Sjunk$ are disjoint. As a consequence, $\Tr(\sigma^{**}) \ge 1 - O(\epsilon)$.
\end{fact}

\begin{proof}
	For any $i_1\not\in\Sjunk$ and $i_2\in\Sjunk$, we have that $p_{i_1}/d^2_{j(i_1)} \ge p_{i_2}/d^2_{j(i_2)}$, so $p_{i_1} \ge p_{i_2}/d^2$. In particular, summing over $i_2\in\Sjunk$, we conclude that $p_{i_1}\cdot |\Sjunk| \ge \epsilon/d^2$, so $p_{i_1} \ge \epsilon/d^3$. By construction of the buckets $S_j$, the first part of the claim follows. For the second part, by definition we have that $\sum_{i\in[d']} \lambda_i \le O(\epsilon)$. Furthermore, $\sum_{i\in\Slight}\lambda_i = O(\epsilon)$ because of the first part of the claim. The second part of the claim follows by triangle inequality.
\end{proof}




Lastly, we will use the following shorthand: for any $j\in\calJ$ and any matrix $\A$, we will let $\A_j\in\R^{d\times d}$ denote the matrix which is zero outside of the principal submatrix indexed by $S_j$ and which agrees with $\A$ within this submatrix.

\subsection{Lower Bound Instance I: General Quantum Paninski}
\label{subsec:instance_instance}

We will analyze the following distinguishing problem. We will pick a diagonal matrix $\Eps$ as follows:

\begin{definition}[Perturbation matrix $\Eps$]\label{defn:perturb}
	For any $i\not\in\Smany$, we will take the $i$-th diagonal entry of $\Eps$ to be zero. For any bucket $j$ of size at least 2, we will take the nonzero diagonal entries of $\Eps_j$ to be $(\epsilon_j,\cdots,-\epsilon_j,\cdots)$ where there are $\floor{d_j/2}$ copies of $\epsilon_j$ and $\floor{d_j/2}$ copies of $-\epsilon_j$, for $\epsilon_j$ to be optimized later.

	Given $\U\in U(d)$, define $\sigma_{\U}\triangleq \sigma + \U^{\dagger}\Eps\U$.
\end{definition}

Throughout this subsection, let $\calD$ denote the distribution over block-diagonal unitary matrices $\U$ which are zero outside of the principal submatrices indexed by $S_j$ for some $j\in\calJ$ with $d_j > 1$, and which within each submatrix indexed by such an $S_j$ is an independent Haar-random unitary if $d_j$ is even, and otherwise is an independent Haar-random unitary in the submatrix consisting of the first $2\floor{d_j/2}$ rows/columns. This distinction will not be particularly important in the sequel, so the reader is encouraged to imagine that $d_j$ is always even when $d_j > 1$.

The objective of this subsection is to show the following lower bound:

\begin{lemma}\label{lem:nonadaptive_1}
	Fix $0 < \epsilon < c$ for sufficiently small absolute constant $c > 0$. Let $\sigma\in\co^{d\times d}$ be a diagonal density matrix. There is a choice of $\Eps$ in Definition~\ref{defn:perturb} for which distinguishing between whether $\rho = \sigma$ or whether $\rho = \sigma + \U^{\dagger}\Eps\U$ for $\U\sim\calD$ using nonadaptive, unentangled measurements has copy complexity at least $\Omega(\norm{\sigma'}_{2/5}/(\epsilon^2\log(d/\epsilon)))$.
\end{lemma}

By definition of $\calD$, $\rho$ is block-diagonal in either scenario, and the block-diagonal structure depends only on $\brc{S_j}$. In particular, this implies that we can without loss of generality assume that the POVMs the tester uses respect this block structure. More precisely:

\begin{lemma}\label{lem:assume_POVM}
 	Let $\rho\in\co^{d\times d}$ be any density matrix which is zero outside of the principal submatrices indexed by the subsets $\brc{S_j}_{j\in\calJ}$. Given an arbitrary POVM $\calM = \brc{M_z}$, there is a corresponding POVM $\calM'$ satisfying the following. Let $p, p'$ be the distributions over measurement outcomes from measuring $\rho$ with $\calM, \calM'$ respectively. Then:
 	\begin{itemize}
 	 	\item For every $z\in\Omega(\calM')$, there exists $j\in\calJ$ for which $M'_z$ is zero outside of the principal submatrix indexed by $S_j$
 	 	\item There is a function $f:\Omega(\calM')\to\Omega(\calM)$ for which the pushforward of $p'$ under $f$ is $p$.
 	\end{itemize} 
\end{lemma}

\begin{proof}
	For every $z\in\Omega(\calM)$ and every $j\in\calJ$, define a POVM element $M_{j,z}\triangleq \Pi_j M_z \Pi_j$, where $\Pi_j\in\co^{d\times d}$ is the matrix which is equal to the identity in the principal submatrix indexed by $S_j$ and is zero elsewhere. Clearly $\brc{M_{j,z}}_{j\in\calJ, z\in\Omega(\calM)}$ is still a POVM because $\sum \Pi_j = \Id$; let $\calM'$ be this POVM. Let $f$ be given by $f((j,z)) = z$. The pushforward of $p'$ under $f$ places mass \begin{equation}
		\sum_{j\in\calJ} \iprod{\rho, \Pi_j M_z\Pi_j} = \iprod*{\sum_{j\in\calJ} \Pi_j \rho \Pi_j, M_z} = \iprod{\rho, M_z}
	\end{equation} on $z\in\Omega(\calM)$ as claimed, where the penultimate step follows by the assumption that $\rho$ is zero outside of the principal submatrices indexed by the subsets $\brc{S_j}$.
\end{proof}

By Lemma~\ref{lem:assume_POVM}, we will henceforth only work with POVMs like $\calM'$. If $\calM^t$ is the $t$-th POVM used by the tester, we may assume without loss of generality that its outcomes $\Omega(\calM^t)$ consist of pairs $(j,z)$, where the POVM element corresponding to such a pair has nonzero entries in the principal submatrix indexed by $S_j$. Henceforth, fix an arbitrary such POVM $\calM$ (we will drop subscripts accordingly) and denote its elements by $\brc{M_{j,z}}$ for $j\in\calJ$. We will denote by $\Omega_j$ the set of $z$ for which there is an element $M_{j,z}$.

Let $p$ denote the distribution over $\calJ$ induced by measuring $\sigma$ with $\calM$ and recording which bucket the outcome belongs to. Concretely, $p$ places mass $p_j\triangleq \sum_{z\in\Omega_j}\iprod{M_{j,z},\sigma_j} = \Tr(\sigma_j)$ on bucket $j\in\calJ$. Similarly, define $q^j$ to be the distribution over $\Omega_j$ conditioned on the outcome falling in bucket $j$, that is, $q^j$ places mass $q^j_z\triangleq \frac{1}{p_j}\iprod{M_{j,z},\sigma_j}$ on $z\in\Omega_j$.

For every $j\in\calJ$, let $\calP_j$ denote the single-copy sub-problem in $d_j$ dimensions given by restricting to the coordinates indexed by $S_j$ and using the POVM $\calM_j\triangleq \brc{(M_{j,z})_j}_{z\in\Omega_j}$. Formally, $\calP_j$ is specified by the data $(\calM_j,\wh{\sigma}_j,\brc{(\wh{\sigma}_{\U})_j}_{\U\sim\calD_j})$, where $\calD_j$ is the Haar measure over $U(d_j)$ if $d_j$ is even and is otherwise the distribution over $d_j\times d_j$ matrices which are Haar-random unitary in the first $2\floor{d_j/2}$ rows/columns and zero elsewhere. Note that the density matrix $(\wh{\sigma}_{\U})_j$ can be written as $\wh{\sigma}_j + \U^{\dagger}\Eps'_j\U$ for $\Eps'_j \triangleq \Eps_j/p_j$.

For any $j\in\calJ$, $z\in\Omega_j$, it will be convenient to define $\wt{M}_{j,z}\triangleq \frac{1}{\iprod{M_{j,z},\sigma_j}}M_{j,z}$. We can write \begin{equation}
	g^{\U_j}_{\calP_j}(z) = \frac{\iprod{M_{j,z},\U^{\dagger}_j\Eps'_j\U_j}}{\iprod{M_{j,z},\wh{\sigma}_j}} = \frac{\iprod{M_{j,z},\U^{\dagger}_j\Eps_j\U_j}}{\iprod{M_{j,z},\sigma_j}} = \iprod{\wt{M}_{j,z},\U^{\dagger}_j\Eps_j\U_j}.\label{eq:gjdef}
\end{equation}
Because $M_{j,z}$ is zero outside of the principal submatrix indexed by $S_j$, we thus have
\begin{align}
	g^{\U}(z) = \frac{\iprod{M_{j,z},\U^{\dagger}\Eps\U}}{\iprod{M_{j,z},\sigma}} = \frac{\iprod{M_{j,z},\U^{\dagger}_j\Eps_j\U_j}}{\iprod{M_{j,z},\sigma_j}} = g^{\U_j}_{\calP_j}(z)\label{eq:gintermsofgj}
\end{align}
and
\begin{equation}
	\phi^{\U,\V} = \E[j,z]*{\frac{\iprod{M_{j,z},\U^{\dagger}_j\Eps_j\U_j}\iprod{M_{j,z},\V^{\dagger}_j\Eps_j\V_j}}{\iprod{M_{j,z},\sigma_j}^2}} = \sum_{j\in\calJ} p_j \cdot \phi^{\U_j,\V_j}_{\calP_j}. \label{eq:phij_def}
\end{equation}

We now give a generic lower bound for the distinguishing problem in Lemma~\ref{lem:nonadaptive_1} that depends on the entries of $\Eps$. After that, we show how to tune the entries of $\Eps$ to complete the proof of Lemma~\ref{lem:nonadaptive_1}.

\subsubsection{Bound Under General Perturbations}
\label{subsubsec:general_nonadaptive}

Our goal is first to show the following generic bound:

\begin{lemma}\label{lem:generic_nonadaptive}
	Distinguishing $\sigma^{\otimes N}$ from $\E[\U]{\sigma^{\otimes N}_{\U}}$ with probability at least 2/3 using an unentangled, adaptive POVM schedule $\calS$ requires \begin{equation}
		N = \Omega\left(\left(\sum_{j\in\calJ} \frac{2^{2j}\epsilon_j^4}{d_j}\right)^{-1/2}\right)\label{eq:Nbound}
	\end{equation}
\end{lemma}

By Lemma~\ref{lem:ingster}, it suffices to show that for any POVM $\calM$, $\E[\U,\V]*{\left(1 + \phi^{\U,\V}_{\calM}\right)^N} = 1 + o(1)$ for $N$ smaller than the claimed bound. To do this, we will bound the moments of each $\phi^{\U,\V}_{\calP_j}$ individually.

As the relevant matrices $(M_{j,z})_j$ are zero outside of the principal submatrix indexed by $S_j$, we will abuse notation and refer to them as $M_{j,z}$ in the sequel whenever the context is clear. Likewise, we will refer to $\U_j\sim\calD_j$ as $\U$.

In the next three lemmas, we verify that the three conditions of Assumption~\ref{assume:main} are satisfied for appropriate choices of $\varsigma, L$ by the $d_j$-dimensional single-copy sub-problem $\calP_j$. For the proofs of these lemmas, it will be convenient to define $\wt{M}_{j,z}\triangleq \frac{1}{\iprod{M_{j,z},\sigma_j}}M_{j,z}$

\begin{lemma}\label{lem:gexp_instance}
	For any $z\in\Omega_j$, $\E[\U]{g^{\U}_{\calP_j}(z)} = 0$, so Condition~\ref{cond:gexp} of Assumption~\ref{assume:main} holds.
\end{lemma}

\begin{proof}
	By the second part of Lemma~\ref{lem:collins}, $\E[\U]{g^{\U}_{\calP_j}(z)} = \Tr(\wt{M}_{j,z})\cdot \Tr(\Eps_j) = 0$.
\end{proof}

\begin{lemma}\label{lem:GV_instance}
	$\E[\U]{g^{\U}_{\calP_j}(z)^2}^{1/2} \le O(2^j\epsilon_j/\sqrt{d_j})$ for any $z\in\Omega_j$, so Condition~\ref{cond:secondmoment} of Assumption~\ref{assume:main} holds.
\end{lemma}

\begin{proof}
	Let $\tau^*\in S_2$ denote transposition. For any $z\in\Omega_j$, by \eqref{eq:gjdef} and Lemma~\ref{lem:collins}, \begin{align}
		\E[\U]{g^{\U}_{\calP_j}(z)^2} &= \E*{\iprod*{\wt{M}_{j,z},\U^{\dagger}\Eps_j\U}^2} \\
		&= \sum_{\pi,\tau\in S_2}\iprod{\Eps_j}_{\tau}\iprod{\wt{M}_{j,z}}_{\pi}\Wg(\pi\tau^{-1},d_j) \\
		&= \iprod{\Eps_j}_{\tau^*}\left(\Tr(\wt{M}^2_{j,z})\cdot \Wg(e,d_j) + \Tr(\wt{M}_{j,z})^2\cdot \Wg(\tau^*,d_j)\right) \\
		&\le d_j \cdot \epsilon_j^2\cdot \frac{\Tr(M_{j,z})^2}{\iprod{M_{j,z},\sigma_j}^2}\left(\frac{1}{d_j^2 - 1}\Tr(\wh{M}_{j,z}^2) - \frac{1}{d_j(d_j^2-1)}\cdot\Tr(\wh{M}_{j,z})^2\right) \\
		&\le \frac{\epsilon_j^2}{d_j+1}\cdot \frac{\Tr(M_{j,z})^2}{\iprod{M_{j,z},\sigma_j}^2} \le 2\cdot 2^{2j}\epsilon_j^2/d_j,
	\end{align} where in the last step we used the fact that $\Tr(\wh{M}^2) \le 1$ for any matrix $\wh{M}$ of trace 1.
\end{proof}

\begin{lemma}\label{lem:glip_instance}
	$\E[z\sim q^j]{(g^{\U}_{\calP_j}(z) - g^{\V}_{\calP_j}(z))^2}^{1/2} \le O((2^{j}/p_j)^{1/2}\epsilon_j)\cdot \norm{\U - \V}_{\HS}$ for any $\U,\V\in U(d)$, so Condition~\ref{cond:normequiv} of Assumption~\ref{assume:main} holds.
\end{lemma}

\begin{proof}
	The matrix $\A\triangleq \U^{\dagger}\Eps_j\U - \U'^{\dagger}\Eps_j\U'$ is Hermitian, so write its eigendecomposition $\A = \W^{\dagger}\Sig\W$. Define $M'_{j,z}\triangleq \W M_{j,z}\W^{\dagger}$ so that $\sum_{z\in\Omega_j}M'_{j,z} = \Id_{d_j}$ and \begin{align}
		\E[z\sim q^j]{(g^{\U}_{\calP_j}(z) - g^{\V}_{\calP_j}(z))^2} &= \E[z\sim q_j]*{\left(\frac{1}{\iprod{M_{j,z},\sigma_j}}\sum^{d_j}_{i = 1} (M'_{j,z})_{ii}\Sig_{ii} \right)^2} \\
		&\le \E[z\sim q_j]*{\left(\frac{1}{\iprod{M_{j,z},\sigma_j}}\sum^{d_j}_{i=1}(M'_{j,z})_{ii}\Sigma_{ii}^2\right)\left(\frac{1}{\iprod{M_{j,z},\sigma_j}}\sum^{d_j}_{i=1}(M'_{j,z})_{ii}\right)} \\
		&\le \frac{1}{p_j}\sum_{z\in\Omega_j}\frac{\Tr(M_{j,z})}{\iprod{M_{j,z},\sigma_j}}\cdot \sum^{d_j}_{i=1}(M'_{j,z})_{ii}\Sig^2_{ii} \\
		&\le \frac{1}{p_j}2^{j+1}\cdot \sum^{d_j}_{i=1}\Sig^2_{ii}\sum_{z\in \Omega_j}(M'_{j,z})_{ii} = \frac{1}{p_j}2^{j+1}\norm{\Sig}^2_{\HS} 
	\end{align} where in the second step we used Cauchy-Schwarz, in the third step we used that $\Tr(M'_{j,z}) = \Tr(M_{j,z})$, in the fourth step we used the fact that the entries of diagonal matrix $\sigma_j$ are lower bounded by $2^{-j-1}$, and in the fifth step we used that $\sum_z \Omega'_{j,z} = \Id_{d_j}$. To upper bound $\norm{\Sig}_{\HS}$, note \begin{equation}
		\norm{\Sig}_{\HS} = \norm{\U^{\dagger}\Eps_j\U - \U'^{\dagger}\Eps_j\U'}_{\HS} = \norm{\U^{\dagger}\Eps_j(\U - \U') + (\U' - \U)^{\dagger}\Eps_j\U'}_{\HS} \le \epsilon_j\norm{\U - \U'}_{\HS},
	\end{equation} from which we conclude that $\E[z\sim q^j]{(g^{\U}_{\calP_j}(z) - g^{\V}_{\calP_j}(z))^2}^{1/2} \le (2^{j+1}/p_j)^{1/2}\epsilon_j\norm{\U - \U'}_{\HS}$.
\end{proof}

By applying \eqref{eq:moment_tri} in Lemma~\ref{lem:subexp}, we get the following bound:

\begin{lemma}\label{lem:phij}
	For any odd $t$, $\E[\U,\V\sim\calD_j]*{\left(\phi^{\U,\V}_{\calP_j}\right)^t} = 0$, and for any even $t$, \begin{equation}
		\E[\U,\V\sim\calD_j]*{\left(\phi^{\U,\V}_{\calP_j}\right)^t}^{1/t} \le O\left(2^{2j}\epsilon_j^2/d_j\cdot\brc*{\Max{\sqrt{t/d_j}}{t/d_j}}\right) \le O\left(t \cdot 2^{2j}\cdot \epsilon_j^2/d_j^{3/2}\right).
	\end{equation}
\end{lemma}

\begin{proof}
	By Lemma~\ref{lem:gexp_instance} and the definition of $\phi^{\U,\V}_{\calP_j}$, $\E{\phi^{\U,\V}_{\calP_j}} = 0$. By Lemmas~\ref{lem:GV_instance} and \ref{lem:glip_instance}, we can take $\varsigma = O(2^{j}\epsilon_j/\sqrt{d_j})$ and $L = O((2^j/p_j)^{1/2}\epsilon_j)$ when invoking \eqref{eq:moment_tri} in Lemma~\ref{lem:subexp}. Note that $p_j \ge d_j 2^{-j-1}$, so $L \le O(\varsigma)$. The claim follows.
\end{proof}

Lemma~\ref{lem:phij}, Lemma~\ref{lem:sum_subexp}, and \eqref{eq:phij_def} immediately imply Lemma~\ref{lem:generic_nonadaptive}.

\begin{proof}[Proof of Lemma~\ref{lem:generic_nonadaptive}]
	From Lemma~\ref{lem:sum_subexp}, Lemma~\ref{lem:phij}, and \eqref{eq:phij_def}, we have that 
	\begin{equation}
		\E[\U,\V\sim\calD]*{\left(\phi^{\U,\V}\right)^t}^{1/t} \le t\left(\sum_{j\in\calJ}p^2_j \cdot O\left(\frac{2^{4j}\epsilon_j^4}{d_j^3}\right)\right)^{1/2} \le t\left(\sum_{j\in\calJ} O\left(\frac{2^{2j}\epsilon_j^4}{d_j}\right)\right)^{1/2}\label{eq:phimoment}
	\end{equation}
	where in the second step we used that $p_j \le d_j 2^{-j}$. We can thus expand \begin{equation}
		\E*{\left(1 + \phi^{\U,\V}\right)^N} = \sum_{2\le t \le N \ \text{even}}\binom{N}{t}\E{(\phi^{\U,\V})^t} \le \left(\frac{e\cdot N}{t}\right)^t\cdot O\left(t^2\sum_{j\in\calJ}\frac{2^{2j}\epsilon_j^4}{d_j}\right)^{t/2},
	\end{equation} from which the claim follows by Lemma~\ref{lem:ingster}.
\end{proof}

\subsubsection{Tuning the Perturbations}
\label{subsubsec:tune_nonadaptive}

Before we explain how to tune $\Eps_j$, we address a minor corner case. Recall from Definition~\ref{defn:perturb} that $\Eps_j$ is zero for buckets $j$ for which $|S_j| = 1$. In the extreme case where all buckets after removal of $\Sjunk$ are of this type, then $\Eps = 0$ and the problem of distinguishing between $\sigma$ and $\sigma + \U^{\dagger}\Eps\U$ would be vacuous. Fortunately, we can show that if the Schatten $2/5$-quasinorm of $\sigma'$ is dominated by such buckets, then the resulting state certification problem requires many copies because of existing \emph{classical} lower bounds.

\begin{lemma}\label{lem:spiky_nonadaptive}
	If $\sum_{i\in\Ssing\backslash\Sjunk} \lambda^{2/5}_i \ge \frac{1}{2}\norm{\sigma'}^{2/5}_{2/5}$, then state certification with respect to $\sigma$ using nonadaptive, unentangled measurements has copy complexity at least $\Omega(\norm{\sigma'}_{2/5}/\epsilon^2)$.
\end{lemma}

\begin{proof}
	Intuitively in this case, the spectrum of $\sigma$ is dominated by eigenvalues in geometric progression, and in fact the instance-optimal lower bound for \emph{classical} identity testing \cite{valiant2017automatic} already implies a good enough copy complexity lower bound (even against entangled measurements). 

	Formally, Corollary~\ref{cor:apply_VV} implies a copy complexity lower bound of $\Omega(\Max{1/\epsilon}{\norm{\sigma^{-\max}_{-\epsilon}}_{2/3}/\epsilon^2})$. We would like to relate this to \begin{equation}
		\left(\sum_{i\in\Ssing\backslash\Sjunk} \lambda^{2/3}_i\right)^{3/2} \ge (1 - 2^{-2/5})^{5/2}\cdot \left(\sum_{i\in\Ssing\backslash\Sjunk}\lambda^{2/5}_i\right)^{5/2} \ge \Omega(\norm{\sigma'}_{2/5}),\label{eq:target}
	\end{equation}
	where the first step follows by Fact~\ref{fact:geoseries}, and the last step follows by the hypothesis of the lemma.

	Suppose that there is some $i$ for which $d_{j(i)} = 1$ and $i$ is not among the indices removed in the definition of $\sigma^{-\max}_{-\epsilon}$. Then we can lower bound $\norm{\sigma^{-\max}_{-\epsilon}}_{2/3}$ by $\lambda_i$, which is at least $(1 - 2^{-2/3})^{3/2} = \Omega(1)$ times the left-hand side of \eqref{eq:target}.

	On the other hand, suppose that all $i$ for which $d_{j(i)} = 1$ are removed in the definition of $\sigma^{-\max}_{-\epsilon}$. As long as $\sigma^{-\max}_{-\epsilon}$ has some nonzero entry, call it $\lambda_{i^*}$, then $\lambda_{i^*}\ge \max_{i\in\Ssing\backslash\Sjunk}\lambda_i$, so we can similarly guarantee that $\norm{\sigma^{-\max}_{-\epsilon}}_{2/3} \ge \lambda_{i^*}$ is at least $(1 - 2^{-2/3})^{3/2} = \Omega(1)$ times the left-hand side of \eqref{eq:target}. Otherwise, we note that $\sigma'$ is zero as well, in which case we are also done.
\end{proof}

It remains to consider the primary case where the hypothesis of Lemma~\ref{lem:spiky_nonadaptive} does not hold, and this is where we will use Lemma~\ref{lem:generic_nonadaptive}. The following together with Lemma~\ref{lem:spiky_nonadaptive} will complete the proof of Lemma~\ref{lem:nonadaptive_1}:

\begin{lemma}\label{lem:nonspiky_nonadaptive}
	If $\sum_{i\in\Ssing\backslash\Sjunk} \lambda^{2/5}_i < \frac{1}{2}\norm{\sigma'}^{2/5}_{2/5}$, then state certification with respect to $\sigma$ using nonadaptive, unentangled measurements has copy complexity at least $\Omega(\norm{\sigma'}_{2/5}/(\epsilon^2\log(d/\epsilon)))$.
\end{lemma}

The proof of Lemma~\ref{lem:nonspiky_nonadaptive} requires some setup. First, obviously the hypothesis of the lemma can equivalently be stated as \begin{equation}
	\sum_{i\in\Smany\backslash\Sjunk}\lambda^{2/5}_i > \frac{1}{2}\norm{\sigma'}^{2/5}_{2/5}.\label{eq:smanybound}
\end{equation}

\begin{definition}[Choice of $\epsilon_j$]\label{defn:nonadaptive_epsjs}
	For every $i\in\Smany$, for $j\in\calJ$ the index of the bucket containing $i$, define $\epsilon_j\triangleq \Min{2^{-j-1}}{\zeta 2^{-2/3(j+1)} d_j^{2/3}}$ for normalizing quantity $\zeta$ satisfying 
	\begin{equation}\label{eq:zetadef}
		\sum_{j\in\calJ: d_j > 1} 2\floor{d_j/2}\cdot \brc*{\Min{2^{-j-1}}{\zeta 2^{-2/3(j+1)}d_j^{2/3}}} = \epsilon.
	\end{equation}
\end{definition}

Note that by ensuring that $\epsilon_j\le 2^{-j-1}$, we ensure that $\sigma + \U^{\dagger}\Eps\U$ has nonnegative spectrum, while \eqref{eq:zetadef} $\zeta$ ensures that for any $\U$ in the support of $\calD$, $\norm{\Eps}_1 = \epsilon$.

The rest of the proof is devoted to showing that for this choice of $\brc{\epsilon_j}$, the lower bound in \eqref{eq:Nbound} is at least the one in Lemma~\ref{lem:nonspiky_nonadaptive}. The main step is to upper bound the normalizing quantity $\zeta$.

\begin{lemma}\label{lem:zetabound}
	For $\zeta$ defined in Definition~\ref{defn:nonadaptive_epsjs},
	\begin{equation}
		\zeta \le O(\epsilon)\cdot \left(\sum_{j\in\calJ', i\in S_j} \lambda^{2/3}_i d^{5/3}_{j}\right)^{-1}.\label{eq:zetabound}
	\end{equation}
\end{lemma}

We will need the following elementary fact (see Appendix~\ref{subsec:factsortmix_proof} for a proof).

\begin{fact}\label{fact:sort_mix}
	Let $u_1<\cdots< u_m$ and $v_1\le\cdots\le v_n$ be numbers for which $u_{i+1} \ge 2 u_i$ for all $i$. Let $d_1,\ldots,d_n > 1$ be arbitrary integers. Let $w_1\le \cdots \le w_{m+n}$ be these numbers in sorted order. For $i\in[m+n]$, define $d^*_i$ to be 1 if $w_i$ corresponds to some $u_j$, and $d_j$ if $w_i$ corresponds to some $v_j$.

	Let $s$ be the largest index for which $\sum^s_{i=1}w_i d^*_i \le 3\epsilon$. Let $a, b$ be the largest indices for which $u_a$, $v_b$ are present among $w_1,\ldots,w_s$ (if none exists, take it to be 0). Then either $b = n$ or $\sum^{b+1}_{i=1}v_i d_i > \epsilon$.
\end{fact}

This allows us to deduce the following bound for buckets not removed in Definition~\ref{defn:remove_nonadaptive}.

\begin{corollary}\label{cor:min}
	Under the hypothesis of Lemma~\ref{lem:nonspiky_nonadaptive}, $\Smany\backslash\Sjunk$ is nonempty, and there exists an absolute constant $c>0$ such that for any $i\in\Smany\backslash\Sjunk$ in some bucket $j$, $\zeta \cdot 2^{-2/3(j+1)}d_j^{2/3} \le c\cdot 2^{-j - 1}$.
\end{corollary}

\begin{proof}
	The first part immediately follows from \eqref{eq:smanybound}. For the second part, take some constant $c$ to be optimized later and suppose to the contrary that for some $i^*\in\Smany\backslash\Sjunk$, lying in some bucket $j^*$, we have that $c\cdot 2^{-j^*-1} < \zeta \cdot 2^{-2/3(j^*+1)}d_j^{2/3}$, or equivalently $2^{-j^*-1}/d_j^{2} < \zeta^3/c^3$. Because in the definition of $\Sjunk$, we sorted by $\lambda_i/d_{j(i)^2}$, for any $i\in\Sjunk$, and because $\lambda_i\in\brk{2^{-j(i) - 1},2^{-j(i)}}$, we also have that $2^{-j(i)-1}/d^2_{j(i)} < \zeta^3/c^3$, or equivalently, $c\cdot 2^{-j(i)-1} < \zeta\cdot 2^{-2/3(j+1)}d_{j(i)}^{2/3}$. 

	So the sum on the left-hand side of \eqref{eq:zetadef} is at least \begin{equation}
		\sum_{j\in\calJ: j \ge j^*, d_j > 1} 2\floor{d_j/2}\cdot (c\cdot 2^{-j-1}) \ge \sum_{j\in\calJ: j\ge j^*, d_j > 1} (2d_j/3)\cdot (c\cdot 2^{-j-1}) \ge \sum_{i\in\Smany, i\le i^*} \lambda_i > \epsilon,
	\end{equation}
	where in the first step we used that for $d_j > 1$, $2\floor{d_j/2} \ge 2d_j/3$, in the second step we took $c = 3$ and used that $\lambda_i \le 2^{-j}$ for $i\in S_j$, and in the third step we used Fact~\ref{fact:sort_mix} applied to the numbers $\brc{u_i}\triangleq \brc{\lambda_i}_{i\in\Ssing}$, $\brc{v_i}\triangleq \brc{\lambda_i/d^2_{j(i)}}_{i\in\Smany}$ and $\brc{d_i}\triangleq \brc{d^2_{j(i)}}_{i\in\Smany}$. This contradicts \eqref{eq:zetadef}.
\end{proof}

We are finally ready to upper bound the normalizing constant $\zeta$.

\begin{proof}[Proof of Lemma~\ref{lem:zetabound}]
	We can now upper bound $\zeta$ as follows. We have \begin{align}
		\epsilon &\ge \Omega(\zeta)\cdot \sum_{j\in\calJ'} 2\floor{d_j/2}\cdot 2^{-2/3(j+1)} d_j^{2/3} \\
		&\ge \Omega(\zeta)\sum_{j\in\calJ'} 2^{-2j/3} d_j^{5/3}
	\end{align}
	where in the first step we used \eqref{eq:zetadef} and Corollary~\ref{cor:min}, and in the second step we again used the fact that for $d_j>1$, $2\floor{d_j/2} \ge 2d_j/3$. The claimed bound follows.
\end{proof}

We are now ready to complete the proof of Lemma~\ref{lem:nonspiky_nonadaptive}:

\begin{proof}
 	Substituting our choice of $\brc{\epsilon_j}$ in Definition~\ref{defn:nonadaptive_epsjs} into the lower bound of Lemma~\ref{lem:generic_nonadaptive} gives
	\begin{align}
		\left(\sum_{j\in\calJ} 2^{2j}\norm{\Eps_j}_{\op}^4/d_j\right)^{-1/2} &\ge \left(\sum_{j\in\calJ: d_j > 1} \brc*{\Min{\frac{2^{-2j-4}}{d_j}}{\zeta^4 2^{-2/3j-8/3}d_j^{5/3}}}\right)^{-1/2}  \\
		&\ge \left(\sum_{j\in\calJ: d_j > 1} \brc*{\Min{\zeta^3 2^{-j-3}d_j}{\zeta^4 2^{-2/3j}d_j^{5/3}}}\right)^{-1/2} \\
		&\ge \Omega(\zeta^{-3/2}) \left(\sum_{j\in\calJ: d_j > 1} 2\floor{d_j/2} \brc*{\Min{2^{-j-1}}{\zeta 2^{-2/3(j + 1)}d_j^{2/3}}}\right)^{-1/2} \\
		&= \Omega(\zeta^{-3/2})\cdot \epsilon^{-1/2} \\
		&\ge \epsilon^{-2}\cdot \left(\sum_{j\in\calJ', i\in S_j} \lambda^{2/3}_i d^{5/3}_{j}\right)^{3/2} \\
		&\ge \max_{j\in \calJ', i\in S_j} \lambda_i d^{5/2}_j/\epsilon^2 \\
		&\ge \left(\sum_{j\in\calJ', i\in S_j} \lambda^{2/5}_i d_{j}\right)^{5/2} \cdot \log(d/\epsilon)^{-1} \\
		&\ge \norm{\sigma'}_{2/5}\cdot \log(d/\epsilon)^{-1},
	\end{align}
	where in the second step we used that the minimum of two nonnegative numbers increases if we replace one of them by a weighted geometric mean of the two numbers, in the third step we use the fact that $\floor{d_j/2}$ and $d_j$ are equivalent up to constant factors if $d_j > 1$, in the fourth step we use \eqref{eq:zetadef}, in the fifth step we use \eqref{eq:zetabound}, in the penultimate step we used Fact~\ref{fact:fewbuckets}, and in the last step we used \eqref{eq:smanybound} and the fact that for any $j$, there are at most $d_{j}$ indices $i\in\Smany\backslash\Sjunk$ within bucket $S_j$.

\end{proof}

With Lemma~\ref{lem:nonspiky_nonadaptive} in place, we conclude the proof of the main lemma of this subsection:

\begin{proof}[Proof of Lemma~\ref{lem:nonadaptive_1}]
	This follows immediately from Lemmas~\ref{lem:spiky_nonadaptive} and \ref{lem:nonspiky_nonadaptive}.
\end{proof}

\subsection{Lower Bound Instance II: Perturbing Off-Diagonals}
\label{subsec:instance_instance_2}

In many cases, the following lower bound instance will yield a stronger lower bound than the preceding argument, at the cost of applying to a limited range of $\epsilon$. Take any $j,j'\in\calJ^*$ for which $d_j \ge d_{j'}$. As we will explain below, if $d_j > 1$, then $j$ and $j'$ need not be distinct.

If $j$ and $j'$ are distinct, then given a matrix $\W^{d_j\times d_{j'}}$ with orthonormal columns, let $\sigma_{\W}$ be the matrix $\sigma + D_{\W}$ where $D_{\W}\in\co^{d\times d}$ is the matrix which is zero outside of the principal submatrix indexed by $S_j\cup S_{j'}$ and which is equal to the matrix 
\begin{gather}
	\centering
	\left(
	\begin{array}{c|c}
		\vec{0}_{d_j} & (\epsilon/2d_{j'})\cdot \W \\ \hline
		(\epsilon/2d_{j'})\cdot \W^{\dagger} & \vec{0}_{d_{j'}}.
	\end{array}
	\right)\label{eq:offdiag_1}
\end{gather}
On the other hand, if $j = j'$ and $d_j > 1$, then partition $S_j$ into contiguous sets $S^1_j, S^2_j$ of size $\ceil{d_j/2}$ and $\floor{d_j/2}$, and given a matrix $\W^{\ceil{d_j/2}\times\floor{d_j/2}}$ with orthonormal columns, define $D_{\W}\in\co^{d\times d}$ to be the matrix which is zero outside the principal submatrix indexed by $S^1_j\times S^2_j$ and which is equal to the matrix
\begin{gather}
	\centering
	\left(
	\begin{array}{c|c}
		\vec{0}_{\ceil{d_j/2}} & (\epsilon/2\floor{d_{j}/2})\cdot \W \\ \hline
		(\epsilon/2\floor{d_j/2})\cdot \W^{\dagger} & \vec{0}_{d_{\floor{d_j/2}}}.
	\end{array}
	\right)\label{eq:offdiag_2}
\end{gather}
In the rest of this subsection, we will consider the case where $j\neq j'$, but as will become evident, all of the following arguments easily extend to the construction for $j = j'$ when $d_j > 1$ by replacing $S_j$ and $S_{j'}$ with $S^1_j$ and $S^2_j$ respectively.

\begin{lemma}\label{lem:rhoW_facts}
	If $\epsilon \le d_{j'}\cdot 2^{-j/2-j'/2}$, then $\norm{\sigma - \sigma_{\W}}_1 \ge \epsilon$ and $\sigma_{\W}$ is a density matrix.
\end{lemma}

\begin{proof}
	For the first part, note that \begin{equation}
		\norm{\sigma - \sigma_{\W}}_1 = \norm{D_{\W}} = 2\cdot(\epsilon/2d_{j'})\norm{\W}_1 = \epsilon,
	\end{equation}
	where in the second equality we used that $D_{\W}$ is the Hermitian dilation of $(\epsilon/d_{j'})\cdot \W$, and in the last equality we used the fact that $\W$ consists of $d_{j'}$ orthogonal columns.

	For the second part, first note that regardless of the choice of $\epsilon$, we have that $\Tr(D_{\W}) = 0$, so $\Tr(\sigma_{\W}) = 1$. Finally, to verify that $\sigma_{\W}$ is positive definite, note that the Schur complement of the principal submatrix of $\sigma_{\W}$ indexed by $S_j\cap S_{j'}$ is given by \begin{equation}
		\sigma_{j'} - \frac{\epsilon^2}{4d_{j'}^2}\sigma^{-1}_{j} \succeq 2^{-j'-1}\Id - \frac{\epsilon^2}{4d_{j'}^2}2^{j+1} \Id,
	\end{equation} which is positive definite provided that $\epsilon \le d_{j'} \cdot 2^{-j/2-j'/2}$. It follows by Lemma~\ref{lem:schur} that $\sigma_{\W}$ is positive definite as claimed.
\end{proof}

The objective of this subsection is to show the following lower bound:

\begin{lemma}\label{lem:nonadaptive_2}
	Fix any $j,j'\in\calJ^*$ satisfying $d_j \ge d_{j'}$. If $d_j > 1$, then we can optionally take $j = j'$. Suppose $\epsilon \le d_{j'}\cdot 2^{-j/2 - j'/2}$. Let $\sigma\in\co^{d\times d}$ be a diagonal density matrix. Distinguishing between whether $\rho = \sigma$ or $\rho = \sigma_{\W}$ for $\W\in\co^{d_j\times d_{j'}}$ consisting of Haar-random orthonormal columns, using nonadaptive unentangled measurements, has copy complexity at least \begin{equation}
		\Omega\left(\frac{\sqrt{d_j} \cdot d_{j'}^2 \cdot 2^{-j'}}{\epsilon^2}\right).
	\end{equation}
\end{lemma}

Note that a random $\W$ is equivalent to $\U\Pi$ for $\U\sim\calD$, where $\calD$ is the Haar measure over $U(d_j)$, and \begin{equation}
	\Pi\triangleq (\Id_{d_{j'}}|\vec{0}_{d_j - d_{j'}})^{\top},
\end{equation}
so we can just as well parametrize $\brc{\sigma_{\W}}$ as $\brc{\sigma_{\U}}$, which we will do in the sequel.

Take any single-copy sub-problem $\calP = (\calM, \sigma, \brc{\sigma_{\U}}_{\U\sim\calD})$ where POVM $\calM$ consists of elements $\brc{M_z}$. Analogously to Lemma~\ref{lem:assume_POVM}, we may without loss of generality assume that one of the POVM elements is the projector to the coordinates outside of $S_j\cup S_{j'}$, and the remaining POVM elements are rank-1 matrices $M_z = \lambda_z v_zv_z^{\dagger}$ where the $\lambda_z \le 1$ satisfy
\begin{equation}
\sum \lambda_z = d_j + d_{j'} < 2d_j \label{eq:lamsumbound}
\end{equation}
and the vectors $v_z$ are unit vectors supported on $S_j\cap S_{j'}$. Let $v^j_z$ and $v^{j'}_z$ denote the $d_j$- and $d_{j'}$-dimensional components of $v_z$ indexed by $S_j$ and $S_{j'}$. Note that for these $z$, \begin{equation}
	g^{\U}_{\calP}(z) = \frac{\iprod{M_z,D_{\W}}}{\iprod{M_z,\sigma}} = \frac{\epsilon}{d_{j'}}\cdot\frac{\Re((v^j_z)^{\dagger}(\U\Pi) v^{j'}_z)}{v_z^{\dagger} \sigma v_z}.\label{eq:gcorner}
\end{equation}
while for the index $z$ corresponding to the projector to $(S_j\cup S_{j'})^c$, $g^{\U}_{\calP}(z) = 0$.

In the next three lemmas, we verify that $\calP$ satisfies Assumption~\ref{assume:main}.

\begin{lemma}
	For any $z$, $\E[\U]{g^{\U}_{\calP}(z)} = 0$, so Condition~\ref{cond:gexp} of Assumption~\ref{assume:main} holds.
\end{lemma}

\begin{proof}
	Clearly $\Tr(D_{\W}) = 0$, so by the second part of Lemma~\ref{lem:collins}, $\E[\W]{g^{\W}(z)} = 0$.
\end{proof}

\begin{lemma}\label{lem:off_diagonal_g2}
	$\E[z,\U]{g^{\U}_{\calP}(z)^2} \le O\left(\frac{\epsilon^2}{d_{j'}^2 2^{-j'}}\right)$, where as usual, expectation is with respect to measurement outcomes when measuring the null hypothesis $\sigma$ with $\calM$, so Condition~\ref{cond:secondmoment} of Assumption~\ref{assume:main} holds.
\end{lemma}

\begin{proof}
	From \eqref{eq:gcorner} we have that \begin{align}
		\E[z,\U]*{g^{\U}_{\calP}(z)^2} &= \frac{\epsilon^2}{d^2_{j'}}\E[\U]*{\sum_z \lambda_z v^{\dagger}_z \sigma v_z \left(\frac{\Re((v^j_z)^{\dagger}(\U\Pi) v^{j'}_z)}{v_z^{\dagger} \sigma v_z}\right)^2} \\
		&= \frac{\epsilon^2}{d^2_{j'}}\sum_{z}\frac{\lambda_z}{v^{\dagger}_z \sigma v_z}\E[\U]*{\left(\Re((v^j_z)^{\dagger}(\U\Pi) v^{j'}_z)\right)^2} \\
		&= \frac{\epsilon^2}{d^2_{j'}}\sum_{z}\frac{\lambda_z}{v^{\dagger}_z \sigma v_z} \cdot\frac{\norm{v^j_z}^2\norm{v^{j'}_z}^2}{d_j},\label{eq:gcornerbound}
		\intertext{As $v_z$ is supported on $S_j\cup S_{j'}$, the supports of $v^j_z$ and $v^{j'}_z$ are disjoint, and the diagonal entries of $\sigma$ indexed by $S_{j'}$ are at least $2^{-j-1}$, we have that $v^{\dagger}_z \sigma v_z \ge 2^{-j'-1}\norm{v^{j'}_z}^2$ and $\norm{v^j_z}^2_2 \le 1$, so we can further bound \eqref{eq:gcornerbound} by}
		&= \frac{\epsilon^2 2^{j'+1}}{d^2_{j'}d_j}\sum_z \lambda_z \le O\left(\frac{\epsilon^2}{d^2_{j'} 2^{-j'}}\right),
	\end{align} where the last step follows by \eqref{eq:lamsumbound}.
\end{proof}

\begin{lemma}\label{lem:off_diagonal_lip}
	$\E[z]{(g^{\U_1}_{\calP}(z) - g^{\U_2}_{\calP}(z))^2} \le O\left(\frac{\epsilon^2}{d_{j'}^2 2^{-j}}\right)\cdot \norm{\U_1 - \U_2}^2_{\HS}$ for any $\U_1,\U_2\in U(d_j)$, so Condition~\ref{cond:normequiv} of Assumption~\ref{assume:main} holds.
\end{lemma}

\begin{proof}
	Define the matrix \begin{equation}
		\vec{D} = \begin{pmatrix}
			\vec{0}_{d_j} & (\epsilon/2d_{j'})\cdot (\U_1\Pi - \U_2\Pi) \\
			(\epsilon/2d_{j'})\cdot (\U_1\Pi - \U_2\Pi)^{\dagger} & \vec{0}_{d_{j'}}
		\end{pmatrix}
	\end{equation}
	Note that for any POVM element $M_z$, \begin{equation}
		\iprod{M_z,\vec{D}}^2 = \frac{\lambda^2_z\epsilon^2}{d^2_{j'}}\Re\left((v^j_z)^{\dagger}(\U_1 - \U_2)\Pi v^{j'}_z\right)^2 \le \frac{\lambda^2_z\epsilon^2}{d^2_{j'}}\cdot\norm{v^j_z(\U_1 - \U_2)}^2\cdot \norm{v^{j'}_z}^2_2 \label{eq:Dbound}
	\end{equation}

	We can then write \begin{align}
		\E[z]{(g^{\U}_{\calP}(z) - g^{\V}_{\calP}(z))^2} &= \sum_z \frac{\iprod{M_z, \vec{D}}^2}{\iprod{M_z,\sigma}} \\
		&\le \frac{\epsilon^2}{d^2_{j'}}\sum_z \frac{\lambda_z\norm{v^j_z(\U_1 - \U_2)}^2\cdot \norm{v^{j'}_z}^2}{2^{-j'-1}\norm{v^{j'}_z}^2} \\
		&\le O\left(\frac{\epsilon^2 2^{j'}}{d^2_{j'}}\right)\cdot \sum_z \lambda_z\norm{v^j_z(\U_1 - \U_2)}^2 \\
		&= O\left(\frac{\epsilon^2 2^{j'}}{d^2_{j'}}\right)\cdot \iprod*{(\U_1 - \U_2)(\U_1 - \U_2)^{\dagger},\sum_z \lambda_z v^j_z(v^j_z)^{\dagger}} \\
		&= O\left(\frac{\epsilon^2}{d_{j'}^2 2^{-j'}}\right)\cdot \norm{\U_1 - \U_2}^2_{\HS},
	\end{align}
	where in the second step we used \eqref{eq:Dbound} and the fact that $\iprod{M_z,\sigma} = \lambda_z v_z^{\dagger}\sigma v_z \ge \lambda_z 2^{-j'-1}\norm{v^{j'}_z}^2$, and in the fifth step we used that $\sum_z\lambda_z v^j_z(v^j_z)^{\dagger} = \Id_{d_j}$.
\end{proof}

We can finally complete the proof of Lemma~\ref{lem:nonadaptive_2}:

\begin{proof}[Proof of Lemma~\ref{lem:nonadaptive_2}]
	As the mixture of alternatives in $\calP$ is parametrized by $\U\sim\calD$ for $\calD$ the Haar measure over the unitary group, the lemma immediately follows from Lemma~\ref{lem:ingster_with_conds} with $L, \varsigma = O\left(\frac{\epsilon}{d_{j'} 2^{-j'/2}}\right)$.
\end{proof}

\subsection{Lower Bound Instance III: Corner Case}
\label{subsec:case3}

We will also need the a lower bound instance that will yield an $\Omega(1/\epsilon^2)$ lower bound for state certification with respect to any $\sigma$ with maximum entry at least $1/2$. We will not use anything about bucketing in this warmup result.

Let $i_1$ be the index of the largest entry of $\sigma$, and let $i_2$ be the index of the second-largest (breaking ties arbitrarily). For any $u\in\brc{\pm 1}$, consider the state $\sigma^u$ which agrees with $\sigma$ everywhere except in the principal submatrix indexed by $\brc{i_1,i_2}$. Within that submatrix, define $\sigma^u_{i_1,i_1} = \sigma_i - \epsilon^2/4$, $\sigma^u_{i_2,i_2} = \sigma_{i_2} + \epsilon^2/4$, and $\sigma^u_{i_1,i_2} = {\sigma^u}^{\dagger}_{i_2,i_1} = (\epsilon/2) u$.

\begin{lemma}
	If the maximum entry of $\sigma$ is at least $3/4$, then for any $\epsilon \le 1/2$, $\norm{\sigma - \sigma^u}_1 \ge \epsilon$ and $\sigma^u$ is a density matrix.
\end{lemma}

\begin{proof}
	Note that for $\epsilon < 1/2$, \begin{equation}
		\norm{\sigma - \sigma^u}_1 = \norm*{\begin{pmatrix}
			-\epsilon^2 & (\epsilon/2)u \\
			(\epsilon/2)\overline{u} & \epsilon^2
		\end{pmatrix}}_1 = 2\sqrt{\epsilon^4/16 + \epsilon^2/4} \ge \epsilon.
	\end{equation} For the second part of the lemma, clearly $\Tr(\sigma^u) = 1$. To verify that $\sigma^u$ is psd, first note that because $\sigma_{i_1,i_1} \ge 3/4$ and $\sigma_{i_2,i_2} \le 1/2$, and $\epsilon^2/4 \le 1/4$, every diagonal entry of $\sigma^u$ is nonnegative. On the other hand, the principal submatrix indexed by $\brc{i_1,i_2}$ has determinant $(\sigma_{i_1,i_1} - \epsilon^2)(\sigma_{i_2,i_2} + \epsilon^2) - \epsilon^2/4 \ge (3/4 - \epsilon^2)\epsilon^2 - \epsilon^2/4 \ge 0$, so $\sigma^u$ is psd as claimed.
\end{proof}

The objective of this subsection is to show the following lower bound:

\begin{lemma}\label{lem:corner}
	Let $\epsilon\le 1/2$. If the maximum entry of $\sigma$ is at least $3/4$, then distinguishing between whether $\rho = \sigma$ or $\rho = \sigma^u$ for $u\sim\brc{\pm 1}$, using nonadaptive unentangled measurements, has copy complexity at least $\Omega(1/\epsilon^2)$. In fact, this holds even for adaptive unentangled measurements.
\end{lemma}

Because we have no a priori bound on $\sigma_{i_2,i_2}$, the KL divergence between the distribution over outcomes from measuring $N$ copies of $\sigma^u$ for random $u\in\brc{\pm 1}$ and the distribution from measuring $N$ copies of $\sigma$ may be arbitrarily large, so we cannot implement the strategy in Section~\ref{sec:framework}. Instead, we will directly upper bound the total variation between these two distributions using the following basic fact:

\begin{fact}
	Given distributions $p,q$ over a discrete domain $S$, if likelihood ratio $p(x)/q(x) \ge 1 - \nu$, then $d_{\TV}(p,q) \le \nu$.
\end{fact}

\begin{proof}
	We can write \begin{equation}
		d_{\TV}(p,q) = \sum_{x: p(x) \le q(x)} |p(x) - q(x)| = \sum_{x: p(x) \le q(x)} q(x)\cdot |p(x)/q(x) - 1| \le \nu
	\end{equation} as claimed.
\end{proof}

\begin{proof}[Proof of Lemma~\ref{lem:corner}]
	Let $\calD$ be the uniform distribution over $\brc{\pm 1}$, and fix an arbitrary unentangled POVM schedule $\calS$. Let $p_0$ denote the distribution over transcripts $z_{\le t}$ of outcomes upon measuring $N$ copies of $\sigma$ with $\calS$, and let $p_1$ denote the distribution upon measuring $N$ copies of $\sigma^u$, where $u\sim\calD$. We will lower bound the likelihood ratio $p_1(z_{\le N})/p_0(z_{\le N})$ for \emph{any} transcript $z_{\le N}$. Let $\calM^{(1)},\ldots,\calM^{(N)}$ denote the (possibly adaptively chosen) POVMs that were used in the course of generating $z_{\le N}$.

	For any $t\in[N]$, suppose $\calM^{(t)}$ consists of elements $\brc{M^{(t)}_z}$. Analogously to Lemma~\ref{lem:assume_POVM}, we may without loss of generality assume that one element of $\calM^{(t)}$ is the projector to the coordinates outside of $\brc{i_1,i_2}$, and the remaining elements are rank-1 matrices $M^{(t)}_z = \lambda^{(t)}_z v^{(t)}_z(v^{(t)}_z)^{\dagger}$ where the $\lambda^{(t)}_z \le 1$ satisfy $\sum \lambda^{(t)}_z = 2$ and the vectors $v^{(t)}_z$ are unit vectors supported on $\brc{i_1,i_2}$. Let $v^{(t)}_{z_t,1}$ and $v^{(t)}_{z_t,2}$ denote the coordinates of $v^{(t)}_z$ indexed by $i_1$ and $i_2$.

	Note that for any $u\in\brc{\pm 1}$ and $t\in[N]$, if $z_t$ does not correspond to the projector to the coordinates outside of $\brc{i_1,i_2}$, we can write \begin{equation}
		\Delta^u_t(z_t) \triangleq \frac{\iprod{M^{(t)}_{z_t},\sigma^u}}{\iprod{M^{(t)}_{z_t},\sigma}} = 1 + \frac{\epsilon u \Re\left(\overline{{v^{(t)}_{z_t,1}}} {v^{(t)}_{z_t,2}}\right) -\epsilon^2 \left(\abs*{v^{(t)}_{z_t,1}}^2 - \abs*{v^{(t)}_{z_t,2}}^2\right)}{v^{(t)\dagger}_{z_t}\sigma v^{(t)}_{z_t}}
	\end{equation} and if $z_t$ does correspond to the projector, then $\Delta^u_t(z_t) = 1$.

	Denoting the $t$-th entry of $z_{\le N}$ by $z_t$, we can use AM-GM to bound the likelihood ratio by \begin{align}
		\frac{p_1(z_{\le N})}{p_0(z_{\le N})} &= \E[u]*{\prod^N_{t=1} \Delta^u_t(z_t)} \\
		&\ge \left(\prod^N_{t=1} \Delta^{+1}_t(z_t) \Delta^{-1}_t(z_t)\right)^{1/2} \label{eq:AMGM}
	\intertext{To prove the lemma, we will lower bound this by $1 - o(1)$. Because $\Delta^u_t(z_t) = 1$ if $z_t$ corresponds to the projector to the coordinates outside of $\brc{i_1,i_2}$, we may assume without loss of generality that this is not the case for any $t\in[N]$. We can then further bound \eqref{eq:AMGM} by} 
		&\ge \prod^N_{t=1}\brc*{\left(1 - \frac{\epsilon^2 \left(\abs*{v^{(t)}_{z_t,1}}^2 - \abs*{v^{(t)}_{z_t,2}}^2\right)}{v^{(t)\dagger}_{z_t}\sigma v^{(t)}_{z_t}}\right)^2 - \frac{\epsilon^2 \Re\left(\overline{{v^{(t)}_{z_t,1}}} {v^{(t)}_{z_t,2}}\right)^2}{\left(v^{(t)\dagger}_{z_t}\sigma v^{(t)}_{z_t}\right)^2}}^{1/2}.\label{eq:diffsquares}
	\end{align}
	For any $v\in\co^d$ which has entries $v_1$ and $v_2$ in coordinates $i_1$ and $i_2$ and is zero elsewhere, we have that \begin{equation}
		\frac{|v_1|^2 - |v_2|^2}{v^{\dagger}\sigma v} \le \frac{|v_1|^2}{\sigma_{i_1,i_1}|v_1|^2} \le 4/3 \qquad
		\frac{\Re(\overline{v_1}v_2)^2}{v^{\dagger}\sigma v} \le \frac{\Re(\overline{v_1}v_2)^2}{\sigma_{i_1,i_1} \abs{v_1}^2} \le 4/3,
	\end{equation}
	where the last step for both estimates follows by the assumed lower bound on $\sigma_{i_1,i_1}$. By \eqref{eq:diffsquares} we have that \begin{equation}
		\frac{p_1(z_{\le N})}{p_0(z_{\le N})} \ge ((1 - 4\epsilon^2/3)^2 - 4\epsilon^2/3)^{N/2} \ge (1 - 32\epsilon^2/9)^{N/2}.
	\end{equation} In particular, for $N = o(1/\epsilon^2)$, the likelihood ratio is at least $1 - o(1)$ as desired.
\end{proof}

\subsection{Putting Everything Together}
\label{subsec:conclude_nonadaptive}

We are now ready to conclude the proof of Theorem~\ref{thm:cert_lower_main}.

\begin{proof}[Proof of Theorem~\ref{thm:cert_lower_main}]
	We proceed by casework depending on whether or not $d_{j} = 1$ for all $j\in\calJ^*$.

	\begin{case}
		$d_j = 1$ for all $j\in\calJ^*$.
	\end{case}

	There are two possibilities. If there is a single bucket $j = j(i)$ for which $i\not\in\Sjunk\cup\Slight$, then $d_{\mathsf{eff}} = 1$ and $\norm{\sigma^{**}}_{1/2} = O(1)$. For $\epsilon$ smaller than some absolute constant, we know that $\sigma_{i,i} \ge 3/4$ and can apply Lemma~\ref{lem:corner} to conclude a lower bound of $\Omega(1/\epsilon^2)$ as desired. Otherwise, let $j'$ be the smallest index for which $j' = j(i')$ for some $i'\in\calJ^*$, and let $j> j'$ be the next smallest index for which $j = j(i)$ for some $i\in\calJ^*$. Consider the lower bound instance in Section~\ref{subsec:instance_instance_2} applied to this choice of $j,j'$. Provided that $\epsilon \le 2^{-j/2-j'/2}$, we would obtain a copy complexity lower bound of $\Omega(2^{-j'}/\epsilon^2) \ge \Omega(\norm{\sigma^*}_{1/2}/(\epsilon^2\log(d/\epsilon)))$, where the inequality is by Fact~\ref{fact:optimize}, and we would be done. On the other hand, if $\epsilon \ge 2^{-j/2 - j'/2}$, then because $2^{-j'} > 2^{-j}$, we would conclude that $2^{-j} \le \epsilon$. In particular, this implies that $\sum_{j''\in\calJ^*, i\in S_{j''}: j''\neq j'} \lambda_i \le 2\epsilon$, so after removing at most an additional $2\epsilon$ mass from $\sigma^*$, we get a matrix $\sigma^{**}$ (see Definition~\ref{defn:remove_nonadaptive}) with a single nonzero entry. Again, $d_{\mathsf{eff}} = 1$ and $\norm{\sigma^{**}}_{1/2} = O(1)$, and if $\epsilon$ is smaller than some absolute constant, we conclude that that single nonzero entry is at least 3/4 and can apply Lemma~\ref{lem:corner} to conclude a lower bound of $\Omega(1/\epsilon^2)$ as desired.

	\begin{case}
		$d_j > 1$ for some $j\in\calJ^*$.
	\end{case}


	Let $j_*\triangleq \arg\max_{j\in\calJ^*}d_j$ and $j'_*\triangleq \arg\max_{j\in\calJ^*} d^2_j 2^{-j}$. By Lemma~\ref{lem:nonadaptive_2}, we have a lower bound of $\Omega\left(\sqrt{d_{j_*}} \cdot d_{j'_*}^2 \cdot 2^{-j'_*}/\epsilon^2\right)$ as long as $\epsilon$ satisfies the bound \begin{equation}
		\epsilon \le d_{j'_*}\cdot 2^{-j_*/2 - j'_*/2}. \label{eq:epsbound}
	\end{equation} Note that because $d_{j_*} > 1$ as we are in Case 2, we do not constrain $j_*,j'_*$ to be distinct necessarily. We would now like to argue that this lower bound, up to log factors, holds even if the bound on $\epsilon$ in \eqref{eq:epsbound} does not hold. In the following, assume that \eqref{eq:epsbound} does not hold.

	To this end, we will also use the lower bound from Lemma~\ref{lem:nonadaptive_1} of $\Omega(\norm{\sigma'}_{2/5}/(\epsilon^2\log(d/\epsilon)))$. We would first like to relate $\norm{\sigma'}_{2/5}$ to $\norm{\sigma^*}_{2/5}$.

	\begin{lemma}\label{lem:eitheror}
		Either $\norm{\sigma'}_{2/5} \ge \Omega(\norm{\sigma^*}_{2/5})$, or the following holds. Let $j^{\circ}$ be the index maximizing $d^{5/2}_j 2^{-j}$. Then 1) $j^{\circ} = \min_{j\in\calJ^*} j$, 2) $d_{j^{\circ}} = 1$, and 3) $j^{\circ} = 0$.
	\end{lemma}

	\begin{proof}
		We will assume that $\norm{\sigma'}_{2/5} = o(\norm{\sigma^*}_{2/5})$ and show that 1), 2), and 3) must hold. Let $j^{\circ}$ be the index maximizing $d^{5/2}_j 2^{-j}$, and let $i_{\max}$ be the index of the top entry of $\sigma^*$. Let $\sigma''$ denote the matrix obtained by zeroing out the top entry of $\sigma^*$. Note that the nonzero entries of $\sigma'$ comprise a superset of those of $\sigma''$, so \begin{equation}
			\frac{\norm{\sigma^*}^{2/5}_{2/5}}{\norm{\sigma'}^{2/5}_{2/5}} \le \frac{\norm{\sigma^*}^{2/5}_{2/5}}{\norm{\sigma''}^{2/5}_{2/5}} = \frac{\sum_{i\in\calJ^*}\sigma^{2/5}_i}{\sum_{i\in\calJ^*\backslash\brc{i_{\max}}}\sigma^{2/5}_i}.
		\end{equation}

		Suppose 1) does not hold. Then \begin{equation}
			\frac{\sum_{i\in\calJ^*}\sigma^{2/5}_i}{\sum_{i\in\calJ^*\backslash\brc{i_{\max}}}\sigma^{2/5}_i} \le \frac{\sigma^{2/5}_{i_{\max}} + \sum_{i\in S_{j^{\circ}}}\sigma^{2/5}_i}{\sum_{i\in S_{j^{\circ}}}\sigma^{2/5}_i} \le 2,
		\end{equation} where the first inequality follows by the elementary fact that for positive integers $a\ge b$ and $c$, $\frac{a+c}{b+c} \le \frac{a}{b}$, and the second inequality follows by the definition of $j^{\circ}$.

		Next, suppose 1) holds but 2) does not hold. Then \begin{equation}
			\frac{\sum_{i\in\calJ^*}\lambda_i^{2/5}}{\sum_{i\in\calJ^*\backslash\brc{i_{\max}}} \lambda_i^{2/5}} \le \frac{\sum_{i\in S_{j^{\circ}}} \lambda_i^{2/5}}{\sum_{i\in S_{j^{\circ}}\backslash\brc{i_{\max}}}\lambda_{i}^{2/5}} \le O(1),
		\end{equation} where the first inequality again uses the above elementary fact, the second inequality follows by our assumption that 2) does not hold. This yields a contradiction.

		Finally suppose 1) and 2) hold, but 3) does not, so that $\norm{\sigma^*}_{\infty} \le 1/2$. Let $\sigma''$ denote the matrix obtained by zeroing out the top entry of $\sigma^*$. We would have \begin{equation}
			\norm{\sigma''}_{2/5} \ge \norm{\sigma''} \ge 1/2 - O(\epsilon),
		\end{equation} so for $\epsilon$ smaller than a sufficiently large absolute constant, we would have that $\norm{\sigma''}^{2/5}_{2/5} \ge \Omega(\norm{\sigma^*}^{2/5}_{\infty})$ and therefore $\norm{\sigma'}_{2/5} \ge \norm{\sigma''}_{2/5} \ge \Omega(\norm{\sigma^*}_{2/5})$, a contradiction.
	\end{proof}


	Suppose the latter scenario in Lemma~\ref{lem:eitheror} happens but the former does not. In this case, because $d_{j^{\circ}} = 1$, we also have that $j'_* = \arg\max_{j\in\calJ^*} d^2_j 2^{-j}$, i.e. $j'_* = j^{\circ}$. In particular, \begin{equation}
		1 \ge d^2_{j'_*} 2^{-j'_*} \ge d^2_{j_*} 2^{-j_*} \ge \Omega(d^{3/2}_{j_*} \epsilon/\log(d/\epsilon)), \label{eq:epsle}
	\end{equation} where the last inequality follows by the fact that $d_j 2^{-j} \ge \Omega(\epsilon/\log(d/\epsilon))$ for all $j\in\calJ^*$ by design. We conclude that $\epsilon \le O(d^{-3/2}_{j_*}\log(d/\epsilon))$. But recall that we are assuming that \eqref{eq:epsbound} is violated, i.e. that \begin{equation}
		\epsilon > d_{j'_*}\cdot 2^{-j_*/2 - j'_*/2} = 2^{-j_*/2 - j'_*/2} \ge \Omega(\epsilon/(d_{j_*} \log(d/\epsilon)))^{1/2}, \label{eq:epsge}
	\end{equation} where the last step is by 3) in Lemma~\ref{lem:eitheror} and the fact that $d_j 2^{-j} \ge \Omega(\epsilon/\log(d/\epsilon))$ for all $j\in\calJ^*$. Combining \eqref{eq:epsle} and \eqref{eq:epsge}, we get a contradiction of the assumption that the former scenario in Lemma~\ref{lem:eitheror} does not hold, unless $d_{j_*} \le \polylog(d/\epsilon)$. But if $d_{j_*} \le \polylog(d/\epsilon)$, then the lower bound claimed in Theorem~\ref{thm:cert_lower_main} still holds as $d_{\mathsf{eff}} \le O(\log(d/\epsilon)\cdot d_{j_*}) \le \polylog(d/\epsilon)$.

	Finally, suppose instead that the former scenario in Lemma~\ref{lem:eitheror} happens, so that Lemma~\ref{lem:nonadaptive_1} gives a lower bound of $\Omega(\norm{\sigma^*}_{2/5}/(\epsilon^2\log(d/\epsilon)))$. Let $j^{\circ}$ still be as defined in Lemma~\ref{lem:eitheror}.

	Now we would certainly be done if this lower bound were, up to log factors, larger than the one guaranteed by Lemma~\ref{lem:nonadaptive_2} to begin with. So suppose to the contrary. We would get that \begin{equation}
		d^{5/2}_{j_*} 2^{-j_*} \le d^{5/2}_{j^{\circ}} 2^{-j^{\circ}} \le \frac{1}{\log^2(d/\epsilon)}\sqrt{d_{j_*}} d^2_{j'_*}\cdot 2^{-j'_*},
	\end{equation} implying that \begin{equation}
		d^2_j 2^{-j_*} \le \frac{1}{\log^2(d/\epsilon)} d^2_{j'_*} 2^{-j'_*}. \label{eq:implication}
	\end{equation}
	If \eqref{eq:epsbound} does not hold, then \begin{equation}
		\frac{1}{\log(d/\epsilon)}\cdot d_{j'_*}\cdot 2^{-j_*/2 - j'_*/2} \le \frac{\epsilon}{\log(d/\epsilon)} \le d_j2^{-j},
	\end{equation} where in the last step we again used the fact that $d_{j} 2^{-j} > \epsilon/\log(d/\epsilon)$ for all $j\in\calJ^*$, yielding the desired contradiction with \eqref{eq:implication} upon rearranging.

	Having lifted the constraint \eqref{eq:epsbound}, we finally note that by Fact~\ref{fact:optimize}, \begin{equation}
		\Omega\left({\sqrt{d_{j_*}} \cdot d_{j'_*}^2 \cdot 2^{-j'_*}}/{\epsilon^2}\right) \ge \Omega\left(\sqrt{d_{\mathsf{eff}}}\cdot\norm{\sigma^*}_{1/2}/(\epsilon^2\polylog(d/\epsilon))\right).
	\end{equation}
	The proof is complete upon invoking Fact~\ref{fact:fidelity} below.
\end{proof}

\begin{fact}\label{fact:fidelity}
	Given psd matrix $\sigma\in\co^{d\times d}$, let $\wh{\sigma}\triangleq \sigma/\Tr(\sigma)$. Then $\norm{\sigma}_{1/2} = d\Tr(\sigma)^2\cdot F(\wh{\sigma},\rhomm)$.
\end{fact}

\begin{proof}
	We may assumed without loss of generality that $\sigma$ is diagonal. By definition \begin{equation}
		F(\wh{\sigma},\rhomm) = \left(\Tr\sqrt{\sqrt{\wh{\sigma}} (\Id/d)\sqrt{\wh{\sigma}}}\right)^2 = \left(\frac{1}{\sqrt{d}\Tr(\sigma)}\cdot\Tr(\sqrt{\sigma})\right)^2 = \frac{1}{d\Tr(\sigma)^2}\cdot \norm{\sigma}_{1/2},
	\end{equation} from which the claim follows.
\end{proof}


\section{State Certification Algorithm}
\label{sec:alg}

In this section we prove the following upper bound on state certification that nearly matches the lower bound proven in Section~\ref{sec:instance}:

\begin{theorem}\label{thm:cert_upper_main}
	Fix $\epsilon,\delta > 0$. Let $\rho\in\co^{d\times d}$ be an unknown mixed state, and let $\sigma\in\co^{d\times d}$ be a diagonal density matrix. Let $\sigma'$ be the matrix given by zeroing out the bottom $O(\epsilon^2)$ mass in $\sigma$ (see Definition~\ref{defn:remove_nonadaptive_upper} below). Let $\wh{\sigma}'\triangleq \sigma'/\Tr(\sigma')$ and let $d_{\mathsf{eff}}$ be the number of nonzero entries of $\sigma'$.


	Given an explicit description of $\sigma$ and copy access to $\rho$, {\sc Certify} takes \begin{equation}N = O(d\sqrt{d_{\mathsf{eff}}}\cdot F(\wh{\sigma}',\rhomm)\polylog(d/\epsilon)\log(1/\delta)/\epsilon^2)\end{equation} copies of $\rho$ and, using unentangled nonadaptive measurements, distinguishes between $\rho = \sigma$ and $\norm{\rho - \sigma}_1 > \epsilon$ with probability at least $1 - \delta$.
\end{theorem}

First, in Section~\ref{subsec:generic_cert} we give a generic algorithm for state certification based on measuring in a Haar-random basis and applying classical identity testing. In Section~\ref{subsec:bucket_upper}, we describe a bucketing scheme that will be essential to the core of our analysis in Section~\ref{subsec:full_cert}, where we use this tool to obtain the algorithm in Theorem~\ref{thm:cert_upper_main}.

\subsection{Simple Subroutine}
\label{subsec:generic_cert}

The main result of this section is a basic state certification algorithm that will be invoked as a subroutine in our instance-near-optimal certification algorithm:

\begin{lemma}\label{lem:basic_id_tester}
	Fix $\epsilon,\delta > 0$. Let $\rho,\sigma\in\co^{d\times d}$ be two mixed states. Given access to an explicit description of $\sigma$ and copy access to $\rho$, {\sc BasicCertify} takes $N = O(\sqrt{d}\log(1/\delta)/\epsilon^2)$ copies of $\rho$ and, using unentangled nonadaptive measurements, distinguishes between $\rho = \sigma$ and $\norm{\rho - \sigma}_{\HS} > \epsilon$ with probability at least $1 - \delta$.
\end{lemma}

\begin{algorithm}
\DontPrintSemicolon
\caption{\textsc{BasicCertify}($\rho,\sigma,\epsilon,\delta$)}
\label{alg:basic}
	\KwIn{Copy access to $\rho$, diagonal density matrix $\sigma$, error $\epsilon$, failure probability $\delta$}
	\KwOut{$\mathsf{YES}$ if $\rho = \sigma$, $\mathsf{NO}$ if $\norm{\rho - \sigma}_{\HS} > \epsilon$, with probability $1 - \delta$}
		$N\gets O(\sqrt{d}/\epsilon^2)$.\;
		\For{$T = 1,\ldots,O(\log(1/\delta))$}{
			Sample a Haar-random unitary matrix $\U$.\;
			Form the POVM $\calM$ consisting of $\brc{\ket{\U_1}\bra{\U_1},\ldots,\ket{\U_d}\bra{\U_d}}$.\;
			Measure each copy of $\rho$ with $\calM$, yielding outcomes $z_1,\ldots,z_N$.\;
			Let $q\in\Delta^d$ denote the distribution over outcomes from measuring $\sigma$ with $\calM$.\;
			Draw i.i.d. samples $z'_1,\ldots,z'_N$ from $q$. \;
			$b_i \gets$\textsc{L2Tester}($\brc{z_i}, \brc{z'_i}$).\;
		}
		\Return majority among $b_1,\ldots,b_T$.\;
\end{algorithm}

To prove Lemma~\ref{lem:basic_id_tester}, we will need the following result from classical distribution testing.

\begin{lemma}[Lemma 2.3 from \cite{diakonikolas2016new}]\label{lem:DK}
	Let $p,q$ be two unknown distributions on $[d]$ for which $\Min{\norm{p}_2}{\norm{q}_{2}} \le b$ for some $b>0$. There exists an algorithm {\sc L2Tester} that takes $N = O(b\log(1/\delta)/\epsilon^2)$ samples from each of $p$ and $q$ and distinguishes between $p = q$ and $\norm{p - q}_2 > \epsilon$ with probability at least $1 - \delta$.\footnote{Note that Lemma 2.3 in \cite{diakonikolas2016new} only gives a constant probability guarantee, but the version we state follows by a standard amplification argument.}
\end{lemma}

We will also need the following moment calculations:

\begin{lemma}\label{lem:moments_basic}
	For any Hermitian $\M\in\co^{d\times d}$ and Haar-random $\U\in U(d)$, let $Z$ denote the random variable $\sum^d_{i=1}\left(\U^{\dagger}_i \M \U_i\right)^2$. Then \begin{equation}
		\E{Z} = \frac{1}{d+1}\left(\Tr(\M)^2 + \norm{\M}^2_{\HS}\right).
	\end{equation} If in addition we have that $\Tr(\M) = 0$, then \begin{equation}
		\E{Z^2} \le \frac{1 + o(1)}{d^2}\norm{\M}^4_{\HS}.
	\end{equation}
\end{lemma}

\begin{proof}
	By symmetry $\E{Z} = d\E{(\U_1\M\U_1)^2}$, and by Lemma~\ref{lem:collins}, if $\Pi$ denotes the projector to the first coordinate, \begin{equation}
		\E{(\U_1\M\U_1)^2} = \sum_{\pi,\tau\in S_2}\Wg(\pi\tau^{-1},d)\iprod{\Pi}_{\pi}\iprod{\M}_{\tau} = \frac{1}{d(d+1)}(\Tr(\M)^2 + \Tr(\M^2)),
	\end{equation} from which the first part of the lemma follows.

	For the second part, let $\calS^*_4\subset S_4$ denote the set of permutations $\pi$ for which $\pi(1),\pi(2) \in \brc{1,2}$ and $\pi(3),\pi(4)\in \brc{3,4}$. Note that
	\begin{equation}
		\E{Z^2} = d\cdot\E*{(\U_1^{\dagger}\M\U_1)^4} + (d^2 - d)\cdot \E*{(\U^{\dagger}_1 \M\U_1)^2(\U^{\dagger}_2\M\U_2)^2}.\label{eq:Z2}
	\end{equation}
	For the first term, by Lemma~\ref{lem:collins} we have \begin{align}
		\E{(\U^{\dagger}_1 \M\U_1)^4} &= \sum_{\pi,\tau\in S_4} \Wg(\pi\tau^{-1},d)\iprod{\M}_{\tau} \\
		&= \frac{1}{d(d+1)(d+2)(d+3)}\sum_{\tau}\iprod{\M}_{\tau} \\
		&= \frac{1}{d(d+1)(d+2)(d+3)}\sum_{\tau \ \text{derangement}}\iprod{\M}_{\tau} \\
		&\le \frac{O(\norm{\M}_{\HS}^4)}{d(d+1)(d+2)(d+3)},
	\end{align} where the third step follows by the fact that $\Tr(\M) = 0$, and the fourth by the fact that for any derangement $\tau\in S_4$, either $\iprod{\M}_{\tau} = \Tr(\M^2)^2 = \norm{\M}_{\HS}^4$, or $\iprod{\M}_{\tau} = \Tr(\M^4)\le\norm{\M}_{\HS}^4$. Similarly, \begin{align}
		\E{(\U^{\dagger}_1 \M\U_1)^2(\U^{\dagger}_2\M\U_2)^2} &= \sum_{\pi\in\calS^*_4,\tau\in S_4} \Wg(\pi\tau^{-1},d)\iprod{\M}_{\tau} \\
		&= \sum_{\tau\in\calS^*_4}\Wg(e,d)\iprod{\M}_{\tau} + \sum_{\pi\in\calS^*_4, \tau\in S_4: \tau\neq \pi} \Wg(\pi\tau^{-1},d)\iprod{\M}_{\tau} \\
		&= \Wg(e,d)\norm{\M}_{\HS}^4 + \sum_{\pi\in\calS^*_4, \tau\in S_4: \tau\neq \pi} \Wg(\pi\tau^{-1},d)\iprod{\M}_{\tau} \\
		&\le \frac{d^4 - 8 d^2 + 6}{d^2 (d^6 - 14 d^4 + 49 d^2 - 36)}\norm{\M}_{\HS}^4 + O(1/d^5)\cdot \norm{\M}_{\HS}^4 \\
		&= \frac{1 + o(1)}{d^4}\norm{\M}_{\HS}^4,
	\end{align} where in the second step $\Wg(e,d)$ denotes the Weingarten function corresponding to the identity permutation, in the third step we used the fact that the only $\tau\in\calS^*_4$ which is a derangement is the permutation that interchanges 1 with 2, and 3 with 4, and in the fourth step we used the form of $\Wg(e,d)$, the fact that $\abs{\Wg(\pi\tau^{-1},d)} = O(1/d^5)$ for $\pi\neq \tau$, and the fact that $\iprod{\M}_{\tau} \le \norm{\M}^4_{\HS}$. The second part of the lemma follows from \eqref{eq:Z2}.
\end{proof}

We can now complete the proof of Lemma~\ref{lem:basic_id_tester}.

\begin{proof}[Proof of Lemma~\ref{lem:basic_id_tester}]
	Let $p$ and $q$ be the distribution over $d$ outcomes when measuring $\rho$ and $\sigma$ respectively using the POVM defined in a single iteration of the main loop of {\sc BasicCertify}. Applying both parts of Lemma~\ref{lem:moments_basic} to $\M = \rho - \sigma$, for which the random variable $Z$ is $\norm{p - q}^2_2$, we conclude that for some sufficiently small absolute constant $c > 0$, $\Pr{\norm{p - q}_2\ge c\norm{\M}_{\HS}/\sqrt{d}} \ge 5/6$. Applying the first part of Lemma~\ref{lem:moments_basic} to $\M = \rho$ and $\M = \sigma$, for which the random variable $Z$ is $\norm{p}^2_2$ and $\norm{q}^2_2$ respectively, we have that $\E{\norm{p}^2_2}, \E{\norm{q}^2_2} \le 2/d$, so by Markov's, for some absolute constant $c' > 0$, $\norm{p},\norm{q}_2 \le c'/\sqrt{d}$ with probability at least $5/6$. We can substitute these bounds for $\norm{p}_2,\norm{q}_2,\norm{p - q}_2$ into Lemma~\ref{lem:DK} to conclude that the output of \textsc{L2Tester} is correct with some constant advantage. Repeating this $O(\log(1/\delta))$ times and taking the majority among all the outputs from \textsc{L2Tester} gives the desired high-probability guarantee.
\end{proof}

\subsection{Bucketing and Mass Removal}
\label{subsec:bucket_upper}

We may without loss of generality assume that $\sigma$ is the diagonal matrix $\diag(\lambda_1,\ldots,\lambda_d)$, where $\lambda_1 \le \cdots \le \lambda_d$.

We will use the bucketing procedure outlined in Section~\ref{sec:bucketing_nonadaptive}. The way that we remove a small amount of mass from the spectrum of $\sigma$ slightly differs from that outlined in Definition~\ref{defn:remove_nonadaptive} for our lower bound. Our bucketing and mass removal procedure is as follows:

\begin{definition}[Removing low-probability elements- upper bound]\label{defn:remove_nonadaptive_upper}
	Let $d'\le d$ denote the largest index for which $\sum^{d'}_{i=1}\lambda'_i \le \epsilon^2/20$,\footnote{We made no effort to optimize this constant factor.} and let $\Sjunk\triangleq [d']$. Let $\sigma'$ denote the matrix given by zeroing out the diagonal entries of $\sigma$ indexed by $\Sjunk$. For $j\in\Z_{\ge 0}$, let $S_j$ denote the indices $i\not\in\Sjunk$ for which $\lambda_i\in[2^{-j-1},2^{-j}]$, and denote $|S_j|$ by $d_j$. Let $\calJ$ denote the set of $j$ for which $S_j \neq \emptyset$.
\end{definition}

As in the proofs of our lower bounds, we use the following basic consequence of bucketing:

\begin{fact}
	There are at most $\log(10d/\epsilon^2)$ indices $j\in\calJ$.
\end{fact}

\begin{proof}
	The largest element among $\brc{\lambda_i}_{i\in\Sjunk}$ is at least $\epsilon^2/10d$, from which the claim follows.
\end{proof}

We now introduce some notation. Let $m\triangleq \log(10d/\epsilon^2)$ denote this upper bound on the number of buckets in $\calJ$. For $j\in\calJ$, let $\rho[j,j], \sigma[j,j]\in\co^{d\times d}$ denote the Hermitian matrices given by zeroing out entries of $\rho, \sigma$ outside of the principal submatrix indexed by $S_j$. For distinct $j,j'\in\calJ$, let $\rho[j,j']\in\co^{d\times d}$ denote the Hermitian matrix given by zeroing out entries of $\rho$ outside of the two non-principal submatrices with rows and columns indexed by $S_i$ and $S_j$, and by $S_j$ and $S_i$. Lastly, let $\wh{\rho}[j,j], \wh{\sigma}[j,j], \wh{\rho}[j,j'], \wh{\sigma}[j,j']$ denote these same matrices but with trace normalized to 1.

Let $\rhojunk^{\mathsf{diag}}\in\co^{d\times d}$ be the principal submatrix of $\rho$ indexed by $\Sjunk$, and let $\rhojunk^{\mathsf{off}}\in\co^{d\times d}$ be the matrix given by zeroing out the principal submatrices indexed by $\Sjunk$ and by $[d]\backslash\Sjunk$.

Lastly, we will need the following basic fact:

\begin{fact}\label{fact:conditioning}
	Given two psd matrices $\rho,\sigma$, if $\abs{\Tr(\rho) - \Tr(\sigma)} \le \epsilon/2$ and $\norm{\rho - \sigma} \ge \epsilon$, then \begin{equation}
		\norm*{\rho/\Tr(\rho) - \sigma/\Tr(\sigma)}_1 \ge \epsilon/2\Tr(\rho).
	\end{equation}
\end{fact}

\begin{proof}
	Note that \begin{equation}
		\norm*{\sigma/\Tr(\rho) - \sigma/\Tr(\sigma)}_1 = \abs*{\frac{\Tr(\sigma)}{\Tr(\rho)} - 1} \le \frac{\epsilon}{2\Tr(\rho)},
	\end{equation} so by triangle inequality, \begin{equation}
		\norm*{\rho/\Tr(\rho) - \sigma/\Tr(\sigma)}_1 \ge \frac{1}{\Tr(\rho)}\norm*{\rho - \sigma}_1 - \norm*{\sigma/\Tr(\rho) - \sigma/\Tr(\sigma)}_1 \ge \frac{\epsilon}{2\Tr(\rho)}.
	\end{equation}
\end{proof}

\subsection{Instance-Near-Optimal Certification}
\label{subsec:full_cert}

We are ready to prove Theorem~\ref{thm:cert_upper_main}.

\begin{proof}[Proof of Theorem~\ref{thm:cert_upper_main}]


	We have that \begin{equation}
		\rho = \sum_{j\in\calJ} \rho[j,j] + \sum_{j\in\calJ: j\neq j'} \rho[j,j'] + \rhojunk^{\mathsf{diag}} + \rhojunk^{\mathsf{off}} \qquad \sigma' = \sum_{j\in\calJ} \sigma[j,j]
	\end{equation}
	If $\norm{\rho - \sigma}_1 > \epsilon$, then by triangle inequality, \begin{equation}
		\norm*{\sum_{j\in\calJ} (\rho[j,j] - \sigma[j,j]) + \sum_{j,j'\in\calJ: j\neq j'}\rho[j,j'] + \rhojunk^{\mathsf{diag}} + \rhojunk^{\mathsf{off}}}_1 = \norm{\rho - \sigma'}_1 \ge \epsilon - \epsilon^2/20 \ge 9\epsilon/10
	\end{equation} and one of four things can happen:

	\begin{enumerate}
	 	\item $\norm{\rhojunk^{\mathsf{diag}}}_1 \ge \epsilon^2/8$.
	 	\item $\norm{\rhojunk^{\mathsf{off}}}_1 \ge \epsilon/2$,
	 	\item There exists $j\in\calJ$ for which $\norm{\rho[j,j] - \sigma[j,j]}_1 \ge \epsilon/(10m^2)$
	 	\item There exist distinct $j,j'\in\calJ$ for which $\norm{\rho[j,j']}_1 \ge \epsilon/(5m^2)$.
	\end{enumerate}
	Otherwise we would have \begin{equation}
		\norm{\rho - \sigma'}_1 \le m\cdot\frac{\epsilon}{10m^2} + \binom{m}{2}\cdot\frac{\epsilon}{5m^2} + \frac{\epsilon^2}{8} + \frac{\epsilon}{2} = \frac{\epsilon}{10m} + \frac{\epsilon(m-1)}{10m} + \frac{3\epsilon}{4} < 9\epsilon/10,
	\end{equation} a contradiction.

	It remains to demonstrate how to test whether we are in any of Scenarios 1 to 4.

	\begin{lemma}\label{lem:scenario1}
		$O(\log(1/\delta)/\epsilon^2)$ copies suffice to test whether $\rho = \sigma$ or whether Scenario 1 holds, with probability $1 - O(\delta)$.
	\end{lemma}

	\begin{proof}
		We can use the POVM consisting of the projector $\Pi$ to the principal submatrix indexed by $\Sjunk$, together with $\Id - \Pi$, to distinguish between whether $\Tr(\rhojunk^{\mathsf{diag}}) \ge \epsilon^2/8$ or whether $\Tr(\rhojunk^{\mathsf{diag}}) \le \epsilon^2/10$, the latter of which holds if $\rho = \sigma$ by definition of $\Sjunk$. For this distinguishing task, $O(\log(1/\delta)/\epsilon^2)$ copies suffice.
	\end{proof}

	\begin{lemma}\label{lem:scenario2}
		If Scenario 1 does not hold, then Scenario 2 cannot hold.
	\end{lemma}

	\begin{proof}
		Suppose Scenario 1 does not hold so that $\norm{\rhojunk^{\mathsf{diag}}}_1 < \epsilon^2/4$. Then by the first part of Lemma~\ref{lem:tracepsd}, $\norm{\rhojunk^{\mathsf{off}}}^2_1 < (1 - \epsilon^2/4)\cdot \epsilon^2/4 < \epsilon^2/4$, a contradiction.
	\end{proof}

	\begin{lemma}\label{lem:scenario3}
		$O\left(\norm{\sigma'}_{2/5}\polylog(d/\epsilon)\log(m/\delta)/\epsilon^2\right)$ copies suffice to test whether $\rho = \sigma$ or whether Scenario 3 holds, with probability $1 - O(\delta)$.
	\end{lemma}

	\begin{proof}
		If $\Tr(\sigma[j,j]) < \epsilon/(10m^2)$, then to test whether $\rho = \sigma$ or Scenario 3 holds, it suffices to decide whether $\Tr(\rho[j,j]) \ge \Tr(\sigma[j,j]) + \epsilon/(10m^2)$. We can do this by measuring $\rho$ using the POVM consisting of the projection $\Pi_j$ to the principal submatrix indexed by $S_j$, together with $\Id - \Pi_j$, for which $O(m^4\log^2(1/\delta)/\epsilon^2)$ copies suffice to determine this with probability $1 - O(\delta)$.

		Suppose now that $\Tr(\sigma[j,j]) \ge \epsilon/(10m^2)$. We can use $O(\log^4(d/\epsilon)\cdot\log(1/\delta)/\epsilon^2)$ copies to approximate $\Tr(\rho[j,j])$ to additive error $\epsilon/(40m^2)$ with probability $1 - O(\delta)$ using the same POVM.

		If our estimate for $\Tr(\rho[j,j])$ is greater than $\epsilon/(40m^2)$ away from $\Tr(\sigma[j,j])$, then $\rho\neq \sigma$.

		Otherwise, $\abs*{\Tr(\rho[j,j]) - \Tr(\sigma[j,j])} \le \epsilon/(20m^2)$. Then by Fact~\ref{fact:conditioning}, to determine whether we are in Scenario 3, it suffices to design a tester to distinguish whether the mixed states $\wh{\rho}[j,j]$ and $\wh{\sigma}[j,j]$ are equal or $\epsilon'$-far in trace distance for \begin{equation}
			\epsilon'\triangleq \frac{\epsilon}{20m^2\Tr(\sigma[j,j])} = \Theta\left(\frac{\epsilon}{20m^2 d_j 2^{-j}}\right).\label{eq:epsprime}
		\end{equation}
		Note that if $\wh{\rho}[j,j]$ and $\wh{\sigma}[j,j]$ are $\epsilon'$-far in trace distance, they are at least $\epsilon'/\sqrt{d_j}$-far in Hilbert-Schmidt. We conclude from Lemma~\ref{lem:basic_id_tester} that we can distinguish with probability $1 - O(\delta)$ between whether $\wh{\rho}[j,j]$ and $\wh{\sigma}[j,j]$ are equal or $\epsilon'$-far in trace distance using $O(d_j^{3/2}\log(1/\delta)/\epsilon'^2) = O(d^{7/2}_j 2^{-2j}\log^4(d/\epsilon)\log(1/\delta)/\epsilon^2)$ measurements on the conditional state $\wh{\rho}[j,j]$. Note that $\Tr(\rho[j,j]) \ge \Omega(\Tr(\sigma[j,j]))$ because $\Tr(\sigma[j,j]) \ge \epsilon/(10m^2)$ by assumption, so $\Tr(\sigma[j,j]) \ge \Omega(d_j2^{-j})$. As a result, this tester can make the desired number of measurements on the conditional state by using $O(d^{5/2}_j 2^{-j}\log^4(d/\epsilon)\log(1/\delta)/\epsilon^2)$ copies of $\rho$ and rejection sampling.

		By a union bound over distinct pairs $j,j'$, it therefore takes $O(\log(m/\delta))$ times \begin{equation}
			\sum_{j\in\calJ} O\left(d^{5/2}_j 2^{-j}\log^4(d/\epsilon)/\epsilon^2\right) \le \sum_{j\in\calJ} O\left(d^{5/2}_j \lambda_j\log^4(d/\epsilon)/\epsilon^2\right) \le O\left(\norm{\sigma'}_{2/5}\polylog(d/\epsilon)/\epsilon^2)\right),
		\end{equation} copies to test whether Scenario 3 holds, where the last step above follows by Fact~\ref{fact:optimize}.
	\end{proof}

	\begin{lemma}\label{lem:scenario4}
		If Scenario 3 does not hold, then $O\left(\sqrt{d - d'} \norm{\sigma'}_{1/2}\log(m/\delta)\polylog(d/\epsilon)/\epsilon^2\right)$ copies suffice to test whether $\rho = \sigma$ or whether Scenario 4 holds, with probability $1 - O(\delta)$.
	\end{lemma}

	\begin{proof}
		Fix any $j\neq j'\in\calJ$ and suppose without loss of generality that $d_j \ge d_{j'}$. Let $\rho^*$ and $\sigma^*$ denote the matrices obtained by zeroing out all entries of $\rho$ and $\sigma$ except those in the principal submatrix indexed by $S_j \cup S_{j'}$. Let $\wh{\rho}^*_{j,j'}$ and $\wh{\sigma}^*_{j,j'}$ denote these same matrices with trace normalized to 1. For brevity, we will freely omit subscripts.

		If $\Tr(\sigma^*) < \epsilon/(5m^2)$, then $\norm{\sigma[j,j']}_1 \le \epsilon/(10m^2)$ by the second part of Lemma~\ref{lem:tracepsd}. If Scenario 2 holds, then $\norm{\rho[j,j']}_1 \ge \epsilon/(5m^2)$, so by another application of the second part of Lemma~\ref{lem:tracepsd}, we would get that $\Tr(\rho^*) \ge 2\epsilon/(5m^2)$, contradicting the fact that Scenario 1 does not hold.

		Suppose now that $\Tr(\sigma^*) \ge \epsilon/(5m^2)$. As in the proof of Lemma~\ref{lem:scenario3}, we can use $O(\log^4(d/\epsilon)\cdot \log(1/\delta)/\epsilon^2)$ copies to approximate $\Tr(\rho^*)$ to within additive error $\epsilon/(20m^2)$ with probability $1 - O(\delta)$.

		If our estimate is greater than $\epsilon/(20m^2)$ away from $\Tr(\sigma[j,j])$ then we know that $\rho \neq \sigma$. 

		Otherwise, $\abs*{\Tr(\rho^*) - \Tr(\sigma^*)} \le \epsilon/(10m^2)$, and in particular $\Tr(\rho^*) \ge \Omega(\Tr(\sigma^*))$ as a result. If Scenario 3 holds but Scenario 4 does not, then $\norm{\rho^* - \sigma^*} \ge \epsilon/(5m^2)$. So by Fact~\ref{fact:conditioning}, to determine whether we are in Scenario 2, it suffices to design a tester to distinguish whether the mixed states $\wh{\rho}^*$ and $\wh{\sigma}^*$ are equal or $\epsilon''$-far in trace distance, where \begin{equation}
			\epsilon''\triangleq \frac{\epsilon}{10m^2\Tr(\sigma^*)} = \Theta\left(\frac{\epsilon}{10m^2}\cdot (d_j 2^{-j} + d_{j'}2^{-j'})^{-1}\right)\label{eq:epsprimeprime}
		\end{equation}
		Note that if $\rho^*$ and $\sigma^*$ are $\epsilon''$-far in trace distance, they are at least $\epsilon''/\sqrt{d_j}$-far in Hilbert-Schmidt, by the assumption that $d_j \ge d_{j'}$. We conclude from Lemma~\ref{lem:basic_id_tester} that we can distinguish these two cases using \begin{equation}
			O(\sqrt{d_j}d_{j'}\log(1/\delta)/\epsilon'^2) = O\left(\sqrt{d_j} d_{j'} (d_j 2^{-j} + d_{j'}2^{-j'})^2 \log^4(d/\epsilon)\log(1/\delta)/\epsilon^2\right)
		\end{equation} measurements on the conditional state $\wh{\rho}^*$. Because $\Tr(\rho^*) \ge \Omega(\Tr(\sigma^*)) \ge \Omega(d_j 2^{-j} + d_{j'} 2^{-j'})$, this tester can make the desired number of measurements on the conditional state by using $O\left(\sqrt{d_j} d_{j'} (d_j 2^{-j} + d_{j'}2^{-j'})\log^4(d/\epsilon)\log(1/\delta)/\epsilon^2\right)$ copies of $\rho$ and rejection sampling.

		Summing over $j\neq j'\in\calJ$ for which $d_j \ge d_{j'}$, we conclude that it takes $O(\log(1/\delta))$ times \begin{align}
			\sum_{j\neq j'\in\calJ: d_j \ge d_{j'}} \sqrt{d_j} d_{j'} (d_j 2^{-j} + d_{j'}2^{-j'}) &\le \sum_{j,j'\in\calJ: d_j \ge d_{j'}} d_j^{3/2} d_{j'} 2^{-j} + \sum_{j,j'\in\calJ: d_j \ge d_{j'}} \sqrt{d_j} d_{j'}^2 2^{-j'} \\
			&\le |\calJ|\cdot \sum_{j\in\calJ} d^{5/2}_j 2^{-j} + \left(\sum_{j\in\calJ}\sqrt{d_j}\right)\left(\sum_{j\in\calJ} d^2_j 2^{-j}\right) \\
			&\le \polylog(d/\epsilon) \cdot \left(\norm{\sigma'}_{2/5} + \sqrt{d - d'} \cdot \norm{\sigma'}_{1/2}\right),
		\end{align} copies to test whether Scenario 4 holds, where the last step above uses Fact~\ref{fact:optimize}.

		We claim that the above bound is dominated by $O(\log(m/\delta)\polylog(d/\epsilon))\sqrt{d - d'} \norm{\sigma'}_{1/2}$. Indeed, note that for any vector $v \in\R^m$, \begin{equation}
			\norm{v}^{2/5}_{2/5} = \sum_i v_i^{2/5} \le \left(\sum_i (v_i^{2/5})^{5/4}\right)^{4/5} \cdot \left(\sum_i 1^5 \right)^{1/5} \le \norm{v}^{2/5}_{1/2}\cdot \sqrt{m}^{2/5},
		\end{equation} as desired.
	\end{proof}

	Altogether, Lemmas~\ref{lem:scenario1} to \ref{lem:scenario4} allow us to conclude correctness of the algorithm {\sc Certify} whose pseudocode is provided in Algorithm~\ref{alg:fullcert} below. The copy complexity guarantee follows from these lemmas together with Fact~\ref{fact:fidelity}.
\end{proof}

\begin{remark}\label{remark:efficient}
    As stated, we are performing measurements in Haar-random bases at various points in {\sc Certify} and in particular the subroutine {\sc BasicCertify}. As Lemma~\ref{lem:moments_basic} and Lemma~\ref{lem:basic_id_tester} make clear however, we only exploit the first four moments of the Haar measure over the unitary group. As a result, if we were interested in implementing a gate-efficient protocol for state certification, we could have replaced the Haar measure with an approximate 4-design, for which there are a variety of gate-efficient constructions, e.g. \cite{haferkamp2020quantum}.
\end{remark}

\begin{algorithm}
\DontPrintSemicolon
\caption{\textsc{Certify}($\rho,\sigma,\epsilon,\delta$)}
\label{alg:fullcert}
	\KwIn{Copy access to $\rho$, diagonal density matrix $\sigma$, error $\epsilon$, failure probability $\delta$}
	\KwOut{$\mathsf{YES}$ if $\rho = \sigma$, $\mathsf{NO}$ if $\norm{\rho - \sigma}_{\HS} > \epsilon$, with probability $1 - \delta$.}
		$m\gets \log(10d/\epsilon^2)$.\;
		Let $\Pi$ be the projector to the principal submatrix indexed by $\Sjunk$. \tcp*{Scenario 1}
		$\calM\gets \brc{\Pi, \Id - \Pi}$.\;
		Measure $O(\log(1/\delta)/\epsilon^2)$ copies of $\rho$ with the POVM $\calM$.\;
		\If{$\ge (\epsilon^2/5)$ fraction of outcomes observed correspond to $\Pi$}{
			\Return $\mathsf{NO}$.\;
		}
		\For(\tcp*[f]{Scenario 3}){$j\in\calJ$}{
			Let $\Pi_j$ denote the projection to the principal submatrix indexed by $S_j$.\;
			$\calM_j\gets \brc{\Pi_j, \Id - \Pi_j}$.\;
			Measure $O(\polylog(d/\epsilon)\log(1/\delta)/\epsilon^2)$ copies of $\rho$ with the POVM $\calM_j$.\;
			\If{$\ge (\Tr(\sigma[j,j]) + \epsilon/(40m^2))$ fraction of outcomes observed correspond to $\Pi_j$}{
				\Return $\mathsf{NO}$.\;
			}\Else{
				Define $\epsilon'$ according to \eqref{eq:epsprime}.\;
				$b_j\gets${\sc BasicCertify}($\wh{\rho}[j,j],\wh{\sigma}[j,j],\epsilon',O(\delta/m)$).\;
				\If{$b_j = \mathsf{NO}$}{
					\Return $\mathsf{NO}$.\;
				}
			}
		}
		\For(\tcp*[f]{Scenario 4}){$j,j'\in\calJ$ distinct and satisfying $d_j \ge d_{j'}$}{	
			Let $\Pi_{j,j'}$ denote the projection to the principal submatrix indexed by $S_j\cup S_{j'}$.\;
			$\calM_{j,j'}\gets \brc{\Pi_{j,j'}, \Id - \Pi_{j,j'}}$.\;
			Measure $O(\polylog(d/\epsilon)\log(1/\delta)/\epsilon^2)$ copies of $\rho$ with the POVM $\calM_{j,j'}$.\;
			\If{$\ge (\Tr(\sigma^*_{j,j'}) + \epsilon/(20m^2))$ fraction of outcomes observed correspond to $\Pi_{j,j'}$}{
				\Return $\mathsf{NO}$.\;
			}\Else{
				Define $\epsilon''$ according to \eqref{eq:epsprimeprime}.\;
				$b_{j,j'}\gets${\sc BasicCertify}($\wh{\rho}^*_{j,j'},\wh{\sigma}^*_{j,j'},\epsilon'',O(\delta/m^2)$).\;
				\If{$b_{j,j'} = \mathsf{NO}$}{
					\Return $\mathsf{NO}$.\;
				}
			}
		}
		\Return $\mathsf{YES}$.
\end{algorithm}

\paragraph{Acknowledgments}

The authors would like to thank Robin Kothari for helpful discussions at an early stage of this work, as well as Hsin-Yuan Huang for suggesting the approach of lower bounding the likelihood ratio. Part of this work was completed while SC and JL were visiting the Simons Institute for the Theory of Computing.

\bibliographystyle{alpha}
\bibliography{biblio}

\appendix


\section{Adaptive Lower Bound}
\label{subsec:adaptive_instance}

In this section we prove a lower bound against state certification algorithms that use adaptive, unentangled measurements.

\begin{theorem}\label{thm:cert_adaptive_main}
		There is an absolute constant $c > 0$ for which the following holds for any $0 < \epsilon < c$.\footnote{As presented, our analysis yields $c$ within the vicinity of $1/3$, but we made no attempt to optimize for this constant.} Let $\sigma\in\co^{d\times d}$ be a diagonal density matrix. There is a matrix $\sigma^*$ given by zeroing out the largest entry of $\sigma$ and at most $O(\epsilon\log(d/\epsilon))$ additional mass from $\sigma$ (see Definition~\ref{defn:remove_adaptive} below), such that the following holds:

		Any algorithm for state certification to error $\epsilon$ with respect to $\sigma$ using adaptive, unentangled measurements has copy complexity at least \begin{equation}
		\Omega\left(d\cdot d^{1/3}_{\mathsf{eff}}\cdot F(\wh{\sigma}^{*},\rhomm)/(\epsilon^2\log(d/\epsilon))\right).
	\end{equation}
\end{theorem}

The outline follows that of Section~\ref{sec:instance}. In Section~\ref{sec:bucketing_adaptive}, we describe the procedure by which we remove mass from $\sigma$, which will be more aggressive than the one used for our nonadaptive lower bound. As a result, it will suffice to analyze the lower bound instance given in Section~\ref{subsec:instance_instance_2}, which we do in Section~\ref{subsec:instance_instance_adaptive_2}. For our analysis, we need to check some additional conditions hold for the adaptive lower bound framework of Section~\ref{subsec:adaptive_generic} to apply.

\subsection{Bucketing and Mass Removal}
\label{sec:bucketing_adaptive}

Define $\brc{S_j}, \calJ, \Ssing, \Smany$ in the same way as in Section~\ref{sec:bucketing_nonadaptive}. The way in which we remove mass from $\sigma$ will be more aggressive than in the nonadaptive setting. We will end up removing up to $O(\epsilon\log(d/\epsilon))$ mass (see Fact~\ref{fact:fewbuckets2}) as follows:

\begin{definition}[Removing low-probability elements- adaptive lower bound]\label{defn:remove_adaptive}
	Without loss of generality, suppose that $\lambda_1,\ldots,\lambda_d$ are sorted in ascending order according to $\lambda_i$. Let $d'\le d$ denote the largest index for which $\sum^{d'}_{i=1}{\lambda'}_i \le 4\epsilon$. Let $\Sjunk\triangleq [d']$.



	Let $\sigma^*$ denote the matrix given by zeroing out the largest entry of $\sigma$ and the entries indexed by $\Sjunk$. It will be convenient to define $\calJ^*$ to be the buckets for the nonzero entries of $\sigma^*$, i.e. the set of $j\in\calJ$ for which $S_j$ has nonempty intersection with $[d]\backslash\Sjunk$.
\end{definition}

\begin{fact}\label{fact:fewbuckets2}
	There are at most $O(\log(d/\epsilon))$ indices $j\in\calJ^*$. As a consequence, $\Tr(\sigma^*) \ge 1 - O(\epsilon\log(d/\epsilon))$.
\end{fact}

\begin{proof}
	For any $i_1\not\in\Sjunk$ and $i_2\in\Sjunk$, we have that $p_{i_1} > p_{i_2}$. In particular, summing over $i_2\in\Sjunk$, we conclude that $p_{i_1}\cdot |\Sjunk| > 4\epsilon$, so $p_{i_1} > 4\epsilon/d$. By construction of the buckets $S_j$, the first part of the claim follows. As in the proof of Fact~\ref{fact:fewbuckets}, the second part of the claim follows by definition of $\Slight.$
\end{proof}

\subsection{Analyzing Lower Bound II}
\label{subsec:instance_instance_adaptive_2}

We will analyze the sub-problem defined in Section~\ref{subsec:instance_instance_2} and prove the following lower bound:

\begin{lemma}\label{lem:adaptive_2}
	Fix any $j,j'\in\calJ^*$ satisfying $d_j \ge d_{j'}$. If $d_j > 1$, then we can optionally take $j = j'$. Suppose $\epsilon \le d_{j'}\cdot 2^{-j/2 - j'/2 - 1}$. Distinguishing between whether $\rho = \sigma$ or $\rho = \sigma_{\W}$ for $\W\in\co^{d_j\times d_{j'}}$ consisting of Haar-random orthonormal columns (see \eqref{eq:offdiag_1} and \eqref{eq:offdiag_2}), using adaptive unentangled measurements, has copy complexity at least \begin{equation}
		\Omega\left(\frac{d_j^{1/3} \cdot d_{j'}^2 \cdot 2^{-j'}}{\epsilon^2}\right).
	\end{equation}
\end{lemma}


\begin{proof}
	As in Section~\ref{subsec:instance_instance_2}, we will focus on the case where $j \neq j'$, but at the cost of some factors of two, the following arguments easily extend to the construction for $j = j'$ when $d_j > 1$ by replacing $S_j$ and $S_{j'}$ with $S^1_j, S^2_j$ defined immediately before \eqref{eq:offdiag_2}.

	We have already verified in Section~\ref{subsec:instance_instance_2} that Conditions~\ref{cond:gexp}, \ref{cond:secondmoment}, and \eqref{cond:normequiv} of Assumption~\ref{assume:main} are satisfied by $\calP$ for $L, \varsigma = O\left(\frac{\epsilon}{d_{j'} 2^{-j'/2}}\right)$. 


	It remains to check that $|g^{\U}_{\calP}(z)| \le 0.99$ for all $z$. To this end, recall \eqref{eq:gcorner}. As the diagonal entries of $\rho$ indexed by $S_j$ (resp. $S_{j'}$) are at least $2^{-j-1}$ (resp. $2^{-j'-1}$), \begin{equation}
		v^{\dagger}_z \rho v_z \ge 2^{-j-1}\norm{v^j_z}^2 + 2^{-j'-1}\norm{v^{j'}_z}^2 \ge 2^{-j/2-j'/2}\norm{v^j_z}\norm{v^{j'}_z},
	\end{equation} so \begin{equation}
		g^{\U}_{\calP}(z) \le \frac{\epsilon}{d_{j'}}\cdot \frac{\norm{v^j_z}\norm{v^{j'}_z}}{2^{-j/2-j'/2}\norm{v^j_z}\norm{v^{j'}_z}} \le \frac{\epsilon}{d_{j'} 2^{-j/2-j'/2}}.
	\end{equation} In particular, as long as $\epsilon \le d_{j'}2^{-j/2-j'/2-1}$, we have the bound $|g^{\U}_{\calP}(z)| \le 1/2$.



	We can now apply Theorem~\ref{thm:bcl} with $\tau = O\left(\frac{\epsilon^2}{d^{1/3}_j d^2_{j'}2^{-j'}}\right)$, noting that \begin{equation}
        \exp\left(-\Omega\left(\brc*{\Min{\frac{d_j\tau^2}{L^2\varsigma^2}}{\frac{d\tau}{L^2}}}\right)\right) = \exp\left(-\Omega\left(d_j^{1/3}\right)\right),
	\end{equation}
	to get that for any adaptive unentangled POVM schedule $\calS$, if $p^{\le N}_0$ is the distribution over outcomes from measuring $N$ copies of $\sigma$ with $\calS$ and $p^{\le N}_1$ is the distribution from measuring $N$ copies of $\sigma_{\U}$, then \begin{equation}
		\KL{p^{\le N}_1}{p^{\le N}_0} \le \frac{N\epsilon^2}{d^{1/3}_j d^2_{j'}2^{-j'}} + O(N)\cdot \exp\left(-\Omega\left(d^{1/3}_j - \frac{N \epsilon^2}{d^2_{j'}2^{-j'}}\right)\right).
	\end{equation} In particular, if $N = o\left(\frac{d^{1/3}_j d^2_{j'}2^{-j'}}{\epsilon^2\log(d/\epsilon)}\right)$, then $\KL{p^{\le N}_1}{p^{\le N}_0} = o(1)$ and we get the desired lower bound.
\end{proof}

\subsection{Putting Everything Together}

\begin{proof}[Proof of Theorem~\ref{thm:cert_adaptive_main}]
	As in the proof of Theorem~\ref{thm:cert_lower_main}, we proceed by casework depending on whether $d_j = 1$ for all $j\in\calJ^*$.

	\setcounter{case}{0}

	\begin{case}
		$d_j = 1$ for all $j\in\calJ^*$.
	\end{case}

	The analysis for this case in the nonadaptive setting completely carries over to this setting, because the lower bound from Lemma~\ref{lem:corner} holds even against adaptive POVM schedules. There are two possibilities. If there is a single bucket $j = j(i)$ for which $i\not\in\Sjunk$, then $d_{\mathsf{eff}} = 1$ and $\norm{\sigma^{*}}_{1/2} = O(1)$; for $\epsilon$ smaller than some absolute constant, we have that $\sigma_{i,i}\ge 3/4$ and Lemma~\ref{lem:corner} gives an $\Omega(1/\epsilon^2)$ lower bound as desired. Otherwise, let $j'$ be the smallest index for which $j' = j(i')$ for some $i'\in\calJ^*$, and let $j> j'$ be the next smallest index for which $j = j(i)$ for some $i\in\calJ^*$. Consider the lower bound instance in Section~\ref{subsec:instance_instance_adaptive_2} applied to this choice of $j,j'$. Provided that $\epsilon \le 2^{-j/2-j'/2-1}$, we would obtain a copy complexity lower bound of $\Omega(2^{-j'}/\epsilon^2) \ge \Omega(\norm{\sigma^*}_{1/2}/(\epsilon^2\log(d/\epsilon)))$, where the inequality is by Fact~\ref{fact:optimize}, and we would be done. On the other hand, if $\epsilon \ge 2^{-j/2 - j'/2-1}$, then because $2^{-j'} > 2^{-j}$, we would conclude that $2^{-j} \le 2\epsilon$. In particular, this implies that $\sum_{j''\in\calJ^*, i\in S_{j''}: j''\neq j'} \lambda_i \le 4\epsilon$, contradicting the fact that we have removed all buckets of total mass at most $4\epsilon$ in defining $\Sjunk$.

	\begin{case}
		$d_j > 1$ for some $j\in\calJ^*$.
	\end{case}

	Let $j_*\triangleq \arg\max_{j\in\calJ^*}d_j$ and $j'_*\triangleq \arg\max_{j\in\calJ^*} d^2_j 2^{-j}$. By Lemma~\ref{lem:nonadaptive_2}, as long as $\epsilon$ satisfies the bound \begin{equation}
		\epsilon \le d_{j'_*}\cdot 2^{-j_*/2 - j'_*/2 - 1}, \label{eq:epsbound2}
	\end{equation} we have a lower bound of \begin{equation}
		\Omega\left(d^{1/3}_{j_*} \cdot d_{j'_*}^2 \cdot 2^{-j'_*}/\epsilon^2\right) \ge \Omega\left(d\cdot d_{\mathsf{eff}}^{1/3}\cdot F(\sigma^*,\rhomm)/(\epsilon^2\log(d/\epsilon))\right),
	\end{equation} 
	where the second step follows by Fact~\ref{fact:optimize} and Fact~\ref{fact:fidelity}. Note that because $d_{j_*} > 1$ as we are in Case 2, we do not constrain $j_*,j'_*$ to be distinct necessarily.
	
	But under our assumptions on $j,j'$ and on $\calJ^*$, \eqref{eq:epsbound2} must hold: \begin{equation}
		d_{j'}2^{-j/2 - j'/2 - 1} \ge d_j 2^{-j-1} \ge \epsilon
	\end{equation} where the first step follows by the assumption that $j'\triangleq \arg\max_{j\in\calJ^*} d_j^2 2^{-j}$, and the second by the assumption that every bucket indexed by $\calJ^*$ has total mass at least $4\epsilon$.
\end{proof}


\section{Deferred Proofs}
\label{sec:defer}

\subsection{Proof of Lemma~\ref{lem:subexp}}

Fix an arbitrary single-copy subproblem $\calP = (\calM,\sigma,\brc{\sigma_{\U}}_{\U\sim\calD})$ for $\calD$ the Haar measure over $U(d)$. For any $\V\in U(d)$, define the functions $F_{\V}: U(d)\to\R$ and $G(\U)$ by \begin{equation}
	F_{\V}(\U)\triangleq \phi^{\U,\V}_{\calM} \qquad G(\U)\triangleq \E[z\sim p_0(\calM)]{g^{\U}_{\calP}(z)^2}^{1/2}.\label{eq:FVandGdef}
\end{equation}

We first show that Condition~\ref{cond:gexp} and \ref{cond:normequiv} from Assumption~\ref{assume:main} imply that $F_{\V}$ is mean zero and Lipschitz:

\begin{lemma}\label{lem:Flip}
	If $\calP$ satisfies Assumption~\ref{assume:main}, then for any $\V\in U(d)$, $F_{\V}$ is $G(\V)\cdot L$-Lipschitz and satisfies $\E[\U]{F_{\V}(\U)} = 0$.
\end{lemma}

\begin{proof}
	For any $\U,\U'\in U(d)$, we have that \begin{align}
		F_{\V}(\U) - F_{\V}(\U') &= \E[z\sim p_0(\calM)]{g^{\V}(z)\cdot(g^{\U}(z) - g^{\U'}(z))} \\
		&\le \E[z]{g^{\V}(z)^2}^{1/2}\cdot \E[z]{(g^{\U}(z) - g^{\U'}(z))^2}^{1/2} \le G(\V)\cdot L\cdot\norm{\U - \U'}_{\HS},
	\end{align} where the first inequality is by Cauchy-Schwarz, and the second is by Condition~\ref{cond:normequiv} of Assumption~\ref{assume:main}.
    
    The second part of the lemma immediately follows from Condition~\ref{cond:gexp} of Assumption~\ref{assume:main}.
\end{proof}

Next, we use Conditions~\ref{cond:secondmoment} and \ref{cond:normequiv} of Assumption~\ref{assume:main} to bound the expectation and Lipschitzness of $G$ which, combined with Theorem~\ref{thm:conc}, implies the following sub-Gaussian tail bound for $G$:

\begin{lemma}\label{lem:Gtail}
	If $\calP$ satisfies Assumption~\ref{assume:main}, then for any $s > 0$, \begin{equation}\Pr[\U]{G(\U) > \varsigma + s} \le \exp(-\Omega(ds^2/L^2)).\end{equation}
\end{lemma}

\begin{proof}
	The function $G$ is $L$-Lipschitz. To see this, note that for any $\U,\V\in U(d)$, 
	\begin{equation}
		G(\U) - G(\V) \le \E[z\sim p_0(\calM)]{(g^{\U}_{\calP}(z) - g^{\V}_{\calP}(z))^2}^{1/2} \le L\cdot \norm{\U - \V}_{\HS},
	\end{equation}
	where the first step is triangle inequality and the second is by Condition~\ref{cond:normequiv} of Assumption~\ref{assume:main}.

	By Condition~\ref{cond:secondmoment} and Jensen's, $\E{G(\U)} \le \E{g^{\U}(z)^2}^{1/2} \le \varsigma$. The claim then follows by Theorem~\ref{thm:conc}.
\end{proof}

We can finally prove Lemma~\ref{lem:subexp}:

\begin{proof}[Proof of Lemma~\ref{lem:subexp}]
	Note that $\E{\phi^{\U,\V}} = 0$ by the second part of Lemma~\ref{lem:Flip}. By the first of Lemma~\ref{lem:Flip} and Theorem~\ref{thm:conc}, \begin{equation}
		\Pr[\U]{\abs{\phi^{\U,\V}}>s} \le \exp\left(-\Omega\left(\frac{ds^2}{L^2G(\V)^2}\right)\right).\label{eq:phitail}
	\end{equation}
	We can apply Fact~\ref{fact:stieltjes} to the random variable $Y\triangleq G(\V)$ by taking the parameters as follows. Set $a \triangleq 2\varsigma$, $\tau(x)\triangleq \exp(-cd(x - \varsigma)^2/L^2)$, and $f(x)\triangleq \exp(-c'ds^2/L^2x^2)$ for appropriate constants $c,c'>0$. By \eqref{eq:phitail}, $\Pr[\U,\V]{\abs{\phi^{\U,\V}}>s} \le \E{f(Y)}$, and by Fact~\ref{fact:stieltjes} and Lemma~\ref{lem:Gtail},
	\begin{equation}
		\E{f(Y)} \le 2\exp\left(-\frac{c'ds^2}{L^2\varsigma^2}\right) + \int^{\infty}_{2\varsigma} \frac{2c'ds^2}{L^2x^3}\cdot \exp\left(-\frac{d}{L^2}\left(c(x - \varsigma)^2 + c's^2/x^2\right)\right) \, \d x
	\end{equation}
	Note that for $x\ge 2\varsigma$, by AM-GM, \begin{equation}
		c(x - \varsigma)^2 + c's^2/x^2 \ge \Omega(s(1 - \varsigma/x)) \ge \Omega(s),
	\end{equation} so we can bound \begin{equation}
		\E{f(Y)} \le 2\exp\left(-\frac{c'ds^2}{L^2\varsigma^2}\right) + \Omega\left(\frac{ds^2}{L^2\varsigma^2}\right)\cdot \exp(-\Omega(ds/L^2)) \le \exp\left(-\Omega\left(\Min{\frac{ds^2}{L^2\varsigma^2}}{\frac{ds}{L^2}}\right)\right)
	\end{equation} as claimed.
\end{proof}

\subsection{Proof of Theorem~\ref{thm:bcl}}
\label{app:bcl_defer}

Here we prove Theorem~\ref{thm:bcl} which gives an adaptive lower bound for distinguishing between a state $\sigma$ and a mixture of alternatives $\brc{\sigma_{\U}}_{\U\sim\calD}$ under Assumption~\ref{assume:main} when $\calD$ is the Haar measure over $U(d)$.

In this section, let $\calD$ denote the Haar measure over $U(d)$, and suppose that for any POVM $\calM$, the single-copy sub-problem $\calP = (\calM,\sigma,\brc{\sigma_{\U}}_{\U\sim\calD})$ satisfies Assumption~\ref{assume:main}.

\subsubsection{Additional Notation}

We first introduce some notation. Fix an unentangled, adaptive POVM schedule $\calS$. Given a transcript of measurement outcomes $z_{<t}$ up to time $t$, if $\calM^{z_{<t}}$ is the POVM used in time step $t$, then for convenience we will denote $g^{\U}_{\calP}$ and $\phi^{\U,\V}_{\calP}$ by $g^{\U}_{z_{<t}}$ and $\phi^{\U,\V}_{z_{<t}}$, $K^{\U,\V}_{z_{<t}}$.

Let $p^{\le t}_0$ (resp. $p^{\le t}_1$) denote the distribution over transcripts $z_{\le t}$ of outcomes up to and including time $t$ under measuring $\sigma$ (resp. $\sigma_{\U}$ for $\U\sim\calD$) with the first $t$ steps of $\calS$, and define the quantities
\begin{equation}
	\Delta(z_{\le t}) \triangleq \frac{\d p^{\le t}_1}{\d p^{\le t}_0}(z_{\le t}) \qquad \Psi^{\U,\V}_{z_{<t}} \triangleq \prod^{t-1}_{i=1} (1 + g^{\U}_{z_{<i}})(1 + g^{\V}_{z_{<i}}),
\end{equation}
where $\Delta(\cdot)$ is given by the Radon-Nikodym derivative.


\subsubsection{Helper Lemmas}

We will need the following helper lemmas. The first gives a \emph{lower bound} on the likelihood ratio between $p^{\le t}_1$ and $p^{\le t}_0$.

\begin{lemma}[Implicit in Lemma 6.2 of \cite{bubeck2020entanglement}]\label{lem:symimpliesDelta}
	Under the hypotheses of Theorem~\ref{thm:bcl}, for any transcript $z_{\le t}$,
    $\Delta(z_{\le t}) \ge \exp(-4\varsigma^2 t)$.
\end{lemma}

\begin{proof}
	By convexity of the exponential function and the fact that $1 + g^{\U}_{z_{<t}}(z_t) > 0$ for all $\U,t,z_t$, \begin{equation}
		\Delta(z_{<t}) \ge \prod^{t-1}_{i=1}\exp\left(\E[\U\sim\calD]{\ln(1 + g^{\U}_{z_{<i}}(z_i))}\right).
	\end{equation} For any $i < t$ we have that
	\begin{align}
		\exp\left(\E[\U\sim\calD]{\ln(1 + g^{\U}_{z_{<i}}(z_i))}\right) &\ge \exp\left(\E[\U]{g^{\U}_{z_{<i}}(z_i) - 4g^{\U}_{z_{<i}}(z_i)^2}\right) \\
		&\ge \exp\left(-4\varsigma^2\right),
	\end{align}
	where the first step follows by the elementary inequality $\ln(x) \ge x - 4x^2$ for all $x\in[-0.99,0.99]$ and the fact that $|g^{\U}_{z_{<t}}(z_t)| \le 0.99$ by hypothesis, and the second step follows by Conditions~\ref{cond:gexp} and \ref{cond:secondmoment} of Assumption~\ref{assume:main}.
\end{proof}

The next lemma gives a bound on the expectation of $(\Psi^{\U,\V}_{z_{<t}})^2$.

\begin{lemma}\label{lem:oneplusZimpliesPsi}
	Under the hypotheses of Theorem~\ref{thm:bcl}, $\E[z_{<t},\U,\V]*{(\Psi^{\U,\V}_{z_{<t}})^2} \le \exp(O(t\varsigma^2))$.
\end{lemma}

To prove this, it will be convenient to define the following for any $\ell$-copy sub-problem corresponding to POVM $\calM$ \begin{equation}
	K^{\U,\V}_{\calP} \triangleq \E[z\sim p_0(\calM)]*{\left(g^{\U}_{\calP}(z) + g^{\V}_{\calP}(z)\right)^2}
\end{equation} and first show the following:

\begin{lemma}\label{lem:tail_imply_oneplus}
	Under the hypothesis of Theorem~\ref{thm:bcl}, $\E[\U,\V]*{\left(1 + \gamma K^{\U,\V}_{\calP}\right)^t} \le \exp(O(\gamma t\varsigma^2))$ for any absolute constant $\gamma > 0$ and any $t = o(d/L^2)$.
\end{lemma}

\begin{proof}
	By the elementary inequality $(a+b)^2 \le 2a^2 + 2b^2$, we have that $K^{\U,\V}_{\calP} \le G(\U)^2 + G(\V)^2$. By Lemma~\ref{lem:Gtail}, we immediately get that $\Pr[\U]*{K^{\U,\V}_{\calP} > (\E{G(\U)} + s)^2} \le \exp(-ds^2/L^2)$. Applying the inequality again allows us to lower bound the left-hand side by $\Pr[\U]*{K^{\U,\V}_{\calP} > 2\E{G(\U)}^2 + 2s^2}$, so we conclude that \begin{equation}
		\Pr[\U]*{K^{\U,\V}_{\calP} > 2\E{G(\U)}^2 + s} \le \exp(-ds/2L^2).
	\end{equation} We can apply Fact~\ref{fact:stieltjes} to the random variable $Z\triangleq K^{\U,\V}_{\calP}$ and the function $f(Z) \triangleq (1 + \gamma Z)^t$ to conclude that \begin{align}
		\E[\U,\V]*{\left(1 + \gamma\cdot K^{\U,\V}_{\calP}\right)^t} &\le 2(1 + 2\gamma\E{G(\U)}^2)^t + \int^{\infty}_0 \gamma t(1 + \gamma x)^{t-1} \cdot e^{-x\cdot d/2L^2} \, \d x  \\
		&\le 2(1 + 2\gamma\E{G(\U)}^2)^t + \gamma t\int^{\infty}_0 e^{-x(d/2L^2 - \gamma(t-1))} \, \d x \\
		&\le 2(1 + 2\gamma\E{G(\U)}^2)^t + \frac{\gamma t}{d/2L^2 - \gamma(t-1)} \le \exp(O(t\gamma \E{G(\U)}^2)),
	\end{align} where in the last two steps we used that $t = o(d/L^2)$ to ensure that the integral is bounded and that the second term in the final expression is negligible.
\end{proof}

We can now prove Lemma~\ref{lem:oneplusZimpliesPsi}:

\begin{proof}[Proof of Lemma~\ref{lem:oneplusZimpliesPsi}]
	As $g^{\V}_{z_{<t-1}}(z) \le O(1)$, we know that for any constant $a,b \ge 2$, \begin{equation}
		\E[z\sim \Omega(\calM^{z_{<t-1}})]*{g^{\U}_{z_{<t-1}}(z)^a\cdot g^{\V}_{z_{<t-1}}(z_t)^b} \le \frac{1}{4}\E[z]{g^{\U}_{z_{<t-1}}(z)^2},	
	\end{equation}
	so we conclude that \begin{multline}
		\E[z\sim p_0(\calM^{z_{<t-1}})]*{(1 + g^{\U}_{z_{<t-1}}(z))^c(1 + g^{\V}_{z_{<t-1}}(z))^c} \\
		\le 1 + O_c\left(\E[z]{g^{\U}_{z_{<t-1}}(z)^2}\right) + O_c\left(\E[z]{g^{\V}_{z_{<t-1}}(z)^2}\right) + O_c\left(\phi^{\U,\V}_{z_{<t-1}}\right) \le 1 + C(c)\cdot K^{\U,\V}_{z_{<t-1}}\label{eq:gc}
	\end{multline} for some absolute constant $C(c)>0$, where the last step follows by AM-GM. For $\alpha_i \triangleq 2\cdot\left(\frac{t-1}{t-2}\right)^i$, we have that 
	\begin{align}
		\MoveEqLeft \E[z_{<t},\U,\V]*{\left(\Psi^{\U,\V}_{z_{<t}}\right)^{\alpha_i}} \\
		&\le \E[z_{<t-1},\U,\V]*{\left(\Psi^{\U,\V}_{z_{<t-1}}\right)^{\alpha_i} \cdot \left(1 + C(\alpha_i)\cdot K^{\U,\V}_{z_{<t-1}}\right)} \label{eq:apply_gc} \\
		&\le \E[z_{<t-1},\U,\V]*{\left(\Psi^{\U,\V}_{z_{<t-1}}\right)^{\alpha_i (t-1)/(t-2)}}^{(t-2)/(t-1)}\cdot \E[z_{<t-1},\U,\V]*{\left(1 + C(\alpha_i)\cdot K^{\U,\V}_{z_{<t-1}}\right)^{t-1}}^{1/(t-1)}\label{eq:holder} \\
		&\le \E[z_{<t-1},\U,\V]*{\left(\Psi^{\U,\V}_{z_{<t-1}}\right)^{\alpha_{i+1} (t-1)/(t-2)}}\cdot \E[z_{<t-1},\U,\V]*{\left(1 + C(\alpha_i)\cdot K^{\U,\V}_{z_{<t-1}}\right)^{t-1}}^{1/(t-1)}.
	\end{align} where \eqref{eq:apply_gc} follows by \eqref{eq:gc}, and \eqref{eq:holder} follows by Holder's. Unrolling this recurrence, we conclude that
	\begin{align}
		\E[z_{<t},\U,\V]*{\left(\Psi^{\U,\V}_{z_{<t}}\right)^2} &\le \prod^{t-1}_{i=1}\E[z_{<i},\U,\V]*{\left(1 + C(\alpha_{t-1-i})\cdot K^{\U,\V}_{z_{<i}}\right)^{t-1}}^{1/(t-1)}\label{eq:unroll} \\
		&\le \prod^{t-1}_{i=1}\E[z_{<i},\U,\V]*{\left(1 + C(2e)\cdot K^{\U,\V}_{z_{<i}}\right)^{t-1}}^{1/(t-1)},\label{eq:2e} \\
		&\le \sup_{\calM}\E[\U,\V]*{\left(1 + O(K^{\U,\V}_{\calM})\right)^{t-1}}
	\end{align} where \eqref{eq:2e} follows by the fact that for $1\le i \le t - 1$, $\alpha_{t-1-i} \le 2\left(1 + \frac{1}{t-2}\right)^{t-2} \le 2e$, and the supremum in the last step is over all POVMs $\calM$. The lemma then follows from Lemma~\ref{lem:tail_imply_oneplus}.
\end{proof}

\subsubsection{Putting Everything Together}

The key inequality used in \cite{bubeck2020entanglement} is the following consequence of the chain rule for KL:

\begin{lemma}[Lemma 6.1, \cite{bubeck2020entanglement}]\label{lem:chain}
	\begin{equation}
		\KL{p^{\le N}_1}{p^{\le N}_0} \le \sum^N_{t=1}Z_t \ \ \ \text{for} \ \ \ Z_t \triangleq \E[z_{<t}\sim p^{\le t - 1}_0]*{\frac{1}{\Delta(z_{<t})}\E[\U,\V]*{\Psi^{\U,\V}_{z_{<t}}\cdot \phi^{\U,\V}_{z_{<t}}}}.
	\end{equation}
\end{lemma}

We now have all the ingredients to complete the proof of Theorem~\ref{thm:bcl}.

\begin{proof}[Proof of Theorem~\ref{thm:bcl}]
	Given transcript $z_{<t}$ and $\U,\V\sim\calD$, let $\bone*{\calE^{\U,\V}_{z_{<t}}(\tau)}$ denote the indicator of whether $\abs*{\phi^{\U,\V}_{z_{<t}}} > \tau$; note that by Lemma~\ref{lem:subexp}, this event happens with probability at most $\xi(\tau)$, where \begin{equation}\xi(s)\triangleq \exp\left(-\Omega\left(\Min{\frac{ds^2}{L^2\varsigma^2}}{\frac{ds}{L^2}}\right)\right).\end{equation} We have that \begin{align}
		\E[\U,\V]*{\Psi_{z_{<t}}^{\U,\V}\cdot \phi_{z_{<t}}^{\U,\V}} &= \E[\U,\V]*{\Psi_{z_{<t}}^{\U,\V}\cdot \phi_{z_{<t}}^{\U,\V} \cdot \left(\bone{\calE_{z_{<t}}^{\U,\V}(\tau)} + \bone{\calE_{z_{<t}}^{\U,\V}(\tau)^c}\right)} \\
		&\le \E[\U,\V]*{\Psi_{z_{<t}}^{\U,\V}\cdot \bone{\calE_{z_{<t}}^{\U,\V}(\tau)}} + \tau\cdot \E[\U,\V]*{\Psi_{z_{<t}}^{\U,\V}\cdot \bone{\calE_{z_{<t}}^{\U,\V}(\tau)^c}} \\
		&\le \cdot \underbrace{\E[\U,\V]*{\Psi_{z_{<t}}^{\U,\V}\cdot \bone{\calE_{z_{<t}}^{\U,\V}(\tau)}}}_{\circled{B}_{z_{<t}}} + \tau\cdot \underbrace{\E[\U,\V]*{\Psi_{z_{<t}}^{\U,\V}}}_{\circled{G}_{z_{<t}}},
	\end{align} where in the second step we used the assumption that $|g^{\U}_{z_{<t}}(z_t)| \le 0.99$ for all $z_t$ to conclude that $\phi^{\U,\V}_{z_{<t}} \le 1$. Note that for any transcript $z_{<t}$, $\Delta(z_{<t})^2 = \E[\U,\V]{\Psi^{\U,\V}_{z_{<t}}} = \circled{G}_{z_{<t}}$, so by this and the fact that the likelihood ratio between two distributions always integrates to 1, \begin{equation}
		\E[z_{<t}\sim p^{\le t-1}_0]*{\frac{1}{\Delta^{(t-1)}(z_{<t})}\cdot \circled{G}_{z_{<t}}} = \E[z_{<t}\sim p^{\le t-1}_0]{\Delta^{(t-1)}(z_{<t})} = 1.\label{eq:one}
	\end{equation}

	Recalling the definition of $Z_t$ in Lemma~\ref{lem:chain}, we conclude that \begin{align}
		Z_t &\le  \E[z_{<t}\sim p^{\le t-1}_0]*{\frac{1}{\Delta^{(t-1)}(z_{<t})}\cdot \circled{B}_{z_{<t}}} + \tau \cdot \E[z_{<t}\sim p^{\le t-1}_0]*{\frac{1}{\Delta^{(t-1)}(z_{<t})}\cdot \circled{G}_{z_{<t}}} \\
		&\le \exp(4t\varsigma^2) \E[z_{<t}\sim p^{\le t-1}_0]*{\circled{B}_{z_{<t}}} + \tau,
	\end{align} where the second step follows by Lemma~\ref{lem:symimpliesDelta} and \eqref{eq:one}.

	To upper bound $\E[z_{<t}\sim p^{\le t-1}_0]{\circled{B}_{z_{<t}}}$, apply Cauchy-Schwarz to get \begin{align}
		\E[z_{<t}\sim p^{\le t-1}_0]*{\circled{B}_{z_{<t}}} &\le \E[z_{<t}\sim p^{\le t-1}_0,\U,\V]*{\left(\Psi^{\U,\V}_{z_{<t}}\right)^2}^{1/2}\cdot \Pr[z_{<t}\sim p^{\le t-1}_0,\U,\V]*{\calE_{z_{<t}}^{\U,\V}(\tau)}^{1/2} \\
		&\le \exp(O(t\varsigma^2))\cdot \xi(\tau),
	\end{align}
	where the second step follows by Lemma~\ref{lem:subexp} and Lemma~\ref{lem:oneplusZimpliesPsi}. Invoking Lemma~\ref{lem:chain} concludes the proof.
\end{proof}



\subsection{Proof of Fact~\ref{fact:sort_mix}}
\label{subsec:factsortmix_proof}

\begin{proof}
	We may assume $s < m+n$ (otherwise obviously $b = n$). Assume to the contrary that $\sum^{b+1}_{i=1}v_i d_i \le \epsilon$. We proceed by casework based on whether $w_{s'+1} = u_{a+1}$ or $w_{s'+1} = v_{b+1}$.

	If $w_{s'+1} = u_{a+1}$, then \begin{equation}
		3\epsilon < \sum^{s+1}_{i=1} w_i d^*_i = \sum^{a+1}_{i=1}u_i + \sum^b_{i=1}v_i d_i \le \sum^{a+1}_{i=1} v_{b+1}\cdot 2^{1-i} + \sum^b_{i=1} v_i \le 2\epsilon + \sum^b_{i=1}v_i d_i, \label{eq:seq1}
	\end{equation}
	where in the first step we used maximality of $s$, in the third step we used that $u_{a+1} \le v_{b+1}$ and that $u_{i+1} \ge 2 u_i$ for all $i$, and in the last step we used that $v_{b+1} \le \sum^{b+1}_{i=1}v_i d_i \le \epsilon$. From this we conclude that $\sum^b_{i=1}v_i d_i > \epsilon$, a contradiction.

	If $w_{s'+1} = v_{b+1}$, the argument is nearly identical. We have \begin{equation}
		3\epsilon < \sum^{s+1}_{i=1}w_i d^*_i = \sum^a_{i=1}u_i + \sum^{b+1}_{i=1}v_i d_i \le \sum^a_{i=1}v_{b+1}\cdot 2^{1 - i} + \sum^{b+1}_{i=1}v_i d_i \le 2\epsilon + \sum^{b+1}_{i=1}v_i,
	\end{equation} where in the first step we again used maximality of $s$, in the third step we used that $u_a \le v_{b+1}$ and $u_{i+1} \ge 2u_i$ for all $i$, and in the last step we used that $v_{b+1} \le \sum^{b+1}_{i=1}v_i d_i \le \epsilon$. From this we conclude that $\sum^b_{i=1} v_i d_i > \epsilon$, a contradiction.
\end{proof}

\end{document}